\documentclass[english,11pt]{article}
\usepackage[colorlinks=true,linkcolor=blue,citecolor=red,urlcolor=magenta]{hyperref}
\makeatletter
\usepackage{amsfonts,amssymb}
\usepackage{hyperref}
\usepackage{titlesec}
\usepackage{titletoc}
\usepackage{amsmath}
\usepackage{listings}
\usepackage{verbatim}
\usepackage{graphicx}
\usepackage{enumitem}
\usepackage{geometry} 
\usepackage{mathrsfs}
\usepackage{setspace}
\usepackage{cite}
\usepackage{bbm}
\usepackage{bbm, dsfont}
\usepackage{mathtools}
\usepackage{geometry} 
\usepackage{tabularx}
\usepackage{makecell} 
\usepackage{authblk}
\usepackage[nomath]{stix} 
\usepackage{amssymb,amsmath,amsfonts,latexsym}
\usepackage{amsmath,graphicx,bm,xcolor,url}
\usepackage[caption=false]{subfig} 
\usepackage{array}
\usepackage{verbatim}
\usepackage{bm}
\usepackage{verbatim}
\usepackage{textcomp}
\usepackage{mathrsfs}
\usepackage{relsize}
\usepackage{subfig}
 \usepackage{amsthm}

\catcode`~=11 \def\UrlSpecials{\do\~{\kern -.15em\lower .7ex\hbox{~}\kern .04em}} \catcode`~=13 

\allowdisplaybreaks[3]

\newcommand{\nn}{\nonumber}

\newcommand{\calE}{\mathcal{E}}

\newcommand{\calI}{\mathcal{I}}
\newcommand{\calJ}{\mathcal{J}}
\newcommand{\calK}{\mathcal{K}}

\newcommand{\calN}{\mathcal{N}}

\newcommand{\calT}{\mathcal{T}}
\newcommand{\calU}{\mathcal{U}}
\newcommand{\calV}{\mathcal{V}}

\newcommand{\calX}{\mathcal{X}}

\newcommand{\ba}{\mathbf{a}}
\newcommand{\bA}{\mathbf{A}}

\newcommand{\bB}{\mathbf{B}}

\newcommand{\bC}{\mathbf{C}}

\newcommand{\be}{\mathbf{e}}

\newcommand{\bg}{\mathbf{g}}

\newcommand{\bi}{\mathbf{i}}
\newcommand{\bI}{\mathbf{I}}

\newcommand{\bP}{\mathbf{P}}

\newcommand{\bR}{\mathbf{R}}

\newcommand{\bu}{\mathbf{u}}

\newcommand{\bv}{\mathbf{v}}

\newcommand{\bw}{\mathbf{w}}

\newcommand{\bx}{\mathbf{x}}
\newcommand{\bX}{\mathbf{X}}
\newcommand{\by}{\mathbf{y}}

\newcommand{\bz}{\mathbf{z}}

\newcommand{\bdelta}{\bm{\delta}}

\newcommand{\btau}{\bm{\tau}}

\newcommand{\bepsilon}{\bm{\epsilon}}

\newcommand{\bzeta}{\bm{\zeta}}

\newcommand{\bPhi}{\bm{\Phi}}

\DeclareMathOperator{\diag}{diag}

\DeclareMathOperator{\supp}{supp}

\newtheorem{theorem}{Theorem}

\newcommand{\qednew}{\nobreak \ifvmode \relax \else
      \ifdim\lastskip<1.5em \hskip-\lastskip
      \hskip1.5em plus0em minus0.5em \fi \nobreak
      \vrule height0.75em width0.5em depth0.25em\fi}

\newcommand{\scrH}{\mathscr{H}}

\newcommand{\scrN}{\mathscr{N}}

\numberwithin{equation}{section}
\numberwithin{theorem}{section} 
\usepackage{amsthm}
\usepackage{xcolor, tikz} 
\newtheorem{pro}{Proposition}
\numberwithin{pro}{section}
 
\newtheorem{coro}{Corollary}
\newtheorem{lem}{Lemma}
\newtheorem{claim}{Claim}
\numberwithin{lem}{section} 
\newtheorem{rem}{Remark}
\numberwithin{rem}{section}
\newtheorem*{theorem*}{Theorem (Informal)} 
\numberwithin{coro}{section} 
\usepackage{bm} 
\DeclareMathOperator{\sign}{sign}
 
\DeclareMathOperator{\rad}{rad}

\DeclareMathOperator{\rip}{RIP}

\newif\ifrevision
\revisionfalse   
\ifrevision
  \newcommand{\rev}[1]{\textcolor{teal!80!black}{#1}}
\else
  \newcommand{\rev}[1]{#1}
\fi

\newif\ifjrnotes
\jrnotesfalse   

\ifjrnotes
  \newcommand{\jrnote}[1]{%
    {\color{red}\footnotesize\textsf{[JR: #1]}}%
  }
\else
  \newcommand{\jrnote}[1]{}
\fi

\title{Robust Instance Optimal Phase-Only Compressed Sensing}
\author{Junren Chen\thanks{Department of Mathematics,  University of Maryland, College Park. ({\it email}: \url{jchen58@umd.edu})}~~\qquad Michael K. Ng\thanks{Department of Mathematics, Hong Kong Baptist Univeristy. ({\it email}: \url{michael-ng@hkbu.edu.hk})}~~\qquad Jonathan Scarlett\thanks{Department of Computer Science, National University of Singapore. ({\it email}: \url{scarlett@comp.nus.edu.sg})}} 
\date{\today}
\begin{document} 
 
    \maketitle

\long\def\symbolfootnote[#1]#2{\begingroup\def\thefootnote{\fnsymbol{footnote}}\footnote[#1]{#2}\endgroup}

\begin{abstract}
Phase-only compressed sensing (PO-CS) concerns the recovery of sparse signals from the phases of complex measurements. Recent results show that sparse signals in the standard sphere $\mathbb{S}^{n-1}$ can be exactly recovered from complex Gaussian phases by a linearization procedure, which  recasts PO-CS as linear compressed sensing  and then applies (quadratically constrained) basis pursuit to obtain $\bx^\sharp$. This paper focuses on the instance optimality and robustness of $\bx^{\sharp}$. First, we strengthen the nonuniform instance optimality of Jacques and Feuillen (2021) to a uniform one over the entire signal space. We show the existence of some universal constant $C$ such that 
$\|\bx^\sharp-\bx\|_2\le Cs^{-1/2}\sigma_{\ell_1}(\bx,\Sigma^n_s)$ holds for {\it all} $\bx$ in the unit Euclidean sphere, where $\sigma_{\ell_1}(\bx,\Sigma^n_s)$ is the $\ell_1$ distance of $\bx$ to its closest $s$-sparse signal. This is achieved by showing that the new sensing matrices corresponding to {\it all} approximately sparse signals simultaneously satisfy RIP. Second, we investigate the estimator's robustness to noise and corruption. We show that dense noise with entries bounded by some small  $\tau_0$,  appearing either {\it prior} or {\it posterior} to retaining the phases, increments $\|\bx^\sharp-\bx\|_2$ by $O(\tau_0)$. This is near-optimal (up to log factors) for any algorithm. On the other hand,   adversarial corruption, which changes an arbitrary $\zeta_0$-fraction of the measurements to any   phase-only values, increments $\|\bx^\sharp-\bx\|_2$ by $O(\sqrt{\zeta_0\log(1/\zeta_0)})$. We demonstrate the tightness  of this result via a partial analysis under suboptimal noise parameter and numerical evidence, while showing that the impact of sparse corruption can be eliminated: to this end, we propose an extended linearization approach that can   {\it exactly} recover $\bx$ from the corrupted phases. The  developments are then combined to yield a robust instance optimal guarantee that resembles the standard one in linear compressed sensing. 
\end{abstract}
\vspace{1mm}


\vspace{2mm}

\section{Introduction}
Compressed sensing  has proven to be an effective method in acquiring and reconstructing high-dimensional signals  \cite{Foucart2013,donoho2006compressed,candes2006robust,cai2013sparse,cohen2009compressed}. Mathematically, the goal of linear compressed sensing is to reconstruct sparse signals $\bx$ from a set of  measurements $\by =\bA\bx+\bepsilon$, under the sensing matrix $\bA \in \mathbb{R}^{m\times n}$ and noise vector $\bepsilon\in \mathbb{R}^m$. Restricted isometry property (RIP) lies at the center of linear compressed sensing theory, whose major finding is a set of  efficient algorithms achieving {\it instance optimality} under RIP sensing matrices \cite{Foucart2013,tropp2007signal,zhang2011sparse,blumensath2009iterative,dai2009subspace,cai2013sparse}. Here, the instance optimality describes the capacity of an algorithm to achieve estimation error proportional to the signal's distance to the cone of $s$-sparse vectors $\Sigma^n_s$. In the noiseless case, this translates into exact reconstruction of sparse signals and accurate estimate of approximately sparse signals.  As an example, if $\bA$ satisfies RIP over the cone of sparse vectors and $\varepsilon\ge\|\bepsilon\|_2$, then
basis pursuit  
\begin{align}\label{linearbp}
    \hat{\bx} =  \mathrm{arg}\min\|\bu\|_1,\quad \text{subject to }\|\bA\bu-\by\|_2\le\varepsilon
\end{align}
achieves 
\begin{align}\label{iop00}
    \|\hat{\bx}-\bx\|_2 \le C_1\frac{\sigma_{\ell_1}(\bx,\Sigma^n_s)}{\sqrt{s}} + C_2\varepsilon,\qquad\forall \bx\in \mathbb{R}^n
\end{align}
for some absolute constants $C_1,C_2$, where $\sigma_{\ell_1}(\bx,\Sigma^n_s):= \min_{\bu\in \Sigma^n_s}\|\bx-\bu\|_1$ denotes the $\ell_1$ distance of $\bx$ to $\Sigma^n_s.$ More details are given in Section \ref{sec:pre}. \rev{We write $\Delta(\bA;\by;\varepsilon)$ for the estimator $\hat{\bx}$ in (\ref{linearbp}) to emphasize the required inputs $(\bA,\by,\varepsilon).$}

The focus of the present paper is on the nonlinear compressed sensing model of phase-only compressed sensing (PO-CS), which concerns the reconstruction of  sparse signals in the standard sphere $\mathbb{S}^{n-1}=\{\bx\in \mathbb{R}^n:\|\bx\|_2=1\}$ from\footnote{\rev{We focus on real-valued signals in this paper, but note that our results extend to the complex-valued signals by incorporating arguments from \cite{chen2023uniform}; see \cite[P. 6740]{chen2023uniform} for a summary of the additional technicalities for complex-valued signals. In particular, PO-CS of complex-valued signals can be reformulated as a linear compressed sensing problem \cite[Eq.~(III.7)]{chen2023uniform}, and the analysis in \cite[Lem.~13, Thm. 1]{chen2023uniform} shows that the new sensing matrix satisfies the RIP over $\Sigma^n_{4s}$ with distortion arbitrarily close to $\frac{1}{3}$, which is sufficient for our purposes.}} 
\begin{align}\label{eq:po_mea}
    \bz=\sign(\bPhi\bx)=[\sign(\bPhi_1^*\bx),\cdots,\sign(\bPhi_m^*\bx)]^\top
\end{align}
 under  a  {\it complex} sensing matrix $ \bPhi = [\bPhi_1,\cdots,\bPhi_m]^*\in\mathbb{C}^{m\times n}$, with the phase function being given by $\sign(c)=\frac{c}{|c|}$ for $c\in \mathbb{C}\setminus \{0\}$ and $\sign(0)=1$ by convention. \rev{Let $\breve{\bz}\in \mathbb{C}^m$ be our observation vector, which equals $\bz=\sign(\bPhi\bx)$ in the noiseless case and is a perturbed version of  $\bz$   in the noisy case.}\footnote{\rev{Throughout this paper, $\bx$ denotes the underlying true signal, $\bPhi\in \mathbb{C}^{m\times n}$ denotes the complex sensing matrix,  $\bz$ is reserved for the noiseless observation vector $\bz=\sign(\bPhi\bx)$, and $\breve{\bz}$ denotes the actual observation vector that is a possibly noisy version of $\bz$. Therefore, $\bz$ and $\breve{\bz}$ implicitly depend on $\bx$.}} \rev{Since phase lies on the complex unit circle and is invariant under unknown scaling, phase-only sensing may be applicable to systems with severe dynamic-range limitations, as well as to settings involving extreme quantization or multiplicative noise.}  We note that phase-only reconstruction of unstructured signals was studied extensively several decades ago \cite{hayes1980signal,oppenheim1981importance,chen2023signal,hayes1982reconstruction,espy1983effects,li1983arrival,urieli1998optimal,loveimi2010objective}.

 PO-CS was initially considered as   
 a natural   extension of 1-bit compressed sensing\footnote{This concerns the recovery of   $\bx\in\Sigma^{n,*}_s$ from $\by=\sign(\bA\bx)$ with real sensing matrix $\bA\in \mathbb{R}^{m\times n}.$ \rev{Note that the phase-only measurements~(\ref{eq:po_mea}) take continuous values (on the complex unit circle) and are therefore more informative than 1-bit measurements $\by = \sign(\bA\bx) \in \{-1,1\}^m$. This distinction leads to a fundamental difference between the two problems.  In particular, while 1-bit compressed sensing admits a strictly positive information-theoretic lower bound on the reconstruction error~\cite{jacques2013robust}, exact recovery can often be achieved in phase-only compressed sensing.}}   \cite{boufounos2013angle,boufounos2013sparse} 
 and was recently revisited in \cite{feuillen2020ell,chen2023uniform,jacques2021importance,chen2025phase}. Exact reconstruction was observed experimentally in \cite{boufounos2013sparse} and theoretically proved very recently by Jacques and Feuillen 
 \cite{jacques2021importance} who proposed to recast PO-CS as a linear compressed sensing problem. We proceed to introduce this linearization approach under the noiseless setting $\breve{\bz}=\bz$. Since $\breve{\bz}=\sign(\bPhi\bx)$ implies that the entries of $\diag(\breve{\bz}^*)\bPhi \bx$ are non-negative real numbers, the phases give the linear measurements \begin{align}
     \frac{1}{\sqrt{m}}\Im(\diag(\breve{\bz}^*)\bPhi)\bx =0,  \label{eq:linearmea}
 \end{align}   
 where we use $\Re$ and $\Im$ to denote the real part and imaginary part, $\breve{\bz}^*$ to denote the conjugate transpose of $\breve{\bz}$, and $\diag(\ba)$ to denote the diagonal matrix with diagonal $\ba$. Since 
 (\ref{eq:linearmea}) does not contain any information on $\|\bx\|_2$, 
 we note that $\Re(\breve{\bz}^*\bPhi\bx)=\|\bPhi\bx\|_1$ and further enforce an additional measurement 
 \begin{align}
      \frac{1}{\kappa m}\Re(\breve{\bz}^*\bPhi)\bx =1\label{eq:virtualmea}
 \end{align}
  with $\kappa = \sqrt{\frac{\pi}{2}}$ to specify the norm of the desired signal. (The value of $\kappa$ here is   non-essential  but chosen   to facilitate subsequent analysis.) Combining (\ref{eq:linearmea}) and (\ref{eq:virtualmea}), we arrive at the   {\it linear compressed sensing} problem 
 \begin{align}\label{eq:new_CS}
     \text{find sparse}~\bu,~~\text{such that }\begin{bmatrix}
         \frac{1}{\kappa m}\Re(\breve{\bz}^*\bPhi)\\
         \frac{1}{\sqrt{m}}\Im(\diag(\breve{\bz}^*)\bPhi)
     \end{bmatrix}\bu=\be_1.
 \end{align}
In the noisy case with $\breve{\bz}\ne \bz$, the linear equations in (\ref{eq:new_CS}) become inexact and we instead encounter a noisy linear compressed sensing problem.

For convenience, for any $\bw\in\mathbb{C}^m$, we introduce the notation
\begin{align}\label{Awmatrix}
    \bA_{\bw}:=\begin{bmatrix}
         \frac{1}{\kappa m}\Re(\bw^*\bPhi)\\
         \frac{1}{\sqrt{m}}\Im(\diag(\bw^*)\bPhi)
     \end{bmatrix}.
\end{align}
We will refer to $\bA_{\breve{\bz}}$ in (\ref{eq:new_CS}) as the new sensing matrix    in order to distinguish it with the original complex sensing matrix $\bPhi$. For a fixed $\bx\in\mathbb{S}^{n-1}$ and $\bPhi$ with i.i.d. $\calN(0,1)+\calN(0,1)\bi$ entries, it was proved   \cite{jacques2021importance}   that,  with small enough $\|\breve{\bz}-\bz\|_{\infty}=\max_{i\in[m]}|\breve{z}_i-z_i|$, the matrix $\bA_{\breve{\bz}}\in \mathbb{R}^{(m+1)\times n}$ with $m=O(s\log(\frac{en}{s}))$ satisfies RIP over the cone of $2s$-sparse vectors with high probability.  Hence,  one may solve (\ref{eq:new_CS})   by instance optimal algorithms from linear compressed sensing theory to exactly recover sparse signals and accurately recover approximately sparse signals. However,  this guarantee from \cite{jacques2021importance} is a nonuniform instance optimality result that only works for a fixed $\bx\in \mathbb{S}^{n-1}$. In fact,  the new sensing matrix $\bA_{\breve{\bz}}$ depends on $\bx$ through $\breve{\bz}$, so proving the RIP of  $\bA_{\breve{\bz}}$ for a fixed $\bx$
only implies the nonuniform recovery of this specific $\bx$. 

  In  \cite{chen2023uniform}, the  authors  used a covering argument to show that the      matrices $\bA_{\bz}=\bA_{\sign(\bPhi\bx)}$ in $\rev{\{\bA_{\sign(\bA\bx)} :\bx\in\Sigma^{n,*}_s:=\Sigma^n_s\cap \mathbb{S}^{n-1}\}}$  simultaneously  obey RIP under a near-optimal number of measurements. This leads to a uniform exact reconstruction guarantee over $\Sigma^{n,*}_s$, but provides no guarantee for $\bx \notin \Sigma^{n,*}_s$. Thus, their guarantee is not instance optimal. 


The first contribution of this work is to show that the above linearization approach indeed achieves uniform instance optimality, which is  stronger than the nonuniform result in \cite{jacques2021importance}. For concreteness, we focus on quadratically constrained basis pursuit, that is to solve $\hat{\bx}$ from 
$\Delta(\bA_{\breve{\bz}};\be_1;\varepsilon)$
    \begin{align}\label{eq:nbp}
   \min~\|\bu\|_1,\quad\text{subject to }~~\|\bA_{\breve{\bz}}\bu-\be_1\|_2\le \varepsilon 
\end{align} 
for some suitably chosen $\varepsilon$, and then use $\bx^\sharp = \frac{\hat{\bx}}{\|\hat{\bx}\|_2}$ as an estimator for $\bx$. In the noiseless case,
we show that  when using a complex Gaussian matrix $\bPhi$ with $O(s\log(\frac{en}{s}))$ rows,  with high probability $\bx^\sharp$ satisfies
\begin{align*}
    \|\bx^\sharp-\bx\|_2 \le \frac{10\sigma_{\ell_1}(\bx,\Sigma^n_s)}{\sqrt{s}},\qquad\forall \bx\in \mathbb{S}^{n-1}. 
\end{align*}

 The main ingredient is to show the matrices $\bA_{\bz}=\bA_{\sign(\bPhi\bx)}$ in $\{\rev{ \bA_{\sign(\bPhi\bx)}}:\bx\in \mathbb{B}_1^n(\sqrt{2s})\cap \mathbb{S}^{n-1}\}$, where $\mathbb{B}_1^n(\sqrt{2s})=\{\bu\in\mathbb{R}^{n}:\|\bu\|_1\le\sqrt{2s}\}$ is the scaled-$\ell_1$ ball, simultaneously satisfy RIP. Note that  \cite{chen2023uniform} proved that the  matrices in $\{\bA_{\sign(\bPhi\bx)}:\bx\in \Sigma^{n,*}_s\}$  simultaneously satisfy RIP  through a covering argument, yet their arguments are not sufficient to prove the RIP of $\{\rev{\bA_{\sign(\bPhi\bx)}}:\bx\in \mathbb{B}_1^n(\sqrt{2s})\cap \mathbb{S}^{n-1}\}$. The main issue is that $ \mathbb{B}_1^n(\sqrt{2s})\cap \mathbb{S}^{n-1}$ is   essentially larger than $\Sigma^{n,*}_s$ in terms of metric entropy (or covering number) under an $o(1)$ covering radius. To that end, we utilize a finer treatment to the perturbation of the complex phases to avoid a heavy-tailed random process. See Appendix \ref{apptechnical} for a summary. More generally, we establish the RIP of $\{\bA_{\sign(\bA\bx)}:\bx\in\calK\}$ over some cone $\calU$ for an arbitrary set $\calK\subset\mathbb{S}^{n-1}$. Here is an informal version, \rev{where $\omega(\calU)$ denotes the Gaussian width of a set $\calU$ (see Section \ref{sec:gwme}).}  

 \begin{theorem*}
 Given a cone $\calU$ in $\mathbb{R}^n$ and $\calK\subset\mathbb{S}^{n-1}$, if $m\ge C_1\big(\omega^2(\calU\cap \mathbb{S}^{n-1})+\omega^2(\calK)\big)$ with sufficiently large $C_1$, then with high probability on the complex Gaussian $\bPhi$, the matrices $\{\rev{\bA_{\sign(\bPhi\bx)}}:\bx\in\calK\}$ satisfy RIP over $\calU$ with small enough distortion. 
 \end{theorem*}


Our second contribution is to understand the robustness of $\bx^\sharp$ to different patterns of noise and corruption.\footnote{As a convention, noise refers to perturbation with small magnitude, while corruption can change a measurement to an arbitrary phase-only value.} Prior works  \cite{jacques2021importance,chen2023uniform} only considered small dense noise appearing after applying the sign function (termed as {\it post-sign noise}), formulated as $\breve{\bz} = \sign(\bPhi\bx) + \btau$   where $\btau\in \mathbb{C}_m$ satisfies $\|\btau\|_\infty = \max_i |\tau_i|\le \tau_0$. Under small enough $\tau_0$, they showed a stability result that such $\btau$ increments the estimation error of $\bx^{\sharp}$ by $O(\tau_0)$.  However, many questions remain unaddressed: Is the $O(\tau_0)$ bound for post-sign noise tight? Is $\bx^\sharp$ robust to small dense noise appearing before $\sign(\cdot)$, which we call {\it pre-sign noise}? Is the estimator robust to malicious sparse phase corruption? If so, how do   the pre-sign noise and the sparse corruption increment $\|\bx^\sharp-\bx\|_2$? Are these increments tight or suboptimal, in some sense?

 Our results provide answers to all of these questions. Let us consider the nonuniform recovery of a fixed sparse signal. We show that small dense noise $\btau$ considered in \cite{chen2023uniform,jacques2021importance}, even when appearing before taking the phases (i.e., pre-sign noise), increments $\|\bx^\sharp-\bx\|_2$ by $O(\tau_0)$. Moreover, the $O(\tau_0)$ bound achieved by $\bx^\sharp$ for pre-sign/post-sign dense noise is nearly tight over all algorithms. We also investigate the impact of sparse   corruption which adversarially moves $\zeta_0m$ measurements to arbitrary   phase-only values. This increments $\|\bx^\sharp-\bx\|_2$ by $O(\sqrt{\zeta_0\log(1/\zeta_0)})$ if $\zeta_0$ is small enough. We expect that $\tilde{O}(\sqrt{\zeta_0})$ is tight for the specific estimator $\bx^\sharp$, which we support by providing a partial analysis and numerical evidence. However, for general estimators such dependence is suboptimal and can be improved to ``zero,'' as in this regime we can still exactly recover $\bx$. In particular,  we propose to reformulate PO-CS under sparse corruption to a noiseless linear compressed sensing problem with an extended new sensing matrix. This matrix is then shown to satisfy RIP, thus implying exact reconstruction. See Table \ref{table:robust} for a summary of these results. 
\begin{table}[ht!]
    \centering
    \begin{tabular}{|c|c|c|c|c|c|}
        \hline 
        ~ &  Assumption & Error Bound  &  \makecell{Tightness w.r.t. \\ the estimator $\bx^\sharp$} & \makecell{Tightness w.r.t.\\ all estimators} & Simulation \\
         \hline 
        \makecell{Bounded Dense\\ Noise $\btau$} & $\|\btau\|_\infty\le \tau_0$ & \makecell{$O(\tau_0)$ \\ Thms. \ref{thm:noisy_sparse} \& \ref{thm:pre-noise}}& \makecell{Yes, up to log \\ Thm. \ref{thm:tighttau}} & \makecell{Yes, up to log \\ Thm. \ref{thm:tighttau}} & Figs.  \ref{fig:postnoise}--\ref{fig:prenoise}
        \\\hline
        \makecell{Sparse  \\Corruption $\bzeta$} & $\|\bzeta\|_0\le \zeta_0m$ & \makecell{$O\Big(\sqrt{\zeta_0\log(1/\zeta_0)}\Big)$\\ Thm. \ref{thm:zetabound}} & \makecell{Partially Yes\\ Prop. \ref{pro2}} & \makecell{No, $0$ is optimal\\Thm. \ref{thm:perfect}} &  Fig. \ref{fig:corruption}
        \\\hline 
    \end{tabular}
    \caption{A summary of our results on robustness.  The estimator $\bx^\sharp$ is defined in (\ref{eq:nbp}). By tightness, we mean the question of whether the scaling of the error bound is the best possible (for estimator $\bx^\sharp$ in the fourth column, and general estimators in the fifth column). \label{table:robust}}
    \label{tab:my_label}
\end{table}

Combining the two developments, we obtain the following result that closely resembles the instance optimal guarantee in linear compressed sensing in (\ref{iop00}); see Section \ref{sec:simulation}.  

 \begin{theorem*} Under noisy observations $\breve{\bz}=\sign(\bPhi\bx+\btau_{(1)}+\bzeta_{(1)})+\btau_{(2)}+\bzeta_{(2)}$ with $\|\btau_{(j)}\|_\infty\le \tau_0$ and $\|\bzeta_{j}\|_0\le \zeta_0m$ for $j=1,2$ and some small enough $\tau_0,\zeta_0$, and $\|\bzeta_{(2)}\|_\infty\le 2$, consider $\bx^\sharp$ with properly tuned $\varepsilon$. Then with high probability on the complex Gaussian $\bPhi$, we have
 \begin{align}
     \|\bx^\sharp -\bx\|_2\le C_1\frac{\sigma_{\ell_1}(\bx,\Sigma^n_s)}{\sqrt{s}} + C_2 \tau_0+\left(C_3\sqrt{\zeta_0\log(1/\zeta_0)} +C_4\sqrt{\frac{s\log(en/s)}{m}}\right)\mathbbm{1}(\zeta_0>0),\quad\forall\bx\in\mathbb{S}^{n-1}. \label{informalintro}
 \end{align}
 \end{theorem*}
 To our knowledge, this type of result is novel in nonlinear sensing problems, for which most existing guarantees provide the same error rate to all signals of interest and hence are not instance optimal (e.g., \cite{xu2020quantized,chen2024optimal,plan2012robust,plan2013one,plan2017high,plan2016generalized,genzel2023unified}).\footnote{For instance, the best known $\ell_2$ error rate for 1-bit compressed sensing of the approximately sparse signals in $\sqrt{s}\mathbb{B}_1^n\cap \mathbb{S}^{n-1}$  is $\tilde{O}((s/m)^{1/3})$, while the optimal rate for sparse signals in $\Sigma^{n,*}_s$ is $\tilde{O}(s/m)$.} One exception is the instance optimality in sparse phase retrieval achieved by an intractable algorithm \cite{gao2016stable}. We also note a generic discussion \cite{keriven2018instance} which does not lead to efficient algorithm or address specific model.

{\bf Organization.} The remainder of this paper is arranged as follows. In Section \ref{sec:pre} we give the preliminaries. We present our main results on instance optimality and (nonuniform) robustness in Section \ref{sec:main_res}, along with a few simulation results. In Section \ref{sec:simulation} we combine the instance optimality and nonuniform robustness in the previous section to establish the above (\ref{informalintro}). We give concluding remarks in Section \ref{sec:concluding} to close the paper. Some lengthy and secondary proofs are relegated to the appendices. The proof of Theorem \ref{thm:urip} is rather technical and modified from the arguments in \cite{chen2023uniform} (with crucial improvements); it is hence presented in Appendix \ref{sec:proofs}. The proofs of two side results are postponed to Appendix \ref{app:missing}.

\section{Preliminaries} \label{sec:pre} 
We start with some notational conventions. We denote matrices and vectors by boldface letters, and scalars by regular letters.  $|\mathcal{S}|$   denotes the cardinality of a finite set $\mathcal{S}$. We use $\log(\cdot)$ to denote the natural logarithm to the base of the mathematical constant $e$. The standard Euclidean sphere,  the $\ell_2$-ball and the $\ell_1$-ball in $\mathbb{R}^n$ are denoted by   $\mathbb{S}^{n-1} $, $\mathbb{B}_2^n$ and $\mathbb{B}_1^n$, respectively. Given  $\calK,\calK'\subset \mathbb{R}^n$ and some $\lambda\in\mathbb{R}$, we let $\calK+\lambda\calK':=\{\bu+\lambda \bv:\bu \in\calK,\bv\in \calK'\}$, $\calK_{(\lambda)} = (\calK-\calK)\cap (\lambda\mathbb{B}_2^n)$ and \rev{$\calK^{\mathbb{S}}=\calK\cap\mathbb{S}^{n-1}$}.  We  write $\mathbb{B}_2^n(\bu;r):=\bu+r\mathbb{B}_2^n$, $\mathbb{B}_2^n(r):=r\mathbb{B}_2^n$ and $\mathbb{B}_1^n(r):=r\mathbb{B}_1^n$. We also define $\rad(\calK)=\sup_{\bu\in\calK}\|\bu\|_2$. Recall that $\Sigma^{n,*}_s:=(\Sigma^n_s)^{\mathbb{S}}$ is the set of all $s$-sparse signals in $\mathbb{S}^{n-1}$.

We refer to complex numbers with absolute value $1$ as the phase-only values. For a vector $\bu=[u_i]\in \mathbb{C}^n$, we work with the $\ell_p$-norm $\|\bu\|_{p}=(\sum_i |u_i|^p)^{1/p}$ ($p\geq 1$), max norm $\|\bu\|_\infty=\max_i|u_i|$, and zero ``norm'' $\|\bu\|_0$ that counts the number of non-zero entries. Further given $\bv=[v_i]\in \mathbb{C}^n$, we work with the inner product $\langle \bu,\bv\rangle = \bu^*\bv=\sum_{i=1}^n \overline{u_i}v_i$ and the Hadamard product $\bu\odot\bv = (u_1v_1,u_2v_2,\cdots,u_nv_n)^\top$. 
For a complex matrix $\bA=\bB+\bC\bi$ with $\bi$ reserved for the \rev{imaginary unit}, we will use $\bA^\Re$ or $\Re(\bA)$ to denote its real part $\bB$, and $\bA^\Im$ or $\Im(\bA)$ to denote its imaginary part $\bC$. For a random variable $X$ we define the sub-Gaussian norm as $\|X\|_{\psi_2}=\inf\{t>0:\mathbbm{E}\exp(X^2/t^2)<2\}$. 
For independent zero-mean sub-Gaussian variables $\{X_i\}_{i=1}^N$, there exists absolute constant $C$ such that
\begin{align}\label{sgsum}
    \left\|\sum_{i=1}^N X_i\right\|_{\psi_2}^2 \le C \sum_{i=1}^N\|X_i\|_{\psi_2}^2.
\end{align}
The sub-Gaussian norm of a random vector $\bX\in \mathbb{R}^n$ is defined as $\|\bX\|_{\psi_2} = \sup_{\bv\in \mathbb{S}^{n-1}}\|\bv^\top \bX\|_{\psi_2}$. We refer readers to \cite[Sec. 2]{vershynin2018high} for details on these definitions and properties.

We use $\{C,C_1,C_2,\cdots\}$ and $\{c,c_1,c_2,\cdots\}$ to denote absolute constants whose values may vary from line to line. For two positive quantities $I_1$ and $I_2$, we write $I_1=O(I_2)$   if $I_1\leq CI_2$  holds for some  $C$, and write $I_1=\Omega(I_2)$  if $I_1\geq cI_2$ for some $c>0$. We write $I_1=\Theta(I_2)$ if $I_1=O(I_2)$ and $I_1=\Omega(I_2)$ simultaneously hold. We  also use $\tilde{O}(\cdot),\tilde{\Omega}(\cdot),\tilde{\Theta}(\cdot)$ as the less precise versions of these that hide log factors \rev{in $(m,n,s)$}. We use $o(1)$ to generically denote quantity that tends to zero when $m,n,s\to\infty$.

We say $I_1$ is small enough (or sufficiently small) if $I_1\le c_1$ for some suitably small constant $c_1$. Conversely, it is large enough (or sufficiently large) if $I_1\ge C_1$ for some suitably large constant $C_1$. We say $I_1$ is bounded away from $0$ if $I_1\ge c_1$ for some $c_1>0$.


\subsection{Linear Compressed Sensing} 
Let $\calU$ be a cone in $\mathbb{R}^n$, we say $\bA$ satisfies RIP over $\calU$ with distortion $\delta>0$, denoted by $\bA\sim\rip(\calU,\delta)$, if 
 \begin{align*}
    (1-\delta)\|\bu\|_2^2 \le \|\bA\bu\|_2^2\le (1+\delta)\|\bu\|_2^2,\qquad \forall \bu\in\calU.
 \end{align*}
 By homogeneity, this is equivalent to 
 \begin{align*}
     \sqrt{1-\delta}\le\|\bA\bu\|_2\le \sqrt{1+\delta},\qquad \forall \bu\in\calU^{\mathbb{S}}=\calU\cap\mathbb{S}^{n-1}.  
 \end{align*}
 To be specific, we will focus on sparse recovery and the corresponding program of $\ell_1$-norm minimization (\ref{linearbp}). Therefore, we typically utilize the RIP over $\Sigma^n_{ts}$ for certain $t>0$ to imply the instance optimality. We shall work with RIP over $\Sigma^n_{2s}$ and utilize the following. (Note that RIP$(\Sigma^n_{2s},\delta)$ with $\delta<\frac{\sqrt{2}}{2}$ works, and we set $\delta=\frac{1}{3}$ just for concreteness.)
 \begin{pro}[see  Thm. 2.1 in \cite{cai2013sparse}] \label{pro1}
 Consider $\hat{\bx}$ obtained by solving $\Delta(\bA;\by;\varepsilon)$ in (\ref{linearbp}). If $\bA\sim\rip(\Sigma^n_{2s},\frac{1}{3})$ and $\|\by-\bA\bx\|_2\le\varepsilon$, then we have 
 \begin{align}\label{iop11}
     \|\hat{\bx}-\bx\|_2 \le 7\varepsilon + 5\frac{\sigma_{\ell_1}(\bx,\Sigma^n_s)}{\sqrt{s}},\qquad \forall \bx\in\mathbb{R}^n.
 \end{align}
 \end{pro}
 
     Guarantees of this type are standard in linear compressed sensing theory (e.g., the block sparsity example below;   see also \cite{traonmilin2018stable,Foucart2013}). We note three important features of (\ref{iop11}): instance optimality characterized by $O(s^{-1/2}\sigma_{\ell_1}(\bx,\Sigma^n_s))$, robustness captured by $O(\varepsilon)$, and uniformity over the entire signal space $\mathbb{R}^n$. Indeed, the central goal of this work is to prove an analog for an efficient PO-CS algorithm.

 To recover block sparse signals $\bx=(\bx_1^\top,\bx_2^\top)^\top\in \Sigma^{n_1}_{s_1}\times \Sigma^{n_2}_{s_2}\subset\mathbb{R}^n$ from $\by=\bA\bx+\bepsilon$, we can use a constrained weighted $\ell_1$ minimization, 
 \begin{align}
     \hat{\bx}=(\hat{\bx}_1^\top,\hat{\bx}_2^\top)^\top = \mathrm{arg}\min_{\bu=(\bu_1^\top,\bu_2^\top)^\top}~ \frac{\|\bu_1\|_1}{\sqrt{s_1}} + \frac{\|\bu_2\|_1}{\sqrt{s_2}},\quad\text{subject to }\|\bA\bu-\by\|_2\le \varepsilon. \label{weightedbp}
 \end{align}
 This weighted $\ell_1$ norm can better promote the above block sparsity, which is slightly more structured than the ordinary sparsity $\Sigma^{n}_{s_1+s_2}$. Similarly to Proposition \ref{pro1}, we have the following. (Note that $\rip(\Sigma^{n_1}_{2s_1}\times\Sigma^{n_2}_{2s_2},\delta)$ with $\delta<\frac{1}{2}$ works, and we set $\delta=\frac{1}{3}$ just for concreteness.)
 \begin{pro}
     [Thms. 4.3 \& 4.6 in \cite{traonmilin2018stable}] \label{blockpro}Consider recovering $\bx=(\bx_1^\top,\bx_2^\top)^\top$ by solving $\hat{\bx}$ from (\ref{weightedbp}). If $\bA\sim\rip(\Sigma^n_{2s_1}\times\Sigma^n_{2s_2},\frac{1}{3})$ and $\|\by-\bA\bx\|_2\le\varepsilon$, then for some absolute constants $C_1,C_2$, we have  
     \begin{align*}
         \|\hat{\bx}-\bx\|_2 \le C_1\varepsilon + C_2 \left(\frac{\sigma_{\ell_1}(\bx_1,\Sigma^{n_1}_{s_1})}{\sqrt{s_1}}+\frac{\sigma_{\ell_1}(\bx_2,\Sigma^{n_2}_{s_2})}{\sqrt{s_2}}\right),\qquad \forall \bx\in\mathbb{R}^n. 
     \end{align*}
 \end{pro}

 \subsection{Gaussian Width \& Metric Entropy} \label{sec:gwme} 
 We need to work with two natural quantities that characterize the complexity of a set $\calK$. The first one is the Gaussian width 
$
    \omega(\calK):= \mathbbm{E}\sup_{\bu\in\calK}~\bg^\top\bu$,
where $\bg\sim \calN(0,\bI_n)$. The second one is the metric entropy $\scrH(\calK,r) = \log \scrN(\calK,r)$  where $\scrN(\calK,r)$ denotes the covering number of $\calK$ under radius $r$, defined as the minimal number of radius-$r$ $\ell_2$-balls needed to cover $\calK$. Metric entropy can be bounded in terms of the Gaussian width via Sudakov's inequality \cite[Coro. 7.4.3]{vershynin2018high}, 
\begin{align}\label{sudakov}
    \scrH(\calK,r)\le   \frac{C\cdot\omega^2(\calK)}{r^2}
\end{align}
for some absolute constant $C$. We also have Dudley's inequality \cite[Sec. 8.1]{vershynin2018high} for the converse purpose. A notable difference is that the Gaussian width remains  invariant  after taking the convex hull, while the metric entropy under $o(1)$ covering radius could change significantly. For instance, the set $\sqrt{s}\mathbb{B}_1^n\cap \mathbb{B}_2^n$ (whose elements are known as the approximately $s$-sparse signals in $\mathbb{B}_2^n$) can be essentially viewed as the convex hull of $\Sigma^n_s\cap \mathbb{B}_2^n$ \cite[Lem. 3.1]{plan2013one}. Their Gaussian widths are of the same order \cite[Sec. 2]{plan2012robust},
\begin{align}\label{gwsparse}
  c_1\sqrt{s\log\big(\frac{en}{s}\big)}\le  \omega(\Sigma^n_s\cap \mathbb{B}_2^n) \le \omega(\sqrt{s}\mathbb{B}_1^n\cap \mathbb{B}_2^n) \le C_2\sqrt{s\log\big(\frac{en}{s}\big)}
\end{align}
for some absolute constants $c_1,C_2$. However, while we have \begin{align} \label{entroSigmans}
    \scrH(\Sigma^n_s\cap \mathbb{B}_2^n,r)\le C_1s\log\Big(\frac{en}{r s}\Big), 
\end{align} after convexification we only have 
\begin{align}\label{entroeffspa}
    \scrH(\sqrt{s}\mathbb{B}_1^n\cap \mathbb{B}_2^n,r)\le C_2r^{-2}s\log\Big(\frac{en}{s}\Big). 
\end{align}  The dependence on $r$ in (\ref{entroeffspa}) is tight in some regime; see \cite[Sec. 3]{plan2013one} and  \cite[Sec. 4.3.3]{plan2017high}. In particular, the cardinality of an $r$-net for $\Sigma^n_s\cap\mathbb{B}_2^n$ logarithmically increases with $r^{-1}$, while that of $\sqrt{s}\mathbb{B}_1^n$ increases quadratically with $r^{-1}$. 

\subsection{Concentration Bound}
Next, we introduce some useful sub-Gaussian concentration bounds that capture the Gaussian width of the relevant set. The following   has proven highly effective in dealing with sparse corruption and yielding uniformity \cite{chen2024optimal,dirksen2021non,jung2021quantized,dirksen2022sharp}, and we will rely on it to achieve similar goals. 
\begin{lem}[e.g., Thm. 2.10 in  \cite{dirksen2021non}]\label{lem:max_k_sum}
    Let $\ba_1,...,\ba_m$ be independent isotropic random vectors with $\max_i\|\ba_i\|_{\psi_2}\leq L$, and consider some given $\calT\subset \mathbbm{R}^n$. If $1\leq k\leq m$, then the event 
    \begin{equation}\label{eq:k_largest_sum}
        \sup_{\bu\in\calT}\max_{\substack{I\subset [m]\\|I|\leq k}}\left(\frac{1}{k}\sum_{i\in I}|\langle\ba_i,\bu\rangle|^2\right)^{1/2}\leq C_1\left(\frac{\omega(\calT)}{\sqrt{k}}+\rad(\calT)\sqrt{\log\big(\frac{em}{k}\big)}\right)
    \end{equation}
    holds with probability at least $1-2\exp(-C_2k\log(\frac{em}{k}))$, where $C_1$ and $C_2$ are absolute constants only depending on $L$.  
\end{lem}

  The following upper bound is a simple consequence of the matrix deviation inequality \cite[Sec. 9.1]{vershynin2018high} and will be of recurring use: if $\bA$ has i.i.d. $\calN(0,1)$ entries and $\calT\subset\mathbb{R}^{n}$, then for any $t\ge 0$, 
\begin{align}\label{upperdevi}
    \mathbbm{P}\left(\sup_{\bu\in\calT}\frac{\|\bA\bu\|_2}{\sqrt{m}} \le \rad(\calT)+ \frac{C_1\omega(\calT)+C_2t\cdot\rad(\calT)}{\sqrt{m}}\right)\ge 1-2\exp(-t^2).
\end{align}
A simple consequence of (\ref{upperdevi}) is that, for $\bPhi$ with i.i.d. $\calN(0,1)+\calN(0,1)\bi$ entries and some $\calT\subset\mathbb{S}^{n-1}$, if $m=\Omega(\omega^2(\calT))$, then with probability at least $1-4\exp(-c_1m)$ we have
\begin{align}
    \sup_{\bu\in \calT}\frac{\|\bPhi\bu\|_2}{\sqrt{m}}\le\sup_{\bu\in \calT}\frac{\|\bPhi^\Re\bu\|_2}{\sqrt{m}}+\sup_{\bu\in \calT}\frac{\|\bPhi^\Im\bu\|_2}{\sqrt{m}} \le C_2. \label{Oonebound}
\end{align}

    \subsection{Perturbation of Complex Phase}
    Under the convention $\frac{x}{0}=\infty$ for any $x\geq 0$,   it holds for any $a,b\in\mathbb{C}$ that (e.g., \cite[Lem. 8]{chen2023uniform})
    \begin{equation}\label{eq:sign_conti}
        |\sign(a)-\sign(b) | \leq \min\left\{\frac{2|a-b|}{\max\{|a|,|b|\}},2\right\}.
    \end{equation}
   Note that $a\in \mathbb{C}$ can be identified with $(\Re(a),\Im(a))^\top\in\mathbb{R}^2$, and we can indeed generalize (\ref{eq:sign_conti}) to any $\ba,\mathbf{b}\in \mathbb{R}^n$,  
    \begin{align}
     \label{eq:vect_sign_conti} 
     \left\|\frac{\ba}{\|\mathbf{a}\|_2}-\frac{\mathbf{b}}{\|\mathbf{b}\|_2}\right\|_2 \leq \min\left\{\frac{2\|\ba-\mathbf{b}\|_2}{\max\{\|\ba\|_2,\|\mathbf{b}\|_2\}},2\right\},
    \end{align} 
    by a proof identical to \cite[Lem. 8]{chen2023uniform}. In general, $|\sign(b+\delta)-\sign(b)|$ is harder to control under smaller $|b|$. This inspires us to introduce the index set    
    \begin{align*}
        \calJ_{\bx,\eta} = \{i\in[m]:|\bPhi_i^*\bx|\le \eta\}
    \end{align*}
    for some $\bx\in \mathbb{S}^{n-1}$ and $\eta>0$. Intuitively, the measurements in $\calJ_{\bx,\eta}$ are possibly problematic in PO-CS in terms of the sensitivity to pre-sign perturbation. 

\section{Main Results}\label{sec:main_res}
We present our results for $\bx^\sharp$ obtained by normalizing $\hat{\bx}=\Delta(\bA_{\breve{\bz}};\be_1;\varepsilon)$ in  (\ref{eq:nbp}).  Throughout the paper, we assume complex Gaussian $\bPhi$ with i.i.d. $\calN(0,1)+\calN(0,1)\bi$ entries (whose real part and imaginary part are independent) without always explicitly stating it.  We will first study the performance of $\hat{\bx}$ and then transfer this to $\bx^\sharp$ via (\ref{eq:vect_sign_conti}); thus, it is useful to identify the ground truth that $\hat{\bx}$ approximates. This is a scaled version of $\bx$ given by \cite{chen2023uniform,jacques2021importance}
\begin{equation}\label{eq:GT}
    \bx^\star := \frac{\kappa m\cdot \bx}{\|\bPhi\bx\|_1}, 
\end{equation}
as it is easy to check $\bA_{\bz}\bx^\star = \be_1$. (This means that $\bx^\star$ is the point that satisfies the linear measurements in (\ref{eq:nbp}) in a noiseless case.) With $\kappa=\sqrt{\frac{\pi}{2}}$, $\frac{\|\bPhi\bx\|_1}{\kappa m}$ sharply concentrates about $1$: by (\ref{sgsum}), we have \[\Big\|\frac{\|\bPhi\bx\|_1}{\kappa m}-1\Big\|_{\psi_2}=\Big\|\frac{1}{m}\sum_{i=1}^m(\kappa^{-1}|\bPhi_i^*\bx|-1)\Big\|_{\psi_2}\le \frac{C}{m}\Big(\sum_{i=1}^m\|\kappa^{-1}|\bPhi_i^*\bx|-1\|_{\psi_2}^2\Big)^{1/2}\le \frac{C_1}{\sqrt{m}},\] which gives  the sub-Gaussian tail bound 
\begin{align}
    \mathbbm{P}\Big(\big|\frac{\|\bPhi\bx\|_1}{\kappa m}-1\big|\ge t\Big)\le 2\exp(-c_2mt)\,,\quad \forall t\ge 0.  \label{l1l2fixedpoint}
\end{align}  Hence, $\bx^\star$ is in general very close to $\bx$. We will provide a uniform version of this fact in the subsequent Proposition \ref{lem:l1l2_rip}.

\rev{Before proceeding, we provide a deterministic proof framework of our upper bounds, which is adapted from the original work of \cite{jacques2021importance}.}

\begin{lem}[Deterministic error bound] \label{pro:deter}
\rev{Fix   $\bx\in \mathbb{S}^{n-1}$, $\bPhi\in \mathbb{C}^{m\times n}$ such that $\|\bPhi\bx\|_1>0$ and suppose that we observe $\breve{\bz}\in \mathbb{C}^m$ that is a possibly noisy version of $\bz:=\sign(\bPhi\bx)$.  Suppose that we compute $\bx^\sharp=\frac{\hat{\bx}}{\|\hat{\bx}\|_2}$ with $\hat{\bx}$ obtained by solving $\Delta(\bA_{\breve{\bz}};\be_1,\varepsilon)$ in (\ref{eq:nbp}) with  $\bA_{\breve{\bz}}$ defined in (\ref{Awmatrix}) and some $\varepsilon\ge 0$.
If $\bA_{\breve{\bz}}\sim \rip(\Sigma^n_{2s},\frac{1}{3})$ and \begin{align}
    \varepsilon\ge \frac{\kappa m}{\|\bPhi\bx\|_1}\Big[\frac{|\Re((\breve{\bz}-\bz)^*\bPhi\bx)|}{\kappa m}+\frac{\|\Im(\diag((\breve{\bz}-\bz)^*)\bPhi\bx)\|_2}{\sqrt{m}}\Big],\label{varepirequire}
\end{align} then \[\|\bx^\sharp - \bx\|_2 \le \frac{10\sigma_{\ell_1}(\bx;\Sigma^n_s)}{\sqrt{s}}+\frac{14\|\bPhi\bx\|_1\varepsilon}{\kappa m}.\]
Moreover, the condition (\ref{varepirequire}) can be relaxed to $\varepsilon\ge \|\bA_{\breve{\bz}}\bx^\star-\be_1\|_2$. 
}   
\end{lem}
\begin{proof}
    \rev{By (\ref{Awmatrix}), $\bx^\star = \frac{\kappa m}{\|\bPhi\bx\|_1}\bx$, $\bA_{\bz}\bx^\star =\be_1$, and \eqref{varepirequire}, we have 
    \begin{align*}
        \|\bA_{\breve{\bz}}\bx^\star-\be_1\|_2 &= \|\bA_{\breve{\bz}}\bx^\star - \bA_{\bz}\bx^\star\|_2 = \|\bA_{\breve{\bz}-\bz}\bx^\star\|_2 =\frac{\kappa m}{\|\bPhi\bx\|_1} \|\bA_{\breve{\bz}-\bz}\bx\|_2 \\
        &\le \frac{\kappa m}{\|\bPhi\bx\|_1}\Big[\frac{|\Re((\breve{\bz}-\bz)^*\bPhi\bx)|}{\kappa m}+\frac{\|\Im(\diag((\breve{\bz}-\bz)^*)\bPhi\bx)\|_2}{\sqrt{m}}\Big]\le \varepsilon.  
    \end{align*}
    (In view of these steps, the condition (\ref{varepirequire}) can be relaxed to $\varepsilon\ge \|\bA_{\breve{\bz}}\bx^\star-\be_1\|_2$.)
    Combining with $\bA_{\breve{\bz}}\sim \rip(\Sigma^n_{2s},\frac{1}{3})$,   Proposition \ref{pro1} yields 
    \[\|\hat{\bx}-\bx^\star\|_2\le 7\varepsilon + 5\frac{\sigma_{\ell_1}(\bx^\star,\Sigma^n_s)}{\sqrt{s}}.\]
    Then, by $\bx^\sharp = \frac{\hat{\bx}}{\|\hat{\bx}\|_2}$, $\bx= \frac{\bx^\star}{\|\bx^\star\|_2}$, (\ref{eq:vect_sign_conti}), and $\|\bx^\star\|_2= \frac{\kappa m}{\|\bPhi\bx\|_1}$, we have
    \begin{align*}
        \|\bx^{\sharp}- \bx\|_2 &= \Big\|\frac{\hat{\bx}}{\|\hat{\bx}\|_2}-\frac{\bx^\star}{\|\bx^\star\|_2}\Big\|_2 \le \frac{2\|\hat{\bx}-\bx^\star\|_2}{\max\{\|\hat{\bx}\|_2,\|\bx^\star\|_2\}}\le \frac{2\|\hat{\bx}-\bx^\star\|_2}{\|\bx^\star\|_2} \\
        &\le \frac{14\varepsilon}{\|\bx^\star\|_2} + \frac{10\sigma_{\ell_1}(\bx^\star,\Sigma^n_s)}{\sqrt{s}\|\bx^\star\|_2} = \frac{14\|\bPhi\bx\|_1\varepsilon}{\kappa m} + \frac{10\sigma_{\ell_1}(\bx,\Sigma^n_s)}{\sqrt{s}}, 
    \end{align*}
    as desired.}
\end{proof} 
\subsection{Instance Optimality}
\rev{For any $\bw\in \mathbb{C}^m$, recall that the matrix $\bA_{\bw}$  is defined in (\ref{Awmatrix}).} Our first result concerns the RIP of \rev{$\{\bA_{\sign(\bPhi\bx)} :\bx\in \calK\}$} for arbitrary $\calK\subset \mathbb{S}^{n-1}$ over a general cone $\calU$. 
\begin{theorem}[RIP of a set of $\bA_{\bz}=\bA_{\sign(\bPhi\bx)}$]\label{thm:urip}
     Given a set $\calK$  contained in  $ \mathbb{S}^{n-1}$,  a cone $\calU\subset \mathbb{R}^n$, and any small enough $\eta\in(0,1)$, we let $r=\eta^2\log^{1/2}(\eta^{-1})$ and consider drawing a complex Gaussian $\bPhi\in \mathbb{C}^{m\times n}$. If \begin{equation}
        \label{eq:urip_sample}
        m\ge C_1\left( \frac{\omega^2(\calU^{\mathbb{S}})}{\eta^2\log(\eta^{-1})}+\frac{\scrH(\calK,\eta^3)}{\eta^2}+\frac{\omega^2(\calK_{(r)})}{\eta^4\log(\eta^{-1})}+\frac{\omega^2(\calK_{(\eta^3)})}{\eta^8\log(\eta^{-1})}\right)
    \end{equation}
    for some large enough     $C_1$, then  for some $C_2$ the event 
    \begin{align*}
       \bA_{\bz}=\rev{ \bA_{\sign(\bPhi\bx)}} \sim \rip(\calU,C_2\eta\log^{1/2}(\eta^{-1})),\quad\forall \bx\in\calK
    \end{align*}
    \rev{holds 
    with probability at least $1- C_3\exp(-c_4\eta^2m)$.}
\end{theorem}

\begin{rem}[Recovering $\calK=\Sigma^{n,*}_s$ in \cite{chen2023uniform}]
    Setting $(\calK,\calU)=(\Sigma^{n,*}_s,\Sigma^n_{2s})$, noticing $\calK_{(t)} = (\calK-\calK)\cap(t\mathbb{B}_2^n)\subset t(\Sigma^n_{2s}\cap\mathbb{B}_2^n)$ and using (\ref{gwsparse})--(\ref{entroSigmans}), we have  
    $$\text{Right-hand side of (\ref{eq:urip_sample})}=O\left(\frac{s\log(\frac{en}{s\eta^3})}{\eta^2}\right).$$Further,
    setting $\eta = \frac{c\delta}{\sqrt{\log(\delta^{-1})}}$ with small enough $c$, Theorem \ref{thm:urip} yields the following: The matrices in $\{\bA_{\sign(\bPhi\bx)}:\bx\in\Sigma^{n,*}_s\}$  simultaneously satisfy $\rip(\Sigma^n_{2s},\delta)$ (w.h.p.)~as long as $m=\Omega(\delta^{-2}\log^2(\delta^{-1})s\log(\frac{en}{s}) )$. This improves on \cite[Thm. 1]{chen2023uniform}, which requires $m=\Omega(\delta^{-4}s\log(\frac{n\log(mn)}{\delta s}))$ for the same purpose, in terms of log factors and the dependence on $\delta$. Indeed, the dependence on $\delta$   matches that of achieving RIP via a Gaussian matrix up to log factors (e.g., see \cite{li2020lower}).
\end{rem} 
\begin{rem}[Arbitrary $\calK\subset\mathbb{S}^{n-1}$]
More importantly, Theorem \ref{thm:urip} applies to arbitrary $\calK\subset\mathbb{S}^{n-1}$ with a number of measurements proportional to $\omega^2(\calK)$. To see this, by Sudakov's inequality (\ref{sudakov}) and $\omega^2(\calK_{(t)})\le \omega^2(\calK-\calK)=4\omega^2(\calK)$ \cite[Sec. 7.5.1]{vershynin2018high}, we find that (\ref{eq:urip_sample}) can be implied by the following based only on the Gaussian width,
\begin{align}\label{gaussianonly}
    m\ge C_1'\left(\frac{\omega^2(\calU^{\mathbb{S}})}{\eta^2\log(\eta^{-1})}+\frac{\omega^2(\calK)}{\eta^8}\right)\quad\text{with large enough }C_1'. 
\end{align}
The first informal theorem in introduction thus follows. 
\end{rem}
We specialize $\calK$ to the set of approximately sparse signals $\mathbb{B}_1^n(\sqrt{2s})\cap\mathbb{S}^{n-1}$ and choose sufficiently small $\eta$, along with the sufficiency of (\ref{gaussianonly}) and (\ref{gwsparse}), to obtain the following.

\begin{coro}
    [RIP of $\bA_{\bz}=\bA_{\sign(\bPhi\bx)}$ over approximately sparse signals]  \label{coro1}  If $m\ge C_1s\log(\frac{en}{s})$, then with probability at least $1-C_2\exp(-c_3m)$ over the complex Gaussian $\bPhi$, we have 
    \begin{align*}
         \rev{ \bA_{\sign(\bPhi\bx)}}\sim\rip(\Sigma^n_{2s},1/3),\quad \forall \bx\in\mathbb{B}_1^n(\sqrt{2s})\cap\mathbb{S}^{n-1},
    \end{align*} 
    The distortion $1/3$ can be replaced by any given positive constant $\delta$, up to changes in the values of $C_1,C_2,c_3$.
\end{coro}

\rev{The proof of Theorem \ref{thm:urip} is based on covering techniques and is analogous to \cite{chen2023uniform}.} Nonetheless, \cite{chen2023uniform} is restricted to $\calK=\Sigma^{n,*}_s$ or at most other $\calK$ with metric entropy logarithmically depending on the covering radius, and the techniques therein do not suffice for proving Corollary \ref{coro1}. We make a number of nontrivial modifications, with the most notable one being to introduce an additional index set when controlling the orthogonal term (see Appendix \ref{apptechnical}). To preserve the presentation flow, we postpone the proof of   Theorem \ref{thm:urip} and the detailed discussions to Appendix \ref{sec:proofs}.

Our first main reconstruction guarantee  immediately follows from Corollary \ref{coro1} and Lemma \ref{pro:deter}.

\begin{theorem}[Uniform instance optimality] 
    \label{thm:io_noiseless}
    Consider the noiseless case where $\breve{\bz}=\bz = \sign(\bPhi\bx)$ and the estimator $\bx^\sharp=\frac{\hat{\bx}}{\|\hat{\bx}\|_2}$ with $\hat{\bx}$ obtained by solving $\Delta(\bA_{\bz};\be_1;0)$ in (\ref{eq:nbp}).
    If $m\ge C_1 s\log(\frac{en}{s})$ for some sufficiently large absolute constant $C_1$, then 
    \begin{equation}\label{eq:iop_guarantee}
        \|\bx^\sharp-\bx\|_2\leq \frac{10\sigma_{\ell_1}(\bx,\Sigma^n_s)}{\sqrt{s}},\qquad\forall\bx\in\mathbb{S}^{n-1}
    \end{equation}
    holds with probability at least $1- C_2\exp(-c_3m)$  over the complex Gaussian $\bPhi$. 
\end{theorem}
\begin{proof} The proof relies on the following fact: with the promised probability, $\|\bPhi\bx\|_1>0$ holds for all $\bx\in \mathbb{B}_1^n(\sqrt{2s})\cap \mathbb{S}^{n-1}$. It is a simple consequence of  the matrix deviation inequality in \cite[Sec. 9.1]{vershynin2018high} or our subsequent Lemma \ref{lem:l1l2_rip}.
     We now consider $\bx\in\mathbb{S}^{n-1}$.  If $\bx\in \mathbb{B}_1^n(\sqrt{2s})$, then the event in Corollary \ref{coro1} gives $\bA_{\bz}=\bA_{\sign(\bPhi\bx)}\sim \rip(\Sigma^n_{2s},1/3)$ with the promised probability. \rev{Since $\|\bPhi\bx\|_1>0$, Lemma \ref{pro:deter} implies  
         \begin{align*}
             \|\bx^\sharp - \bx\|_2 \le 10\frac{\sigma_{\ell_1}(\bx,\Sigma^n_s)}{\sqrt{s}}.
         \end{align*}} 
   If $\bx\notin\mathbb{B}_1^n(\sqrt{2s})$, meaning that $\|\bx\|_1>\sqrt{2s}$, then we let $\bx_{[s]}=\mathrm{arg}\min_{\bu\in\Sigma^n_s}\|\bu-\bx\|_2$ and notice that $\|\bx_{[s]}\|_1\le\sqrt{s}\|\bx_{[s]}\|_2\le\sqrt{s}$, and we have $\sigma_{\ell_1}(\bx,\Sigma^n_s)=\|\bx-\bx_{[s]}\|_1\ge \|\bx\|_1-\|\bx_{[s]}\|_1\ge (\sqrt{2}-1)\sqrt{s}$. Therefore, 
     $$\|\bx^\sharp-\bx\|_2\le 2 \le \frac{2}{\sqrt{2}-1}\frac{\sigma_{\ell_1}(\bx,\Sigma^n_s)}{\sqrt{s}}\le 10\frac{\sigma_{\ell_1}(\bx,\Sigma^n_s)}{\sqrt{s}}.$$ The proof is now complete. 
\end{proof}

\begin{rem}
    While we focus on sparse recovery via basis pursuit, the generality of Theorem \ref{thm:urip} in terms of $(\calK,\calU)$ allows for a straightforward generalization to other signal structures, such as PO-CS of low-rank matrices. Our subsequent technical results (Theorem \ref{thm:nonuni}, Lemmas \ref{lem:improved_lem9}, \ref{lem:l1l2_rip}) are also presented in a similar manner. 
\end{rem}

\subsection{Bounded Dense Noise}

\rev{Next, we study robustness to noise and corruption in a nonuniform setting. 
Specifically, we consider the reconstruction of a fixed sparse signal $\bx \in \Sigma^{n,*}_s$, 
which allows us to decouple the desiderata of uniformity, instance optimality, and robustness, and simply focus on 
the robustness for now. 
We discuss in Section~\ref{sec:simulation} how these nonuniform results can be extended to uniform, instance-optimal guarantees.} For convenience, we  specialize Lemma \ref{pro:deter} to our setting. 

\rev{\begin{lem}[Specialization of Lemma \ref{pro:deter}]\label{lem:specializede} Under complex Gaussian $\bPhi$ and a fixed $\bx\in\Sigma^{n,*}_s$, let $\breve{\bz}$ be a noisy version of $\bz=\sign(\bPhi\bx)$, and suppose that we compute $\bx^\sharp=\frac{\hat{\bx}}{\|\hat{\bx}\|_2}$ with $\hat{\bx}$ obtained by solving $\Delta(\bA_{\breve{\bz}};\be_1,\varepsilon)$ in (\ref{eq:nbp}). If (\ref{varepirequire}) (which can be relaxed to $\varepsilon\ge\|\bA_{\breve{\bz}}\bx^\star-\be_1\|_2$) and 
\begin{align}\label{triAzu}
    \sup_{\bu\in\Sigma^{n,*}_{2s}} \frac{|\Re((\breve{\bz}-\bz)^*\bPhi\bu)|}{\kappa m} + \sup_{\bu\in\Sigma^{n,*}_{2s}}\frac{\|\Im(\diag((\breve{\bz}-\bz)^*)\bPhi\bu)\|_2}{\sqrt{m}}<\frac{1}{9} 
\end{align}
hold,  then with probability at least $1-C\exp(-cm)$, we have 
\begin{align}
    \|\bx^\sharp-\bx\|_2 \le \frac{14\|\bPhi\bx\|_1\varepsilon}{\kappa m}.\label{finalbound11}
\end{align}
\end{lem}}
\begin{proof}
\rev{The statement follows from  Lemma \ref{pro:deter} after the validation of $\bA_{\breve{\bz}}\sim \rip(\Sigma^n_{2s},\frac{1}{3})$. By $\bA_{\breve{\bz}}=\bA_{\bz}+\bA_{\breve{\bz}-\bz}$, 
\[\|\bA_{\bz}\bu\|_2 -  \|\bA_{\breve{\bz}-\bz}\bu\|_2 \le \|\bA_{\breve{\bz}}\bu\|_2 \le \|\bA_{\bz}\bu\|_2 + \|\bA_{\breve{\bz}-\bz}\bu\|_2 \,,\,\quad \forall \bu\in \Sigma^{n,*}_{2s}.\]
By the RIP of $\bA_{\bz}=\bA_{\sign(\bPhi\bx)}$ in Corollary \ref{coro1},  
\begin{align}
       1-c'\le \|\bA_{\bz}\bu\|_2 \le 1+c'\,,\, \forall \bu\in\Sigma^{n,*}_{2s}\,,\quad\text{where~}c':=\min\Big\{\sqrt{\frac{4}{3}}-\frac{10}{9},\frac{8}{9}-\sqrt{\frac{2}{3}}\Big\}\label{ripaz}
\end{align} 
 with the promised probability. By (\ref{Awmatrix}), the triangle inequality, and (\ref{triAzu}), 
\begin{align*}
    \sup_{\bu\in\Sigma^{n,*}_{2s}}\|\bA_{\breve{\bz}-\bz}\bu\|_2 \le \sup_{\bu\in\Sigma^{n,*}_{2s}} \frac{|\Re((\breve{\bz}-\bz)^*\bPhi\bu)|}{\kappa m} + \sup_{\bu\in\Sigma^{n,*}_{2s}}\frac{\|\Im(\diag((\breve{\bz}-\bz)^*)\bPhi\bu)\|_2}{\sqrt{m}}<\frac{1}{9}.   
\end{align*}  
Combining the preceding three displays yields $\sqrt{\frac{2}{3}}\le \|\bA_{\breve{\bz}} \bu\|_2 \le \sqrt{\frac{4}{3}}\,,\, \forall \bu\in \Sigma^{n,*}_{2s}$, which is exactly $\bA_{\breve{\bz}}\sim \rip(\Sigma^n_{2s},\frac{1}{3})$.}   
\end{proof}

We first consider \rev{post-sign} small dense noise $\btau \in\mathbb{C}^m$ obeying $\|\btau\|_\infty \le \tau_0$, that is, observations given by $\breve{\bz}= \bz + \btau$.\footnote{We emphasize that the only constraint on $\btau$ is a small enough max norm; under this constraint, it can be generated by an adversary  having full knowledge of $(\bPhi,\bx)$. } This has been treated in \cite[Sec. IV]{jacques2021importance} and \cite[Thm. 3]{chen2023uniform}. We reproduce the result here as an illustration of the applicability of Lemma \ref{lem:specializede}. 

\begin{theorem}[Post-sign noise] Consider PO-CS of a fixed   $\bx\in \Sigma^{n,*}_s$ from $\breve{\bz}=\sign(\bPhi\bx)+\btau$ with $\btau$ obeying $\|\btau\|_\infty \le \tau_0 \le\frac{1}{36}$. If $m\ge C_1s\log(\frac{en}{s})$ with large enough $C_1$, then the estimator $\bx^\sharp=\frac{\hat{\bx}}{\|\hat{\bx}\|_2}$, with $\hat{\bx}$ being solved from $\Delta(\bA_{\breve{\bz}};\be_1;\frac{5\tau_0}{2})$ in (\ref{eq:nbp}), satisfies 
\begin{align*}
    \|\bx^\sharp-\bx\|_2 \le 36\tau_0
\end{align*}
 with probability \rev{at least $1-C_2\exp(-c_3m)$.} 
\label{thm:noisy_sparse}
\end{theorem}
\begin{proof}  \rev{By Lemma \ref{lem:specializede}, we only need to establish (\ref{triAzu}) and (\ref{varepirequire}) with $\varepsilon=\frac{5}{2}\tau_0$ with the promised probability: once these two conditions hold true, then by Lemma \ref{lem:specializede},
we have $\|\bx^\sharp-\bx\|_2\le \frac{\|\bPhi\bx\|_135\tau_0}{\kappa m}$, and further using $\mathbb{P}(\frac{\|\bPhi\bx\|_1}{\kappa m}\le \frac{36}{35})\ge 1-\exp(-cm)$ yields the claim (see Equation (\ref{l1l2fixedpoint})).
}

    {\bf Establishing (\ref{triAzu}):} Note that $\breve{\bz}-\bz=\btau$ satisfies $\|\btau\|_\infty\le \tau_0 \le \frac{1}{36}$, by triangle inequality,  
    \begin{align}\label{eq18}
         \sup_{\bu\in \Sigma^{n,*}_{2s}} \frac{|\Re(\btau^*\bPhi\bu)|}{\kappa m} + \sup_{\bu\in \Sigma^{n,*}_{2s}} \frac{\|\Im(\diag(\btau^*)\bPhi\bu)\|_2}{\sqrt{m}}\le \tau_0 \sup_{\bu\in\Sigma^{n,*}_{2s}}\frac{\|\bPhi\bu\|_1}{\kappa m} + \tau_0\left(\sup_{\bu\in\Sigma^{n,*}_{2s}}\frac{\|\bPhi^\Re\bu\|_2}{\sqrt{m}} + \sup_{\bu\in\Sigma^{n,*}_{2s}}\frac{\|\bPhi^\Im\bu\|_2}{\sqrt{m}}\right). 
    \end{align}
We use a concentration bound from prior works in the area: by \cite[Lem. 6]{chen2023uniform} (or \cite[Thm. 6]{feuillen2020ell}), if $m\ge C_1 s\log(\frac{en}{s})$ for large enough $C_1$, then \begin{align}
    \sup_{\bu\in\Sigma^{n,*}_{2s}}\frac{\|\bPhi\bu\|_1}{\kappa m}\le \frac{3}{2}\label{add11}
\end{align} with probability at least \rev{$1-2\exp(-c_1m)$}. (This can also be achieved by   Lemma \ref{lem:l1l2_rip} appearing later.) Since $\bPhi^\Re$ and $\bPhi^\Im$ have i.i.d. $\calN(0,1)$ entries,   (\ref{upperdevi}) gives    \begin{align}
    \sup_{\bu\in\Sigma^{n,*}_{2s}}\frac{\|\bPhi^\Re\bu\|_2}{\sqrt{m}} + \sup_{\bu\in\Sigma^{n,*}_{2s}}\frac{\|\bPhi^\Im\bu\|_2}{\sqrt{m}}\le \frac{5}{2}\label{add22}
\end{align}   with probability at least $1-4\exp(-c_1m)$, if $m\ge C_2 s\log(\frac{en}{s})$ for large enough $C_2$. \rev{We now substitute (\ref{add11}) and (\ref{add22}) into (\ref{eq18}) and use  $\tau_0\le\frac{1}{36}$ to obtain (\ref{triAzu}).} 


     {\bf Establishing (\ref{varepirequire}) with $\varepsilon=\frac{5}{2}\tau_0$:}   This can be seen by 
    \begin{align} 
        \frac{\kappa m}{\|\bPhi\bx\|_1}\left[\frac{|\Re(\btau^*\bPhi\bx)|}{\kappa m}+\frac{\|\Im(\diag(\btau^*)\bPhi\bx)\|_2}{\sqrt{m}}\right]  \le \tau_0 + \tau_0\cdot \frac{\kappa m}{\|\bPhi\bx\|_1}\cdot \frac{\|\bPhi\bx\|_2}{\sqrt{m}} \le \frac{5\tau_0}{2}, \label{ex11}
    \end{align}
    \rev{where in the first inequality we use $|\Re(\btau^*\bPhi\bx)|\le \tau_0 \|\bPhi\bx\|_1$ and $\|\Im(\diag(\btau^*)\bPhi\bx)\|_2\le\tau_0\|\bPhi\bx\|_2$, and in the second inequality we use $\frac{\kappa m}{\|\bPhi\bx\|_1}\cdot\frac{\|\bPhi\bx\|_2}{\sqrt{m}}<\frac{3}{2}$, which holds because the sub-Gaussian tail bound in (\ref{l1l2fixedpoint}) guarantees sufficiently small $|\frac{\|\bPhi\bx\|_1}{\kappa m}-1|$, and  Bernstein's inequality \cite[Thm. 2.8.1]{vershynin2018high}       guarantees sufficiently small $|\frac{\|\bPhi\bx\|_2^2}{m}-2|$ with the promised probability.  (In greater detail, since $|\frac{\|\bPhi\bx\|_2^2}{m}-2|=|\frac{1}{m}\sum_{i=1}^m(|\bPhi_i^*\bx|^2-\mathbb{E}(|\bPhi_i^*\bx|^2) )|$ and $|\bPhi_i^*\bx|^2$ has $O(1)$ sub-exponential norm, Bernstein's inequality gives $\mathbb{P}(|\frac{\|\bPhi\bx\|_2^2}{m}-2|\ge t)\le 2\exp(-c_4mt^2)$ for any $t\in(0,1)$).}
\end{proof}

  While explicit constants are provided in some of our results, no attempts have been made to optimize them.

We   move  to \rev{pre-sign small dense noise}.   We again denote the dense noise by $\btau$, but now the noisy observations are $\breve{\bz}=\sign(\bPhi\bx+\btau)$. The robustness in this regime is less straightforward than the post-sign noise. The reason is that a small enough pre-sign perturbation $\tau$ can still greatly affect $\sign(\bPhi_i^*\bx+\tau)$ if $|\bPhi_i^*\bx|$ is small. This makes the RIP of $\bA_{\breve{\bz}}$ less evident, for which we have to separately treat a small fraction of measurements with small $|\bPhi_i^*\bx|$ and the majority with $|\bPhi_i^*\bx|$ bounded away from $0$. On the other hand, it comes a bit surprising that the recovery error remains at $O(\tau_0)$, since some algebra along with (\ref{eq:vect_sign_conti}) yields a bound $ O(\tau_0)$ on the right-hand side of (\ref{varepirequire}).   
  
\begin{theorem}[Pre-sign noise]\label{thm:pre-noise}
Consider PO-CS of a fixed $\bx\in \Sigma^{n,*}_s$ from $\breve{\bz}=\sign(\bPhi\bx+\btau)$ with $\btau$ obeying $\|\btau\|_\infty \le \tau_0\le c_0$ for some small enough absolute constant $c_0$. If $m\ge C_1s\log(\frac{en}{s})$ with sufficiently large $C_1$, then  the estimator $\bx^\sharp=\frac{\hat{\bx}}{\|\hat{\bx}\|_2}$, with $\hat{\bx}$ being solved from $\Delta(\bA_{\breve{\bz}};\be_1;4\tau_0)$ in (\ref{eq:nbp}), satisfies 
\begin{align*}
    \|\bx^\sharp-\bx\|_2 \le 57\tau_0
\end{align*}
 with probability \rev{at least $1-C_2\exp(-c_3m)$.} 
\end{theorem}
\begin{proof} 
\rev{In light of Lemma \ref{lem:specializede}, we only need to establish (\ref{triAzu}) and (\ref{varepirequire}) with $\varepsilon=4\tau_0$ with the promised probability; once the two components are ready,    Lemma \ref{lem:specializede}
yields $\|\bx^\sharp-\bx\|_2\le \frac{\|\bPhi\bx\|_1\cdot 56\tau_0}{\kappa m}$, which leads to the claim on the event $\frac{\|\bPhi\bx\|_1}{\kappa m}\le \frac{57}{56}$ that holds with probability at least $1-2\exp(-cm)$ (see Equation (\ref{l1l2fixedpoint})).
}

 To get started,  we first transfer $\btau$ to a post-sign noise by writing $\breve{\bz}=\sign(\bPhi\bx)+\tilde{\btau}$, where 
    \begin{align*}
        \tilde{\btau} = \sign(\bPhi\bx+\btau) - \sign(\bPhi\bx). 
    \end{align*}
    The entries of $\tilde{\btau}$ may not be uniformly small, but we can establish a decomposition $\tilde{\btau}=\tilde{\btau}_1+\tilde{\btau}_2$ where $\|\tilde{\btau}_1\|_\infty$ is small and $\tilde{\btau}_2$ is sparse. \rev{We let $\eta_0:=144c_0$ and notice that $\eta_0$ is a small enough absolute constant.} For the fixed $\bx\in\Sigma^{n,*}_s$, 
    \begin{align*}
        \mathbbm{P}(|\bPhi_i^*\bx|\le\eta_0)\le \mathbbm{P}\big(|\Re(\bPhi_i^*\bx)|\le \eta_0\big) \le \sqrt{\frac{2}{\pi}}\eta_0,
    \end{align*}  
     hence for $\calJ_{\bx,\eta_0}=\{i\in[m]:|\bPhi_i^*\bx|\le\eta_0\}$ the Chernoff bound gives \begin{align}\label{nonuboundJx}
        \mathbbm{P}\big(|\calJ_{\bx,\eta_0}|\le \eta_0m\big) \ge \mathbbm{P}\Big({\rm Bin}(m,\sqrt{2/\pi}\cdot\eta_0)\le \eta_0m\Big)\ge 1-\exp(-c_1\eta_0m)
    \end{align}
    for some $c_1>0$. We proceed on this event and define $(\tilde{\btau}_1,\tilde{\btau}_2)$ such that the support of $\tilde{\btau}_2$ is contained in $\calJ_{\bx,\eta_0}$, and the support of $\tilde{\btau}_1$ is contained in $\calJ_{\bx,\eta_0}^c=[m]\setminus\calJ_{\bx,\eta_0}$. These two requirements uniquely determine the decomposition $\tilde{\btau}=\tilde{\btau}_1+\tilde{\btau}_2$, where $\tilde{\btau}_2$  satisfies $\|\tilde{\btau}_2\|_0\le 
    |\calJ_{\bx,\eta_0}|\le \eta_0m$ and $\|\tilde{\btau}_2\|_\infty\le \|\tilde{\btau}\|_\infty\le 2$. Moreover, the entries of $\tilde{\btau}_1$ take the form $$[\sign(\bPhi_i^*\bx+\tau_i)-\sign(\bPhi_i^*\bx)]\mathbbm{1}(|\bPhi_i^*\bx|>\eta_0),$$ and thus in light of (\ref{eq:sign_conti})  it satisfies $\|\tilde{\btau}_1\|_\infty \le \frac{2\tau_0}{\eta_0}$. The observations can now be expressed as
    \begin{align}
        \breve{\bz} = \bz + \tilde{\btau}_1+\tilde{\btau}_2.\label{decompositionbz}
    \end{align}

  {\bf Establishing (\ref{triAzu}):} By (\ref{decompositionbz}) and triangle inequality, to guarantee (\ref{triAzu}), it is sufficient to establish 
  \begin{gather}
       T_1:=\sup_{\bu\in\Sigma^{n,*}_{2s}} \frac{|\Re(\tilde{\btau}_1^*\bPhi\bu)|}{\kappa m} + \sup_{\bu\in\Sigma^{n,*}_{2s}}\frac{\|\Im(\diag(\tilde{\btau}_1^*)\bPhi\bu)\|_2}{\sqrt{m}} \le\frac{1}{18},\label{pretau1}
\\\label{pretau2}
       T_2:=\sup_{\bu\in\Sigma^{n,*}_{2s}} \frac{|\Re(\tilde{\btau}_2^*\bPhi\bu)|}{\kappa m} + \sup_{\bu\in\Sigma^{n,*}_{2s}}\frac{\|\Im(\diag(\tilde{\btau}_2^*)\bPhi\bu)\|_2}{\sqrt{m}}\le \frac{1}{18}.
  \end{gather} 
   With the promised probability, the arguments in Equations (\ref{eq18})--(\ref{add22}) from the proof of Theorem \ref{thm:noisy_sparse} imply $T_1\le 4\|\tilde{\btau}_1\|_\infty\le \frac{8\tau_0}{\eta_0},$ \rev{which is smaller than $\frac{1}{18}$ by $\tau_0\le c_0$ and $\eta_0= 144c_0$.} Next, we need to show  $T_2\le\frac{1}{18}$. We let $\mathbf{1}_{\supp(\tilde{\btau}_2)}\in\{0,1\}^m$ be the vector whose $1$'s indicate the support of $\tilde{\btau}_2$.  By    $\|\tilde{\btau}_2\|_0\le \eta_0m$, $\|\tilde{\btau}_2\|_\infty \le 2$ and Cauchy-Schwarz inequality, 
\begin{align}\nn
T_2&\le \|\tilde{\btau}_2\|_2\sup_{\bu\in\Sigma^{n,*}_{2s}}\frac{\|\bPhi\bu\odot \mathbf{1}_{\supp(\tilde{\btau}_2)}\|_2}{\kappa m} + \|\tilde{\btau}_2\|_\infty\sup_{\bu\in \Sigma^{n,*}_{2s}} \frac{\|\bPhi\bu\odot \mathbf{1}_{\supp(\tilde{\btau}_2)}\|_2}{\sqrt{m}}\\
    &\le \Big(2+\frac{\sqrt{\eta_0}}{\kappa}\Big)\sqrt{\eta_0}\sup_{\bu\in \Sigma^{n,*}_{2s}} \frac{\|\bPhi\bu\odot \mathbf{1}_{\supp(\tilde{\btau}_2)}\|_2}{\sqrt{\eta_0m}}.\label{29bound}
\end{align}
This can be controlled by Lemma \ref{lem:max_k_sum},  
\begin{align}\nn
    \sup_{\bu\in \Sigma^{n,*}_{2s}} \frac{\|\bPhi \bu\odot \mathbf{1}_{\supp(\tilde{\btau}_2)}\|_2}{\sqrt{\eta_0m}} &\le  \sup_{\bu\in \Sigma^{n,*}_{2s}}\max_{\substack{I\subset[m]\\|I|\le\eta_0m}}\left(\frac{1}{\eta_0m}\sum_{i\in I}|\bPhi_i^*\bu|^2\right)^{1/2}\\
    &\le C_3\sqrt{\frac{s\log(en/s)}{\eta_0m}} + C_3\sqrt{\log(\frac{e}{\eta_0})}\label{31bound}
\end{align}
for some absolute constant $C_3$ with the   probability $2\exp(-c_3\eta_0\log(\frac{e}{\eta_0})m)$. Under $m = \Omega(s\log(\frac{en}{s}))$, substituting (\ref{31bound}) into (\ref{29bound}) yields \begin{align}\label{tildetau2bound}
   T_2=O\left(\sqrt{\frac{s\log(en/s)}{m}}+\sqrt{\eta_0\log(\frac{e}{\eta_0})}\right),
\end{align} 
\rev{which is small enough due to the scaling of $m$, $\eta_0=144c_0$, and the fact that $c_0$ is small enough.}

  {\bf Establishing (\ref{varepirequire}) with $\varepsilon=4\tau_0$:} Next, 
  we seek to show 
\begin{align}\label{c4small}
    \frac{\kappa m}{\|\bPhi\bx\|_1}\Big[\frac{|\Re(\tilde{\btau}^*\bPhi\bx)|}{\kappa m}+\frac{\|\Im(\diag(\breve{\btau}^*)\bPhi\bx)\|_2}{\sqrt{m}}\Big]\le 4\tau_0.
\end{align}
Letting $\tilde{\tau}_i = \sign(\bPhi_i^*\bx+\tau_i)-\sign(\bPhi_i^*\bx)$ be the $i$-th entry of $\tilde{\btau}$, we have
\begin{gather}\label{surpr1}
    |\Re(\tilde{\btau}^*\bPhi\bx)| = \left|\sum_{i=1}^m\Re\big(\tilde{\tau}_i^*\bPhi_i^*\bx\big)\right|\le \sum_{i=1}^m\big|\tilde{\tau}_i^*\bPhi_i^*\bx\big|,\\
    \big\|\Im(\diag(\tilde{\btau}^*)\bPhi\bx)\big\|_2 = \left(\sum_{i=1}^m \big[\Im(\tilde{\tau}_i^*\bPhi_i^*\bx)\big]^2\right)^{1/2} \le  \left(\sum_{i=1}^m |\tilde{\tau}_i^*\bPhi_i^*\bx|^2\right)^{1/2}.
\end{gather}
The key observation is that $|\tilde{\tau}_i^*\bPhi_i^*\bx|$ is actually well bounded for any $i\in[m]$. Without loss of generality, we assume $|\bPhi_i^*\bx|>0$,  then by (\ref{eq:sign_conti}),  
\begin{align}\label{surpr3}
    &|\tilde{\tau}_i^*\bPhi_i^*\bx| = |\sign(\bPhi_i^*\bx+\tau_i)-\sign(\bPhi_i^*\bx)||\bPhi_i^*\bx|\le  \frac{2|\tau_i|}{|\bPhi_i^*\bx|}|\bPhi_i^*\bx|  \le 2|\tau_i| \le 2\tau_0. 
\end{align}
Thus, we obtain $|\Re(\tilde{\btau}^*\bPhi\bx)|\le 2m\tau_0$ and $\|\Im(\diag(\tilde{\btau}^*)\bPhi\bx)\|_2\le 2\sqrt{m}\tau_0$. \rev{Substituting these bounds into (\ref{c4small}) and using (\ref{l1l2fixedpoint}) to guarantee that $|\frac{\|\bPhi\bx\|_1}{\kappa m}-1|$ is small enough with the promised probability, we arrive at (\ref{c4small}).} This completes the proof. 
\end{proof}

  The next theorem provides converse bounds that indicate the sharpness of the $O(\tau_0)$ bounds in Theorems \ref{thm:noisy_sparse}--\ref{thm:pre-noise}. Using the complex Gaussian $\bPhi$ and observations $\sign(\bPhi\bx+\btau)$ or $\breve{\bz}=\sign(\bPhi\bx)+\btau$, we show that no algorithm can reconstruct $\bx\in\Sigma^{n,*}_s$ to an error substantially smaller than $O(\tau_0)$. The idea is to identify another signal $\bx'\in\Sigma^{n,*}_s$ that is indistinguishable from $\bx$ and satisfies $\|\bx'-\bx\|_2=\tilde{\Omega}(\tau_0)$. 
  \begin{theorem}[$O(\tau_0)$ is nearly sharp] \label{thm:tighttau}
      For the reconstruction of a fixed $\bx\in \Sigma^{n,*}_s$ ($s\ge 4$) in PO-CS with a complex Gaussian design, we have the following:
      \begin{itemize}
          \item       In the setting of Theorem \ref{thm:pre-noise},  no algorithm can guarantee  an $\ell_2$ error smaller than $\frac{\tau_0}{12\sqrt{\log m}}$ with probability at least $1-\frac{4}{m}$;
          \item   In the setting of Theorem \ref{thm:noisy_sparse}, assume $m=C_0s\log(\frac{en}{s})$ for some absolute constant $C_0$, then no algorithm can achieve an $\ell_2$ error smaller than $\frac{\tau_0}{48C_0\log(\frac{en}{s})\sqrt{\log m}}$  with probability at least $1-\frac{4}{m}-\exp(-c_1s)$.
      \end{itemize}
  \end{theorem}
  \begin{proof}
      To obtain the first statement, for some $\tau_*>0$ to be chosen, we pick $\bdelta\in\Sigma^n_s$ such that 
      \begin{align}\label{deltacon}
       |\supp(\bdelta)\cup\supp(\bx)|\le s,\qquad    \bdelta^\top\bx=0,\qquad \|\bdelta\|_2=\tau_*.
      \end{align} We let $\bx' =  \frac{\bx+\bdelta}{\|\bx+\bdelta\|_2}\in\Sigma^{n,*}_s$. When $\tau_*$ is small enough, we have 
      \begin{align}\label{xpix}
         \frac{\tau_*}{2} \le \|\bx'-\bx\|_ 2= \sqrt{\frac{\tau_*^2}{1+\tau_*^2}+\Big(1-\frac{1}{\sqrt{1+\tau_*^2}}\Big)^2} \le \frac{3\tau_*}{2}.
      \end{align}
      Then by the Gaussian tail bound $\mathbbm{P}_{g\sim\calN(0,1)}(|g|\ge t)\le \exp(-\frac{t^2}{2})$,  $\|\bPhi(\bx'-\bx)\|_\infty\le\|\bPhi^\Re(\bx'-\bx)\|_\infty+\|\bPhi^\Im(\bx'-\bx)\|_\infty$, and $\|\bx'-\bx\|_2\le\frac{3\tau_*}{2}$, a standard union bound shows that \[\|\bPhi(\bx'-\bx)\|_\infty\le 6\sqrt{\log m}\cdot\tau_*\] holds with probability at least $1-\frac{4}{m}$. Letting $\tau_* = \frac{\tau_0}{6\sqrt{\log m}}$, we have identified $\bx'\in\Sigma^{n,*}_s$ satisfying $$\|\bx'-\bx\|_2\ge\frac{\tau_0}{12\sqrt{\log m}}\qquad\text{and}\qquad\|\bPhi(\bx'-\bx)\|_\infty \le \tau_0.$$ In the regime of Theorem \ref{thm:pre-noise}, the observations $\breve{\bz}=\sign(\bPhi\bx')=\sign(\bPhi\bx+\bPhi(\bx'-\bx))$
      can be generated through the following two indistinguishable cases: 
      \begin{itemize}
          \item The underlying signal is $\bx$ and $\btau=\bPhi(\bx'-\bx)$ is added by an adversary as pre-sign noise;
          \item The underlying signal is $\bx'$ and the adversary adds nothing.  
      \end{itemize}
      Therefore, no algorithm can distinguish $\bx$ and $\bx'$, and hence no algorithm can achieve estimation error smaller than $\frac{\tau_0}{12\sqrt{\log m}}$.

      We now move on to the second statement. We consider $\calJ_{\bx,s/4m} = \{i\in[m]:|\bPhi_i^*\bx|\le\frac{s}{4m}\}$. Then by re-iterating the arguments for (\ref{nonuboundJx}), we obtain that $|\calJ_{\bx,s/4m}| \le \frac{s}{4}$ holds with probability at least $1-\exp(-c_1s)$. Now we choose $\tilde{\bdelta}\in\Sigma^n_s$ satisfying the conditions  in (\ref{deltacon}): $|\supp(\tilde{\bdelta})\cup\supp(\bx)|\le s$, $\tilde{\bdelta}^\top\bx=0$ and $\|\tilde{\bdelta}\|_2=\tau_*$. Additionally, we require 
      \begin{align}\label{addirequ}
          \bPhi_i^*\tilde{\bdelta} =0,\qquad \forall i\in \calJ_{\bx,s/4m}.
      \end{align}
      On the event $\{|\calJ_{\bx,s/4m}|\le\frac{s}{4}\}$, (\ref{addirequ}) translates into (no more than) $\frac{s}{2}$ real linear equations, so such $\tilde{\bdelta}$ exists when $s\ge 4$. We consider $\bx' = \frac{\bx+\tilde{\bdelta}}{\|\bx+\tilde{\bdelta}\|_2}\in \Sigma^{n,*}_s$. Similarly to  (\ref{xpix}), we have $\frac{\tau_*}{2}\le\|\bx'-\bx\|_2\le\frac{3\tau_*}{2}$, and moreover, $\|\bPhi(\bx'-\bx)\|_\infty\le 6\sqrt{\log m}\cdot\tau_*$ holds with probability at least $1-\frac{4}{m}$. We  further bound $\|\sign(\bPhi\bx')-\sign(\bPhi\bx)\|_\infty$ by discussing two cases:
      \begin{itemize}
          \item   For $i\in\calJ_{\bx,s/4m}$, we have $\bPhi_i^*\tilde{\bdelta}=0$ and hence $\sign(\bPhi_i^*\bx')=\sign(\bPhi_i^*\bx)$. 
          \item     For $i\notin \calJ_{\bx,s/4m}$, (\ref{eq:sign_conti}) gives
      \begin{align*}
          |\sign(\bPhi_i^*\bx')-\sign(\bPhi_i^*\bx)|\le \frac{2|\bPhi_i^*(\bx'-\bx)|}{s/(4m)}\le 48C_0\log\big(\frac{en}{s}\big)\sqrt{\log m}\cdot\tau_*,
      \end{align*}
      where the last inequality we use $\|\bPhi(\bx'-\bx)\|_\infty\le 6\sqrt{\log m}\cdot\tau_*$ and substitute $m=C_0s\log(\frac{en}{s})$.
      \end{itemize}
     Taken collectively, it follows that $$\|\sign(\bPhi\bx')-\sign(\bPhi\bx)\|_\infty\le 48C_0\log\big(\frac{en}{s}\big)\sqrt{\log m}\cdot \tau_*.$$ Setting $$\tau_*=\frac{\tau_0}{48C_0\log(\frac{en}{s})\sqrt{\log m}},$$ we have identified $\bx,\bx'\in\Sigma^{n,*}_s$ such that $$\|\bx'-\bx\|_2\ge \frac{\tau_0}{96C_0\log(\frac{en}{s})\sqrt{\log m}}\qquad\text{and}\qquad\|\sign(\bPhi\bx')-\sign(\bPhi\bx)\|_\infty\le \tau_0.$$ Therefore, no algorithm can distinguish $\bx$ and $\bx'$ in the regime of Theorem \ref{thm:noisy_sparse} where an adversary can add post-sign noise bounded by $\tau_0$. 
  \end{proof}

  {\bf Simulation:}\footnote{The MATLAB code used to generate the figures in this paper is available at \url{https://github.com/junrenchen58/instance-optimal-pocs}.} We pause to use experimental results to provide evidence of the achievability and tightness of $O(\tau_0)$. In all of our experiments, the data points are averaged over $50$ independent trials, each of which concerns the recovery of $\bx$ uniformly drawn from $\Sigma^{500,*}_5$ using $300$ phase-only measurements. We provide the optimally tuned $\varepsilon$ to basis pursuit (\ref{eq:nbp}), namely $\varepsilon=\|\bA_{\breve{\bz}}\bx^\star-\be_1\|_2$. 
  For the post-sign noise, we test $\tau_0\in\{0.04,0.08,0.12,0.16 ,\cdots,0.36,0.40\}$  and adopt such corruption pattern: find $\theta_0\in[0,\frac{\pi}{2}]$ such that $|e^{\bi\theta_0}-1|=\tau_0$ and then corrupt $\bz$ to $\breve{\bz} =e^{\bi\theta_0}\bz$.   For the pre-sign noise, we test $\tau_0\in\{0.04,0.12,0.20,0.28,\cdots,0.76,0.84\}$ and generate the noisy observations through $\breve{\bz}=\sign(\bPhi\bx+\tau_0\sign(\bPhi\bx)\bi)$.  The results are given in Figures \ref{fig:postnoise}--\ref{fig:prenoise} and are consistent with our theorems.

   \begin{figure}[ht!] 
	      \begin{centering}
	          \includegraphics[width=0.73\columnwidth]{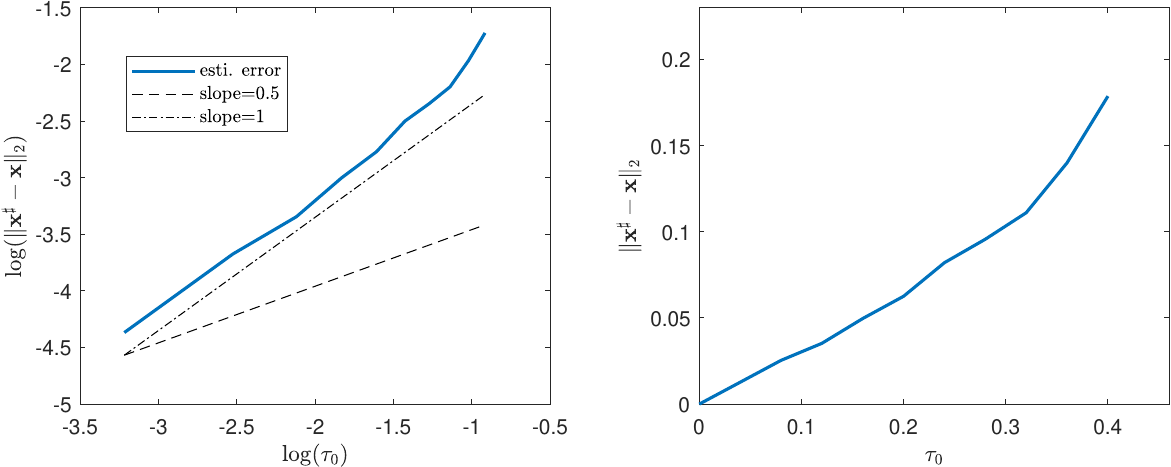} 
	          \par
              \caption{Reconstruction errors under post-sign  bounded by $\tau_0$ \label{fig:postnoise}}
	      \end{centering}
	  \end{figure}
        \begin{figure}[ht!] 
	      \begin{centering}
	          \includegraphics[width=0.73\columnwidth]{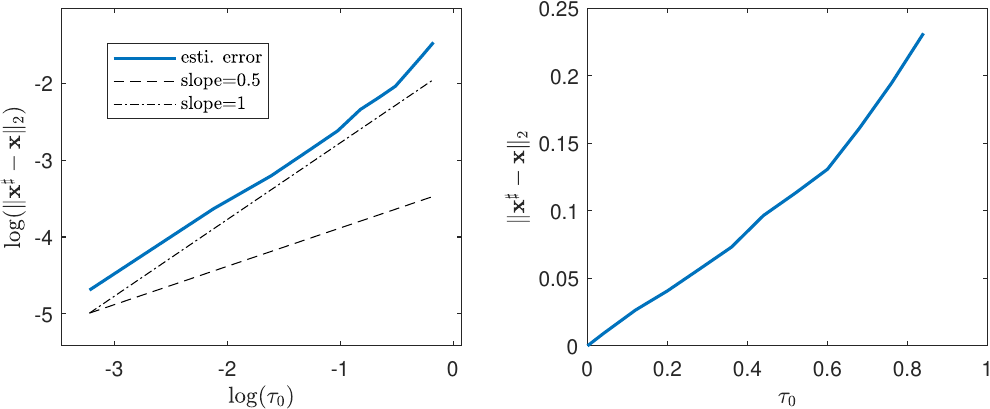} 
	          \par
              \caption{Reconstruction errors under pre-sign   bounded by $\tau_0$ \label{fig:prenoise}}
	      \end{centering}
	  \end{figure}

  \begin{rem}[The choice of noises in Figures \ref{fig:postnoise}--\ref{fig:prenoise}]
  \rev{Note that $O(\tau_0)$ in Theorems \ref{thm:noisy_sparse}, \ref{thm:pre-noise} is a worst-case upper bound that applies to arbitrary $\btau$ obeying $\|\btau\|_\infty\le \tau_0$, but it may not be tight for some $\btau$: for instance, under the pre-sign noise $\btau= \tau_0\sign(\bPhi\bx)$, the observations $\breve{\bz}=\sign(\bPhi\bx+\btau)=\sign(\bPhi\bx)$ are identical to the noiseless ones $\bz$; therefore, exact recovery is achievable and $O(\tau_0)$ is not tight. In view of Lemma \ref{pro:deter},  the final error is at the same order of $\varepsilon$ obeying $\varepsilon\ge \|\bA_{\breve{\bz}}\bx^\star-\be_1\|_2$. To numerically observe the tightness of $O(\tau_0)$, the   noises in Figures \ref{fig:postnoise}--\ref{fig:prenoise} are chosen such that $\|\bA_{\breve{\bz}}\bx^\star-\be_1\|_2=\Omega(\tau_0)$ holds with high probability.}    
  \end{rem}
  
\subsection{Sparse Phase Corruption}
  We consider a corruption $\bzeta$ that can affect a small fraction of the observations arbitrarily over the complex phases. We suppose that there is an adversary (with full knowledge of $\bPhi$ and $\bx$) that can change any $\zeta_0m$ measurements to arbitrary {phase-only values}.\footnote{Note that our formulation (\ref{corruption}) offers slightly more generality.} This setting resembles the adversarial bit flips widely considered in the 1-bit compressed sensing literature \cite{plan2012robust,chen2024optimal,matsumoto2024robust,dirksen2021non} where $\zeta_0m$ signs can be flipped. The mathematical formulation is given by
  \begin{align}\label{corruption}
      \breve{\bz} = \bz + \bzeta
  \end{align}
  for some $\bzeta\in\mathbb{C}^m$ satisfying $\|\bzeta\|_0\le \zeta_0m$ and $\|\bzeta\|_\infty\le 2$. 
  \begin{rem}\label{prepostcorr}
     We consider (\ref{corruption}) only, as the case of pre-sign corruption $\breve{\bz}=\sign(\bPhi\bx+\bzeta)$ can be written as $\breve{\bz}=\sign(\bPhi\bx)+\tilde{\bzeta}$ where $\tilde{\bzeta}=\sign(\bPhi\bx+\bzeta)-\sign(\bPhi\bx)$ satisfies $\|\tilde{\bzeta}\|\le\zeta_0m$ and $\|\tilde{\bzeta}\|_\infty\le 2$.
  \end{rem}
  We show that $\bx^\sharp$ is robust to sparse corruption, in that the $\zeta_0m$ adversarial attacks can only increment the estimation error by $O(\sqrt{\zeta_0\log(1/\zeta_0)})$. 
\begin{theorem}
    [Sparse corruption] Consider PO-CS of a fixed $\bx\in \Sigma^{n,*}_s$ from $\breve{\bz}=\sign(\bPhi\bx)+\bzeta$ with $\bzeta$ obeying $\|\bzeta\|_\infty\le 2$ and $\|\bzeta\|_0 \le \zeta_0m$ for some \rev{$\zeta_0\le c_0$ where $c_0$ is a small enough absolute constant}. If $m\ge C_1s\log(\frac{en}{s})$ with sufficiently large $C_1$, then  the estimator $\bx^\sharp=\frac{\hat{\bx}}{\|\hat{\bx}\|_2}$, with $\hat{\bx}$ being solved from $\Delta(\bA_{\breve{\bz}};\be_1;\rev{5\sqrt{\zeta_0\log(\frac{e}{\zeta_0})}})$ in (\ref{eq:nbp}), satisfies 
\begin{align*}
    \|\bx^\sharp-\bx\|_2 \le 71 \sqrt{\zeta_0\log(e/{\zeta_0})} 
\end{align*}
 with \rev{probability at least $1-C_2\exp(-c_3m)-\exp(-\zeta_0m\log(\frac{e}{\zeta_0}))$  for some absolute constants $C_2,c_3$.} \label{thm:zetabound}
\end{theorem}
\begin{proof}
\rev{By Lemma \ref{lem:specializede}, it suffices to establish (\ref{triAzu}) and (\ref{varepirequire}) with $\varepsilon=5\sqrt{\zeta_0\log(\frac{e}{\zeta_0})}$, after which we have (\ref{finalbound11}) that reads 
$\|\bx^\sharp-\bx\|_2\le \frac{\|\bPhi\bx\|_1\cdot 70\sqrt{\zeta_0\log(e/\zeta_0)}}{\kappa m}$. This yields the desired claim due to $\frac{\|\bPhi\bx\|_1}{\kappa m}<\frac{71}{70}$, which holds with the promised probability due to (\ref{l1l2fixedpoint}). 
}

    {\bf Establishing (\ref{triAzu}):} In the proof of Theorem \ref{thm:pre-noise}, by the arguments in Equations (\ref{29bound})--(\ref{31bound}), we show (\ref{tildetau2bound}) for $\tilde{\btau}_2$ satisfying $\|\tilde{\btau}_2\|_0\le \eta_0m$ and $\|\tilde{\btau}_2\|_\infty\le 2$. For $\bzeta=\breve{\bz}-\bz$  obeying $\|\bzeta\|_0\le \zeta_0m$ and $\|\bzeta\|_\infty\le 2$  in our current setting, where $\zeta_0\le c_0$ for some small enough absolute constant $c_0$,  an identical argument yields
\begin{align*}
    \sup_{\bu\in\Sigma^{n,*}_{2s}}\frac{|\Re(\bzeta^*\bPhi\bu)|}{\kappa m} + \sup_{\bu\in\Sigma^{n,*}_{2s}}\frac{\|\Im(\diag(\bzeta^*)\bPhi\bu)\|_2}{\sqrt{m}} = O\left(\sqrt{\frac{s\log(en/s)}{m}}+\sqrt{c_0\log(e/c_0)}\right)
\end{align*}
    with probability at least $1-2\exp(-c_1m)$. This yields  (\ref{triAzu}) under $m=\Omega(s\log(\frac{en}{s}))$ and small enough $c_0.$

  {\bf Establishing (\ref{varepirequire}) with $\varepsilon=5\sqrt{\zeta_0\log(\frac{e}{\zeta_0})}$:}    
    We let $\mathbf{1}_{\supp(\bzeta)}\in\{-1,1\}^m$ whose $1$'s indicate the support set of $\bzeta$, then 
    \begin{align}\nn 
        &\frac{|\Re(\bzeta^*\bPhi\bx)|}{\kappa m}+\frac{\|\Im(\diag(\bzeta^*)\bPhi\bx)\|_2}{\sqrt{m}} \\\nn&\le \frac{\|\bzeta\|_2\|\bPhi\bx\odot \mathbf{1}_{\supp(\bzeta)}\|_2}{\kappa m}+\frac{2\|\bPhi\bx\odot \mathbf{1}_{\supp(\bzeta)}\|_2}{\sqrt{m}}\\\nn
        &\le \frac{2}{\sqrt{m}}\Big(1+\frac{\sqrt{\zeta_0}}{\kappa}\Big) \|\bPhi\bx\odot \mathbf{1}_{\supp(\bzeta)}\|_2\\
        &\le  \frac{2}{\sqrt{m}}\Big(1+\frac{\sqrt{\zeta_0}}{\kappa}\Big) \max_{\substack{I\subset[m]\\|I|=\zeta_0m}}\left(\sum_{i\in I}|\bPhi_i^*\bx|^2\right)^{1/2}.  \label{43bound} 
    \end{align}
    Without loss of generality, we assume $\zeta_0m$ is a positive integer. For fixed $I$ with cardinality $\zeta_0m$, $\sum_{i\in I}|\bPhi_i^*\bx|^2$ follows Chi-squared distribution with $2\zeta_0m$ degrees of freedom. Then for any $t\ge 0$ and any $I\subset[m]$ with $|I|=\zeta_0m$, a standard concentration bound \cite[Lem. 1]{laurent2000adaptive} yields 
    \begin{align*}
        \mathbbm{P}\left(\sum_{i\in I}|\bPhi_i^*\bx|^2 \le 2\zeta_0m + 2\sqrt{2\zeta_0m t}+2t\right)\ge 1-\exp(-t).  
    \end{align*}
    Taking  a union bound over $\binom{m}{\zeta_0m}$ possible $I$, it yields 
 \begin{align*}
        \mathbbm{P}\left(\max_{\substack{I\subset[m]\\|I|=\zeta_0m}}\sum_{i\in I}|\bPhi_i^*\bx|^2 \le 2\zeta_0m + 2\sqrt{2\zeta_0m t}+2t\right)\ge 1-\exp\Big(\zeta_0m\log\big(\frac{e}{\zeta_0}\big)-t\Big).  
    \end{align*}
    Setting $t=2\zeta_0m \log(\frac{e}{\zeta_0})$ and using the small enough $\zeta_0$, we arrive at
    \begin{align}
        \max_{\substack{I\subset[m]\\|I|=\zeta_0m}}\sum_{i\in I}|\bPhi_i^*\bx|^2\le 5\zeta_0m \log\big(\frac{e}{\zeta_0}\big) \label{45bound} 
    \end{align}    with probability at least $1-\exp(-\zeta_0m\log(\frac{e}{\zeta_0}))$. Substituting  (\ref{45bound}) into (\ref{43bound}),
    using small enough $\zeta_0$ and  using (\ref{l1l2fixedpoint}) to guarantee that $|\frac{\|\bPhi\bx\|_1}{\kappa m}-1|$ is small enough with the promised probability, 
    we establish (\ref{varepirequire}) with $\varepsilon=5\sqrt{\zeta_0\log(\frac{e}{\zeta_0})}$. The proof is complete. 
    \end{proof}
    
    We expect that the error increment $\tilde{O}(\sqrt{\zeta_0})$ is tight for the specific estimator $\bx^\sharp$. To support this, without considering the normalization $\bx^\sharp = \hat{\bx}/\|\hat{\bx}\|_2$, we show $\Omega(\delta_1\sqrt{\zeta_0\log(e/\zeta_0)})$ is a lower bound on $\|\hat{\bx}-\bx^\star\|_2$ under a suboptimal noise parameter     $\varepsilon \ge(1+\delta_1)\|\bA_{\breve{\bz}}\bx^\star-\be_1\|_2$ for basis pursuit (\ref{eq:nbp}).   In practice, this assumption can often be satisfied, and the near-optimal choice $\varepsilon=(1+o(1))\|\bA_{\breve{\bz}}\bx^\star-\be_1\|_2$ could be unrealistic.  
\begin{pro}\label{pro2}
    In the problem setting of Theorem \ref{thm:zetabound},   consider $\hat{\bx}$ solved from $\Delta(\bA_{\breve{\bz}};\be_1;\varepsilon)$ in (\ref{eq:nbp}). If $\varepsilon\ge (1+\delta_1)\|\bA_{\breve{\bz}}\bx^\star-\be_1\|_2$ for some   $\delta_1\in(0,1)$, then under sparse phase corruption $\bzeta$ that changes the $\zeta_0m$ measurements with the largest $|\bPhi_i^*\bx|$ from $z_i$ to $\breve{z}_i = \bi\cdot z_i$, for some absolute constant $c_1$ we have 
    \begin{align*}
        \|\hat{\bx}-\bx^\star\|_2\ge c_1\delta_1\sqrt{\zeta_0\log (e/\zeta_0)} 
    \end{align*} with probability at least $1-C_2\exp(-c_3\zeta_0\log(\frac{e}{\zeta_0})m)$. 
\end{pro}
The key idea is to show all $s$-sparse signals within the ball  $\mathbb{B}_2^n(\bx^\star;\Theta(\delta_1\sqrt{\zeta_0\log(e/\zeta_0)}))$   satisfy the constraint $\|\bA_{\breve{\bz}}\bu-\be_1\|_2\le \varepsilon$. Then, we argue that some signal living on the boundary of this ball is favored  over $\bx^\star$ by the decoder in (\ref{eq:nbp}). Since this statement is positioned as a secondary result, its proof is postponed to Appendix \ref{app:zeta}.

{\bf Simulation:} We pause to provide numerical evidence on the sharpness of $\tilde{O}(\sqrt{\zeta_0})$ for $\bx^\sharp$, even under the optimally tuned noise level $\varepsilon=\|\bA_{\breve{\bz}}\bx^\star-\be_1\|_2$. We adopt the same settings as in earlier simulations \rev{(i.e., $m=300$, $n=500$, $s=5$)} but replace the dense noise $\btau$ by $\zeta_0m$ adversarial phase corruptions. We test $\zeta_0m\in\{1,2,3,5,7,9,11,13\}$ and corrupt the measurements through the mechanism described in Proposition \ref{pro2}. The log-log curve in Figure \ref{fig:corruption}(Left) roughly has a slope of $\frac{1}{2}$ over small $\zeta_0$. This seems to suggest the tightness of $\tilde{O}(\sqrt{\zeta_0})$ for the specific estimator $\bx^\sharp$.   

  \begin{figure}[ht!] 
	      \begin{centering}
	          \quad\includegraphics[width=0.73\columnwidth]{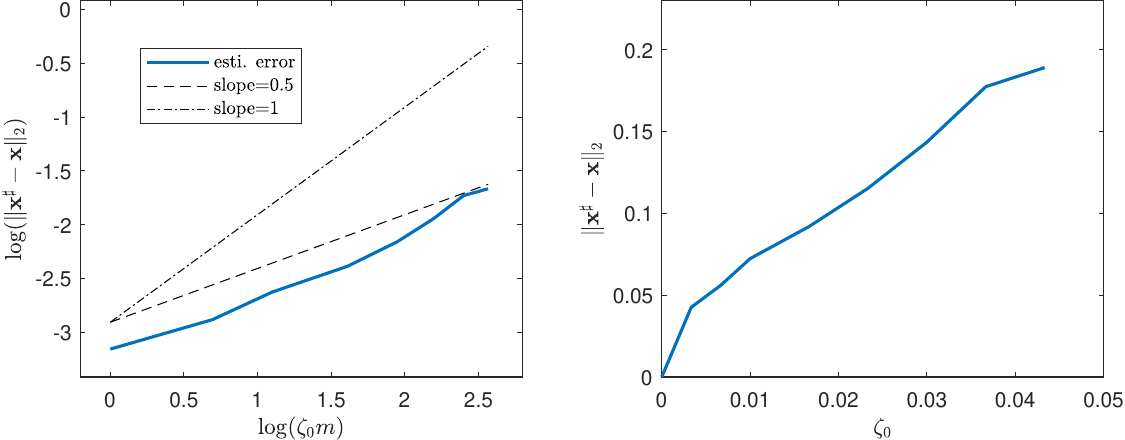} 
	          \par
              \caption{The impact of sparse phase corruption   on $\|\bx^\sharp-\bx\|_2$\label{fig:corruption}}
	      \end{centering}
	  \end{figure}



As related context, in 1-bit compressed sensing, $\zeta_0m$ adversarial bit flips increment the $\ell_2$ error of the convex relaxation approach \cite{plan2012robust} by $\tilde{O}(\sqrt{\zeta_0})$, which was then improved to  $\tilde{O}(\zeta_0)$ using different algorithms \cite{chinot2022adaboost,chen2024optimal,matsumoto2024robust,awasthi2016learning}, and this (almost) linear increment is near-optimal under Gaussian designs (e.g., see \cite[Thm. 2.4]{oymak2015near}).

It is thus natural to investigate the tightness of $\tilde{O}(\sqrt{\zeta_0})$ in Theorem \ref{thm:zetabound} {\it without constraining the algorithm}. We show that   $\tilde{O}(\sqrt{\zeta_0})$   is indeed suboptimal and the impact of the sparse corruption can indeed be {\it eliminated}, meaning that there is an   algorithm being capable of perfectly recovering $\bx$ in this regime. As linear system or compressed sensing with sparse corruption \cite{foygel2014corrupted,haddock2022quantile,mccoy2014sharp,nguyen2012robust}, the intuition is that the uncorrupted measurements remain numerous enough to uniquely identify the signal.        More specifically, we achieve this through
an efficient algorithm, which is an extension of the linearization approach. It reformulates corrupted PO-CS as   linear   compressed sensing with sparse corruption  \cite{nguyen2012robust,foygel2014corrupted,chen2018stable}, which can also be simply viewed as a noiseless extended linear compressed sensing problem.

We consider the setting in Theorem \ref{thm:zetabound}. Combining $\bz = \breve{\bz}-\bzeta$ and $\Im(\frac{1}{\sqrt{m}}\diag(\bz^*)\bPhi)\bx=0$ as in (\ref{eq:linearmea}), we arrive at 
\begin{align}\label{eq:linearmea_cor}
   \frac{1}{\sqrt{m}}\Im(\diag(\breve{\bz}^*)\bPhi)\bx+\bx_{\bzeta} =0, 
\end{align}
where $\bx_{\bzeta}:=\frac{1}{\sqrt{m}}\Im(\diag(\bzeta)\bPhi)\bx$ is $(\zeta_0m)$-sparse. Since (\ref{eq:linearmea_cor}) does not contain any norm information on $(\bx,\bx_{\bzeta})$, as was done in (\ref{eq:virtualmea}), we further  introduce  
\begin{align}\label{eq:virvir}
    \frac{1}{\kappa m}\Re(\breve{\bz}^*\bPhi)\bx=1 
\end{align}
to address the scaling issue. 
  We are faced with a noiseless linear compressed sensing problem with extended signal space, whose goal is to find $(\bx,\bx_{\bzeta})\in\Sigma^{n}_s\times \Sigma^{m}_{\zeta_0m}$ that satisfies (\ref{eq:linearmea_cor}) and (\ref{eq:virvir}). We define an extended new sensing matrix for any $\bw\in \mathbb{C}^m$ as 
\begin{align}
    \widetilde{\bA}_{\bw}: = 
    \begin{bmatrix}
       \frac{1}{\kappa m}\Re(\bw^*\bPhi) &  0~\\
        \frac{1}{\sqrt{m}}\Im(\diag(\bw)^*\bPhi)& \bI_m
    \end{bmatrix}, 
    \label{eq:tilde_Az} 
\end{align}
and then the linear constraints (\ref{eq:linearmea_cor}) and (\ref{eq:virvir}) can be concisely expressed as 
\begin{align}
      \widetilde{\bA}_{\breve{\bz}}\begin{bmatrix}
          \bx \\ \bx_{\bzeta}
      \end{bmatrix} = \be_1. \label{extendedlinear}
\end{align}  

\begin{rem}\label{rem:GTT}
    Like $\bx^\star$ in (\ref{eq:GT}) such that $\bA_{\bz}\bx^\star=\be_1$, the ground truth satisfying (\ref{extendedlinear}) is given by 
    \begin{align*}
        \bx^{\star\star} = \frac{\kappa m\cdot\bx}{\Re(\breve{\bz}^*\bPhi\bx)} \quad\text{and}\quad \bx_{\zeta}^{\star\star} = \frac{\Im(\diag(\bzeta)\bPhi\bx^{\star\star})}{\sqrt{m}}. 
    \end{align*}
\end{rem}

    We propose to find $(\bx,\bx_{\bzeta})$ by solving  weighted $\ell_1$-norm minimization 
     \begin{align}
        (\hat{\bx}_e,\hat{\bx}_{\bzeta}) = \arg\min ~\frac{\|\bu\|_1}{\sqrt{s}} + \frac{\|\bw\|_1}{\sqrt{\zeta_0m}},~~\text{subject to~}\widetilde{\bA}_{\breve{\bz}}\begin{bsmallmatrix}
            \bu\\\bw
        \end{bsmallmatrix}= \be_1  \label{eq:ideal_extend_program}
    \end{align}
    and then use $\bx^\sharp_e  = \hat{\bx}_e/\|\hat{\bx}_e\|_2$ as our final estimate. By establishing the RIP of $\widetilde{\bA}_{\breve{\bz}}$ over $\Sigma^n_{2s}\times\Sigma^m_{2\zeta_0m}$, we obtain the perfect reconstruction $\bx^\sharp_e=\bx$. Let us present the RIP of $\widetilde{\bA}_{\bz}=\widetilde{\bA}_{\sign(\bPhi\bx)}$ for a fixed $\bx$ and then the exact reconstruction guarantee.   
    
\begin{theorem}[RIP of $\widetilde{\bA}_{\bz}=\widetilde{\bA}_{\sign(\bPhi\bx)}$ for fixed $\bx$]
    \label{thm:nonuni}
   Consider $\calU_c =\calU_1\times \calU_2$ for some cones $\calU_1\subset \mathbb{R}^n$ and $\calU_2\subset\mathbb{R}^m$, fixed $\bx\in \mathbb{S}^{n-1}$ and given $\delta\in(0,1)$. For some absolute constants $C_1$ and $c_2$, if $m\ge C_1\delta^{-2}\omega^2(\calU_c^{\mathbb{S}})$, then $\widetilde{\bA}_{\sign(\bPhi\bx)}\sim\rip(\calU_c,\delta)$ with probability at least $1-\exp(-c_2\delta^2m)$.   
\end{theorem}
The proof of Theorem \ref{thm:nonuni} can be found in Appendix \ref{proofthm7}. We can improve it to a uniform statement over $\bx\in\calK$ for some $\calK\subset\mathbb{S}^{n-1}$; see Lemma \ref{lem:RIPextend1}.

\begin{theorem}[Perfect recovery under corruption] 
   Consider the same signal reconstruction problem as in Theorem \ref{thm:zetabound} while using a different estimator $\bx^\sharp_e= \frac{\hat{\bx}_e}{\|\hat{\bx}_e\|_2}$, where $\hat{\bx}_e$ is obtained by solving (\ref{eq:ideal_extend_program}).  If $m\ge C_1s\log(\frac{en}{s})$ for some sufficiently large absolute constant $C_1$, $\zeta_0\le c_0$ for some small enough absolute constant $c_0$, then  $\hat{\bx}_e=\bx$ holds with probability at least $1-C_2\exp(-c_3\zeta_0\log(e/\zeta_0)m)$. \label{thm:perfect}
     \end{theorem} 
    \begin{proof} 
  We first show $\widetilde{\bA}_{\breve{\bz}}\sim\rip(\calU_c,\frac{1}{3})$ where $\calU_c = \Sigma^{n}_{2s}\times\Sigma^{m}_{2\zeta_0m}$. By 
    $\calU_c^{\mathbb{S}}\subset (\Sigma^{n}_{2s}\cap \mathbb{B}_2^n)\times (\Sigma^m_{2\zeta_0m}\cap \mathbb{B}_2^n)$, we have $\omega(\calU_c^{\mathbb{S}})\le \omega(\Sigma^n_{2s}\cap \mathbb{B}_2^n)+\omega(\Sigma^m_{2\zeta_0m}\cap\mathbb{B}_2^m)$, and hence
    $$\omega^2(\calU_c^{\mathbb{S}})\le 2\omega^2(\Sigma^{n}_{2s}\cap\mathbb{B}_2^n)+2\omega^2(\Sigma^{m}_{2\zeta_0m}\cap \mathbb{B}_2^m)\le C_0\Big(s\log\big(\frac{en}{s}\big)+\zeta_0\log\big(\frac{e}{\zeta_0}\big)m\Big)$$
        for some absolute constant $C_0$. Therefore, when $\zeta_0$ is sufficiently small, $m\ge C_1s\log(\frac{en}{s})$ with large enough $C_1$ implies $m\ge C_2\omega^2(\calU_c^{\mathbb{S}})$ with large enough $C_2$. Then, by  Theorem \ref{thm:nonuni}, $\widetilde{\bA}_{\bz}=\widetilde{\bA}_{\sign(\bPhi\bx)}\sim\rip(\calU_c,c_3)$ with some small enough $c_3\le \frac{1}{12}$ holds with probability at least $1-\exp(-c_4m)$. 
        \rev{Under the RIP of $\widetilde{\bA}_{\bz}$ and in view of 
        \[\|\widetilde{\bA}_{\bz}\bu\|_2-\|(\widetilde{\bA}_{\breve{\bz}}-\widetilde{\bA}_{\bz})\bu\|_2\le\|\widetilde{\bA}_{\breve{\bz}}\bu\|_2\le \|\widetilde{\bA}_{\bz}\bu\|_2+\|(\widetilde{\bA}_{\breve{\bz}}-\widetilde{\bA}_{\bz})\bu\|_2\,,\quad\forall \bu\in\calU_{c}^\mathbb{S},\]}
        it remains to ensure that $\sup_{\bu\in\calU_c^{\mathbb{S}}}\|(\widetilde{\bA}_{\breve{\bz}}-\widetilde{\bA}_{\bz})\bu\|_2 $ is sufficiently small. Comparing $\bA_{\bw}$ in (\ref{Awmatrix}) and $\widetilde{\bA}_{\bw}$ in (\ref{eq:tilde_Az}), we have 
        \begin{align}
            &\sup_{\bu\in\calU_c^{\mathbb{S}}}\big\|(\widetilde{\bA}_{\breve{\bz}}-\widetilde{\bA}_{\bz})\bu\big\|_2 \le \sup_{\bu\in \Sigma^{n,*}_{2s}}\big\|(\bA_{\breve{\bz}}-\bA_{\bz})\bu\big\|_2  \nn \\
            &=\sup_{\bu\in \Sigma^{n,*}_{2s}}\|\bA_{\bzeta} \bu\|_2\le \sup_{\bu\in\Sigma^{n,*}_{2s}}\frac{|\Re(\bzeta^*\bPhi\bu)|}{\kappa m} + \sup_{\bu\in\Sigma^{n,*}_{2s}}\frac{\|\Im(\diag(\bzeta^*)\bPhi\bu)\|_2}{\sqrt{m}} \label{bytraingle}
        \end{align}
        In the proofs of Theorems \ref{thm:pre-noise} and \ref{thm:zetabound}, we have shown   the bound \begin{align}
           \sup_{\bu\in\Sigma^{n,*}_{2s}}\frac{|\Re(\bzeta^*\bPhi\bu)|}{\kappa m} + \sup_{\bu\in\Sigma^{n,*}_{2s}}\frac{\|\Im(\diag(\bzeta^*)\bPhi\bu)\|_2}{\sqrt{m}} = O\left(\sqrt{\frac{s\log(en/s)}{m}}+\sqrt{\zeta_0\log({e}/{\zeta_0})}\right)\label{sparsecorrbound}
        \end{align} with probability at least $1-2\exp(-c_5\zeta_0\log(\frac{e}{\zeta_0})m)$; see Equations (\ref{29bound})--(\ref{tildetau2bound}). Thus, $\sup_{\bu\in\calU_c^{\mathbb{S}}}\big\|(\widetilde{\bA}_{\breve{\bz}}-\widetilde{\bA}_{\bz})\bu\big\|_2$ is small enough under $m=\Omega(s\log(\frac{en}{s}))$ and small enough $\zeta_0$, and therefore   $\widetilde{\bA}_{\breve{\bz}}\sim\rip(\calU_c,\frac{1}{3})$ holds with promised probability.

        By Proposition \ref{blockpro} and the observation made in Remark \ref{rem:GTT}, we have $\hat{\bx}_e = \bx^{\star\star} = \frac{\kappa m\cdot\bx}{\Re(\breve{\bz}^*\bPhi\bx)}$. To show $\bx_e^\sharp =\bx$, it remains to show $\frac{\Re(\breve{\bz}^*\bPhi\bx)}{\kappa m}>0$. This can be seen from 
        \begin{align*}
            \frac{\Re(\breve{\bz}^*\bPhi\bx)}{\kappa m} = \frac{\Re(\bz^*\bPhi\bx)+\Re(\bzeta^*\bPhi\bx)}{\kappa m} \ge \frac{\|\bPhi\bx\|_1}{\kappa m} - \|\bA_{\bzeta}\bx\|_2 \ge \frac{1}{2} - \sup_{\bu\in \Sigma^{n,*}_{2s}}\|\bA_{\bzeta} \bu\|_2 \ge \frac{1}{4}, 
        \end{align*}
        where we use $\frac{|\Re(\bzeta^*\bPhi\bx)|}{\kappa m}\le \|\bA_{\bzeta}\bx\|_2$, $\frac{\|\bPhi\bx\|_1}{\kappa m}\ge \frac{1}{2}$ that holds
         with the promised probability due to  \rev{ Equation (\ref{l1l2fixedpoint})}, and  $\sup_{\bu\in \Sigma^{n,*}_{2s}}\|\bA_{\bzeta} \bu\|_2\le \frac{1}{4}$ that follows from \rev{Equations (\ref{bytraingle})--(\ref{sparsecorrbound})}.  
\end{proof}

\section{Robust Instance Optimality}\label{sec:simulation}

The main aim of this section is to consolidate our prior results to show that $\bx^\sharp$ is robust and instance optimal over the entire signal space $\mathbb{S}^{n-1}$. We consider the noisy phase-only observations
\begin{align*}
    \breve{\bz} = \sign(\bPhi\bx + \btau_{(1)}+\bzeta_{(1)}) + \btau_{(2)}+\bzeta_{(2)},
\end{align*}
 where the bounded dense noise vectors $\btau_{(1)},\btau_{(2)}\in\mathbb{C}^m$ satisfy
 $\|\btau_{(1)}\|_\infty\le\tau_0$ and $\|\btau_{(2)}\|_\infty\le\tau_0$, $\bzeta_{(1)},\bzeta_{(2)}\in \mathbb{C}^m$ are $(\zeta_0m)$-sparse, and   the post-sign corruption $\bzeta_{(2)}$ additionally satisfies $\|\bzeta_{(2)}\|_\infty\le 2$. Here, $(\btau_{(1)},\btau_{(2)},\bzeta_{(1)},\bzeta_{(2)})$ may be generated by an adversary and can depend on $(\bPhi,\bx)$.

 We first announce our result, which  is novel for nonlinear compressed sensing and closely resembles the standard guarantee in linear case (e.g., Proposition \ref{pro1}). It is also the formal version of the second informal statement provided in introduction.  

\begin{theorem}
    [Robust, instance optimal \& uniform recovery] \label{thm:final}Consider the above setting with $\max\{\tau_0,\zeta_0\}\le c_0$ for some \rev{sufficiently small absolute constant $c_0$}, and the estimator $\bx^\sharp = \frac{\hat{\bx}}{\|\hat{\bx}\|_2}$ where $\hat{\bx}$ is solved from $\Delta(\bA_{\breve{\bz}};\be_1;\varepsilon)$ in (\ref{eq:nbp}) with   \begin{align}\label{epichoicec}
        \varepsilon = C_1\tau_0+\left(C_2\sqrt{\zeta_0\log({e}/{\zeta_0})}+C_3\sqrt{\frac{s\log({en}/{s})}{m}}\right)\mathbbm{1}(\zeta_0>0) 
    \end{align}
    for some large enough absolute constants $C_1,C_2,C_3$. If $m\ge C_4s\log(\frac{en}{s})$ for a large enough absolute constant $C_4$, then \rev{with  probability at least $1-C_5\exp(-c_6m)-C_7\exp(-c_8\zeta_0m\log(\frac{e}{\zeta_0}))\mathbbm{1}(\zeta_0>0)$}, we have 
    \begin{align}
        \|\bx^\sharp-\bx\|_2\le \frac{10\sigma_{\ell_1}(\bx,\Sigma^n_s)}{\sqrt{s}} + 15\varepsilon,\qquad\forall \bx\in \mathbb{S}^{n-1}.  \label{777}
    \end{align} 
\end{theorem}
By similar techniques, along with slightly more work, one can prove that $\bx^\sharp_e$ in Theorem \ref{thm:perfect} satisfies (\ref{777}) with $\varepsilon = O(\tau_0)$ and therefore the error increments of sparse corruption eliminated, up to changes of constants. We omit the details to avoid repetition.
\subsection{Proof Strategy} 
\rev{Without loss of generality, we  assume that the supports of $\bzeta_{(1)}$ and $\bzeta_{(2)}$ are disjoint. By Remark \ref{prepostcorr}, we can consolidate $\bzeta_{(1)}$ and $\bzeta_{(2)}$ to rewrite $\breve{\bz}$ as}
 \begin{align*}
     \breve{\bz} = \sign(\bPhi\bx+\btau_{(1)}) + \btau_{(2)} + \bzeta , 
 \end{align*}
  where $\bzeta\in \mathbb{C}^m$ satisfies $\|\bzeta\|_0\le 2\zeta_0m$ and $\|\bzeta\|_\infty\le 2$. Further, we let $\tilde{\btau}:=\sign(\bPhi\bx+\btau_{(1)})-\sign(\bPhi\bx)$ and write
  \begin{align}\label{expresszz}
      \breve{\bz} = \bz + \btau_{(2)} + \tilde{\btau} + \bzeta.
  \end{align}
  As in the proof of Theorem \ref{thm:io_noiseless}, we  restrict our attention to $\bx\in \mathbb{B}_1^n(\sqrt{2s})\cap\mathbb{S}^{n-1}:=\calX$. 

\rev{
We now make the following claim. 
\begin{claim}
    To prove Theorem \ref{thm:final}, it is sufficient to show that 
     \begin{gather}\label{uniformprove1}
      \sup_{\bx\in\calX}\frac{\kappa m}{\|\bPhi\bx\|_1}\Big[\frac{|\Re((\breve{\bz}-\bz)^*\bPhi\bx)|}{\kappa m}+\frac{\|\Im(\diag((\breve{\bz}-\bz)^*)\bPhi\bx)\|_2}{\sqrt{m}}\Big]\le\varepsilon\,~\textrm{ for $\varepsilon$ in (\ref{epichoicec}),}\\\label{uniformprove2}
      \sup_{\bx\in\calX}\sup_{\bu\in\Sigma^{n,*}_{2s}}\frac{|\Re((\breve{\bz}-\bz)^*\bPhi\bu)|}{\kappa m}+\sup_{\bx\in\calX}\sup_{\bu\in\Sigma^{n,*}_{2s}}\frac{\|\Im(\diag((\breve{\bz}-\bz)^*)\bPhi\bu)\|_2}{\sqrt{m}}<\frac{1}{9},\\
     \frac{13}{14}\le \frac{\|\bPhi\bx\|_1}{\kappa m }\le \frac{15}{14}\,,\quad \forall \bx\in\calX \label{l1l2}
 \end{gather}
  with probability at least $1-C_1\exp(-c_2m)-C_3\exp(-c_4\zeta_0m\log(\frac{e}{\zeta_0}))\mathbbm{1}(\zeta_0>0)$. 
\end{claim} 
\begin{proof}
To see this claim, we note that \eqref{uniformprove1}--\eqref{uniformprove2} strengthen the conditions in Lemma~\ref{lem:specializede} by making them uniform over $\bx \in \calX$. Since the RIP of $\bA_{\bz}=\bA_{\sign(\bPhi\bx)}$ in Corollary \ref{coro1} used in Equation (\ref{ripaz}) is uniform over all $\bx\in\calX$ with  probability at least $1-C\exp(-cm)$, the argument of Lemma \ref{lem:specializede} applies to all $\bx\in\calX$ and yields
\begin{align*}
    \|\bx^\sharp-\bx\|_2\le \frac{10\sigma_{\ell_1}(\bx;\Sigma^n_s)}{\sqrt{s}} + \frac{14\|\bPhi\bx\|_1\varepsilon}{\kappa m}\,,\quad \forall \bx\in\calX. 
\end{align*}
(Note that the term $\frac{10\sigma_{\ell_1}(\bx;\Sigma^n_s)}{\sqrt{s}}$ appears since $\bx$ may not be in $\Sigma^{n,*}_s$.) Combining with (\ref{l1l2}) yields (\ref{777}).     
\end{proof}
}

\rev{
All that remains is to show (\ref{uniformprove1})--(\ref{l1l2}) separately.
To this end, the arguments in Theorems \ref{thm:noisy_sparse}, \ref{thm:pre-noise}, and~\ref{thm:zetabound} must be revisited, together with additional techniques to ensure uniformity over all $\bx \in \calX$. To avoid repetition, we only outline the proof strategy and highlight the additional techniques.}

 \textbf{Showing (\ref{l1l2}).} Previously, we use (\ref{l1l2fixedpoint}) to achieve this for a fixed $\bx$.  
 Here, we need to strengthen it to  uniform concentration over all $\bx\in \calX$. Under $m=\Omega(s\log(\frac{en}{s}))$, the following lemma shows that $\sup_{\bx\in\calX}|\frac{\|\bPhi\bx\|_1}{\kappa m}-1|\le \frac{1}{14}$, which  exactly gives (\ref{l1l2}), with probability at least $1-2\exp(-cm)$. Up to some simple modifications (from $\mathbb{R}$ to $\mathbb{C}$), the proof is identical to that of \cite[Lem. 2.1]{plan2014dimension}, and is hence omitted.

\begin{lem}\label{lem:l1l2_rip}
    Suppose that the entries of $\bPhi\in\mathbb{C}^{m\times n}$ are drawn i.i.d.~from $\calN(0,1)+\calN(0,1)\bi$ and let $\kappa=\sqrt{\frac{\pi}{2}}$.  Given   $\calT\subset \mathbb{S}^{n-1}$,  for some absolute constants $C_1,c_2$ the event  $$\mathbbm{P}\left(\sup_{\bu\in\calT}\Big|\frac{\|\bPhi\bu\|_1}{\kappa m}- 1\Big|\leq \frac{C_1(\omega(\calT)+t)}{\sqrt{m}}\right)\ge 1-2\exp(-c_2t^2).$$
\end{lem}

\textbf{Showing (\ref{uniformprove2}).} By (\ref{expresszz}) and triangle inequality, it suffices to show, for $\bw= \btau_{(2)},\tilde{\btau},\bzeta$, that 
\begin{align}\label{48w}
    \sup_{\bx\in\calX}\sup_{\bu\in\Sigma^{n,*}_{2s}}\frac{|\Re(\bw^*\bPhi\bu)|}{\kappa m}+\sup_{\bx\in\calX}\sup_{\bu\in\Sigma^{n,*}_{2s}} \frac{\|\Im(\diag(\bw^*)\bPhi\bu)\|_2}{\sqrt{m}} <\frac{1}{27}.
\end{align}
For $\bw = \btau_{(2)}$ satisfying $\|\btau_{(2)}\|_\infty\le\tau_0$, re-iterating the arguments in Equations (\ref{eq18})--(\ref{add22}) shows that the left-hand side of (\ref{48w}) is bounded by $O(\tau_0)$ and thus smaller than $\frac{1}{27}$ under small enough $\tau_0$, with probability at least $1-C\exp(-cm)$. For $\bw = \bzeta$ obeying $\|\bzeta\|_0\le 2\zeta_0m$ and $\|\bzeta\|_\infty\le 2$, the arguments in Equations (\ref{29bound})--(\ref{tildetau2bound}) are able to show that the left-hand side of (\ref{48w}) is bounded by $O(\sqrt{\frac{s}{m}\log(\frac{en}{s})})+ \sqrt{\zeta_0\log(1/\zeta_0)})$ and hence is smaller than $\frac{1}{27}$ under   $m=\Omega(s\log(\frac{en}{s}))$ and small enough $\zeta_0$, with probability at least  $1-2\exp(-c \zeta_0m\log(e/\zeta_0))\mathbbm{1}(\zeta_0>0)$; here, the indicator function can be added since the case of $\zeta_0=0$, which means $\bzeta = 0$, is trivial.

It is more tricky to establish (\ref{48w}) for $\bw=\tilde{\btau}=\sign(\bPhi\bx+\btau_{(1)})-\sign(\bPhi\bx)$ where $\|\btau_{(1)}\|_\infty\le\tau_0$. To this end, we first revisit the argument in Theorem \ref{thm:pre-noise}: Define $\eta_0 = C^*c_0$ for some suitably large absolute constant $C^*$ and the index set $\calJ_{\bx,\eta_0}:=\{i\in[m]:|\bPhi_i^*\bx|\le \eta_0\}$, the key ingredient  is a decomposition $\tilde{\btau} = \tilde{\btau}_1+\tilde{\btau}_2$ where $\tilde{\btau}_1$ and $\tilde{\btau}_2$ accommodate the entries in $\calJ_{\bx,\eta_0}^c$ and $\calJ_{\bx,\eta_0}$, respectively; note that the entries of $\tilde{\btau}_1$ is bounded by $\frac{2\tau_0}{\eta_0}\le \frac{2c_0}{C^*c_0}=\frac{2}{C^*}$, which is suitably small under a large $C^*$, and thus $\tilde{\btau}_1$ can be treated similarly to $\btau_{(2)}$ by the arguments in Equations (\ref{eq18})--(\ref{add22}); under the high-probability event $\{|\calJ_{\bx,\eta_0}|\le \eta_0m\}$ (due to Chernoff bound), $\tilde{\btau}_2$ satisfies $\|\tilde{\btau}_2\|_\infty\le 2$ and $\|\tilde{\btau}_2\|_0\le \eta_0m$; since $\eta_0=C^*c_0$ is small enough (under absolute constant $C^*$ and small enough $c_0$), $\tilde{\btau}_2$ can be treated similarly to $\bzeta$ via the arguments in Equations (\ref{29bound})--(\ref{tildetau2bound}). The proof strategy remains valid, with the only gap being that the event $\{|\calJ_{\bx,\eta_0}|\le \eta_0m\}$ for a fixed $\bx$ needs to be strengthened to $\{\sup_{\bx\in\calX}|\calJ_{\bx,\eta_0}|\le \eta_0m\}$, in order to guarantee that the decomposition $\tilde{\btau}=\tilde{\btau}_1+\tilde{\btau}_2$ for some $(\eta_0m)$-sparse $\tilde{\btau}_2$ exists for all $\bx\in\calX$. Under $m=\Omega(s\log(\frac{en}{s}))$, the following lemma shows that $\sup_{\bx\in\calX}|\calJ_{\bx,\eta_0}|\le \eta_0m$ holds with probability at least $1-3\exp(-cm)$.  See a more refined statement and the proof in Appendix \ref{Jxbound}.

  \begin{lem} \label{lem:improved_lem9}
   Given some sufficiently small $\eta \in [\frac{C_1}{m},1)$ for some absolute constant $C_1$  and some  $\calK\subset \mathbb{S}^{n-1}$, if for sufficiently large $C_2$ we have $m\ge C_2\eta^{-3}\log(\eta^{-1})\omega^2(\calK)$,
    then  we have $$\mathbbm{P}\left(\sup_{\bx\in\calK}~|\calJ_{\bx,\eta}|\leq  \eta m\right)\ge 1-3\exp(-c_3\eta m).$$ 
\end{lem}

\textbf{Showing (\ref{uniformprove1}).} By (\ref{expresszz}), triangle inequality, $\sup_{\bx\in\calX}\frac{\kappa m}{\|\bPhi\bx\|_1}\le \frac{14}{13}$ from (\ref{l1l2}) and $\varepsilon$ in (\ref{epichoicec}), it suffices to show
\begin{align}
    \sup_{\bx\in\calX}\frac{|\Re(\bw^*\bPhi\bx)|}{\kappa m} +\sup_{\bx\in\calX} \frac{\|\Im(\diag(\bw^*)\bPhi\bx)\|_2}{\sqrt{m}}= O \left(\tau_0+\left(\sqrt{\zeta_0\log(e/\zeta_0)}+\sqrt{\frac{s\log(en/s)}{m}}\right)\mathbbm{1}(\zeta_0>0)\right) \label{wrip}
\end{align}
for $\bw = \btau_{(2)}$, $\tilde{\btau}$ and $\bzeta$. For $\bw=\btau_{(2)}$ obeying $\|\btau_{(2)}\|_\infty\le \tau_0$,  we use (\ref{Oonebound}) and (\ref{l1l2}) to achieve 
\begin{align*}
    \sup_{\bx\in\calX}\frac{|\Re(\btau_{(2)}^*\bPhi\bx)|}{\kappa m} +\sup_{\bx\in\calX} \frac{\|\Im(\diag(\btau_{(2)}^*)\bPhi\bx)\|_2}{\sqrt{m}}\le \tau_0\left(\sup_{\bx\in\calX}\frac{\|\bPhi\bx\|_1}{\kappa m}+ \sup_{\bx\in\calX}\frac{\|\bPhi\bx\|_2}{\sqrt{m}}\right)=O(\tau_0),
\end{align*}
with probability at least $1-C\exp(-cm)$. 
  For $\bw=\tilde{\btau}=\sign(\bPhi\bx+\btau_{(1)})-\sign(\bPhi\bx)$, our original calculations in Equations (\ref{surpr1})--(\ref{surpr3}) suffice to show that the left-hand side of (\ref{wrip}) is bounded by $O(\tau_0)$. When showing (\ref{wrip}) with $\bw=\bzeta$, however, an additional term $\sqrt{\frac{s}{m}\log(\frac{en}{s})}\mathbbm{1}(\zeta_0>0)$ arises compared to the non-uniform analysis in Theorem \ref{thm:zetabound}: Unlike in  Equations (\ref{43bound})--(\ref{45bound}) where $\bx$ is fixed,  here we seek a uniform bound on $\bx\in\calX$ and use Lemma \ref{lem:max_k_sum} to obtain  
\begin{align*}
     &\sup_{\bx\in\calX}\frac{|\Re(\bzeta^*\bPhi\bx)|}{\kappa m} +\sup_{\bx\in\calX} \frac{\|\Im(\diag(\bzeta^*)\bPhi\bx)\|_2}{\sqrt{m}}\\&\le C_1 \sup_{\bx\in\calX} \max_{\substack{I\subset[m]\\|I|=2\zeta_0m}}\Big(\frac{1}{m}\sum_{i\in I}|\bPhi_i^*\bx|^2\Big)^{1/2}\\
     &\le C_2\sqrt{\zeta_0\log(e/\zeta_0)} + C_3\sqrt{\frac{s\log(en/s)}{m}}\mathbbm{1}(\zeta_0>0) 
\end{align*} 
with probability at least $1-2\exp(-c\zeta_0m\log(e/\zeta_0))\mathbbm{1}(\zeta_0>0)$. The proof is complete. 

\section{Conclusion}\label{sec:concluding}
In this paper, we analyzed the instance optimality and robustness of the recently proposed linearization approach for PO-CS \cite{jacques2021importance}, in which one reformulates PO-CS as linear compressed sensing and then solves it via quadratically constrained basis pursuit. We improved the nonuniform instance optimality in \cite{jacques2021importance} to a uniform one over the entire sphere. The new technical tool is the RIP for all the new sensing matrices corresponding to an {\it arbitrary} set of signals in the unit sphere, which we proved by making important improvements on the arguments in \cite{chen2023uniform}.

Beyond Theorem \ref{thm:noisy_sparse} known from \cite{jacques2021importance}, we provided a new set of  robustness results. First, dense noise bounded by small enough $\tau_0$ (either before or after taking the phases) increments the estimation error by $O(\tau_0)$, and no algorithm can do substantially better than this. Second, an adversarial $\zeta_0$-fraction of sparse  corruption increments the error by $\tilde{O}(\sqrt{\zeta_0})$. We conjectured that this is tight for our specific estimator and provided some evidence. Yet we showed that it can be improved to $0$ by  proposing an extended linearization approach which perfectly recovers sparse signal under sparse corruption.

We believe the following questions are interesting for future study: 
\begin{itemize}
    \item {\it Non-Gaussian sensing matrix.} All existing recovery guarantees (that are exact in noiseless case) are built upon complex Gaussian $\bPhi$. Can we develop similar results for sub-Gaussian matrices or structured sensing matrices? 
    \item {\it New algorithms \& RIPless analysis.} Existing works  are building on the same linearization approach and similar RIP analysis. Can we develop new algorithms with comparable theoretical guarantee for PO-CS? Without linearization, can we directly analyze the original nonlinear phase-only observations? 
    \item {\it Instance optimality in nonlinear sensing.} Are there similar instance optimal results in other nonlinear sensing problems?  
\end{itemize}

\subsubsection*{Acknowledgement}
  The authors would like to thank Zhaoqiang Liu for his helps and discussions concerning the numerical simulations. Junren Chen also thanks Laurent Jacques for insightful discussions on instance optimality and phase-only compressive sensing. 
\subsubsection*{Funding}
J. Chen is supported by a Novikov Postdoctoral Fellowship from the Department of Mathematics at University of Maryland. M. K. Ng is supported by the
GDSTC: Guangdong and Hong Kong Universities “1+1+1” Joint Research Collaboration Scheme project No.: 2025A0505000007, National Key Research and Development Program of China under Grant 2024YFE0202900, RGC GRF 12300125. J. Scarlett is supported by the National University of Singapore under the Presidential Young Professorship grant scheme.

\bibliographystyle{plain}
\bibliography{libr}

\newpage 
 
\begin{appendix}
\noindent{\centering \bf\LARGE Appendix\\}

\section{RIP for New Sensing Matrices (Theorems \ref{thm:urip} \& \ref{thm:nonuni})}\label{sec:proofs}
Here we present the proofs for the RIP of  $\bA_{\bz}=\bA_{\sign(\bPhi\bx)}$ (Theorem \ref{thm:urip}) and  $\widetilde{\bA}_{\bz}=\widetilde{\bA}_{\sign(\bPhi\bx)}$ (Theorem \ref{thm:nonuni}). We prove the following more general statement that asserts  the RIP of $\widetilde{\bA}_{\bz}$ defined in (\ref{eq:tilde_Az}) over a subset of $\mathbb{S}^{n-1}$.

\begin{lem}\label{lem:RIPextend1} Suppose $\calU_c =\calU_1\times \calU_2$ for some cones $\calU_1\subset \mathbb{R}^n$ and $\calU_2\subset\mathbb{R}^m$. There exist some absolute constants $c_1,C_1,C_2,C_3,c_4$ such that
for any $\eta\in(0,c_1)$, if we let $r=\eta^2(\log(\eta^{-1}))^{1/2}$ and $\delta_\eta =C_1\eta(\log(\eta^{-1}))^{1/2}$, if   \begin{align}
    m\ge C_2 \left[ \frac{\omega^2(\calU_c^{\mathbb{S}})}{\eta^2\log(\eta^{-1})}+\frac{\scrH(\calK,\eta^3)}{\eta^2}+\frac{\omega^2(\calK_{(r)})}{\eta^4\log(\eta^{-1})}+\frac{\omega^2(\calK_{(\eta^3)})}{\eta^8\log(\eta^{-1})}\right],\label{samplelem4}
\end{align}   
then the event
\begin{align}\label{eq:tildeAzrip}
    \widetilde{\bA}_{\bz}= \widetilde{\bA}_{\sign(\bPhi\bx)}\sim \rip(\calU_c,\delta_\eta),\quad\forall\bx\in\calK  
\end{align}
holds with probability at least $1- C_3\exp(-c_4\eta^2 m)$.     
\end{lem}
We first show that this statement immediately leads to Theorem \ref{thm:urip}. 
\begin{proof}[Proof of Theorem \ref{thm:urip}] 
Setting    $\calU_1=\calU$ and $\calU_2=0$  in Lemma \ref{lem:RIPextend1} yields Theorem \ref{thm:urip}. 
\end{proof}
 In the remainder of this subsection, we will first establish a number of intermediate bounds, and then combine  them to prove Lemma \ref{lem:RIPextend1}. 

By homogeneity and some algebra, (\ref{eq:tildeAzrip}) is equivalent to  
\begin{align}\label{eq:extended_goal}
        \sup_{\bx\in \calK}\sup_{(\bu,\bw)\in  \calU_c ^{\mathbb{S}}} \Big|\underbrace{\frac{[\Re(\bz^*\bPhi)\bu]^2}{\kappa^2m^2}+ \Big\|\frac{\Im(\diag(\bz)^*\bPhi)\bu}{\sqrt{m}} +\bw\Big\|_2^2-1}_{:=f(\bx,\bu,\bw)}\Big|\le \delta _\eta.
\end{align}
Given the underlying signal $\bx\in \mathbb{S}^{n-1}$ and some $\bu\in\mathbb{R}^n$, we have the decomposition \begin{equation}
\label{eq:decompose}    \bu = \underbrace{\langle\bu,\bx\rangle \bx}_{:=\bu_{\bx}^\|} + \underbrace{(\bu - \langle\bu,\bx\rangle \bx)}_{:=\bu_{\bx}^\bot}
\end{equation}
where $
    \|\bu_{\bx}^\|\|_2^2+\|\bu_{\bx}^\bot\|_2^2=\|\bu\|_2^2.$ 
  We recall that $\kappa:=\mathbbm{E}|\Phi_{i,j}|=\sqrt{\frac{\pi}{2}}$ and $\bu_{\bx}^\bot = \bu - \langle\bu,\bx\rangle\bx$, and observe that \[1 = \|\bu\|_2^2+\|\bw\|_2^2= |\langle \bu,\bx\rangle|^2+ \|\bu_{\bx}^\bot\|_2^2 +\|\bw\|_2^2,\quad \forall (\bu,\bw)\in\calU_c^{\mathbb{S}}.\]   Hence,   we    decompose $  \sup_{\bx\in\calK}\sup_{(\bu,\bw)\in\calU_c^{\mathbb{S}}}|f(\bx,\bu,\bw)|$ into 
    \begin{align}\label{eq:divide_extend}
     \sup_{\bx\in\calK}\sup_{(\bu,\bw)\in\calU_c^{\mathbb{S}}}|f(\bx,\bu,\bw)|\leq \sup_{\bx\in\calK}\sup_{(\bu,\bw)\in\calU_c^{\mathbb{S}}}|f^\|(\bx,\bu)|+\sup_{\bx\in\calK}\sup_{(\bu,\bw)\in\calU_c^{\mathbb{S}}}|f^\bot(\bx,\bu,\bw)|
     \end{align}
     where 
     \begin{gather}
     f^\|(\bx,\bu)=\frac{[\Re(\bz^*\bPhi)\bu]^2}{\kappa^2m^2} - |\langle\bu ,\bx\rangle|^2  \label{eq:tildefpara}
    \\f^\bot(\bx,\bu,\bw) =\left\|\frac{\Im(\diag(\bz)^*\bPhi)\bu}{\sqrt{m}} +\bw\right\|_2^2 - \big(\|\bu_{\bx}^\bot\|_2^2 +\|\bw\|_2^2\big). \label{eq:tildefbot}
\end{gather}
We refer to (\ref{eq:tildefpara}) as the  parallel term and (\ref{eq:tildefbot}) as the orthogonal term. We first control $\sup_{(\bu,\bw)\in\calU_c^{\mathbb{S}}}f^{\|}(\bx,\bu)$ in Lemma \ref{lem:para_fix_x} and $\sup_{(\bu,\bw)\in\calU_c^{\mathbb{S}}}f^{\bot}(\bx,\bu)$ in Lemma \ref{lem:tildef_ortho_fixedx}, and then strengthen them to uniform bounds over $\bx\in \calK$ by covering arguments in Lemma \ref{lem:para_bound} and Lemma \ref{lem:tilde_ortho_unibound}, respectively.

\subsection{Technical Contributions}\label{apptechnical} 
Since the proof is lengthy and builds on existing work \cite{chen2023uniform}, we first pause to discuss the key differences and innovations. Specifically, some improvement is necessary to deal with arbitrary $\calK\subset\mathbb{S}^{n-1}$.  


{\bf Main Improvement --- Finer Treatment to Phase Perturbation via Introducing $\calI_{\bx-\bx_r,\eta'}$:} The most notable improvement is made when seeking a uniform bound on the orthogonal term (see our Lemma \ref{lem:tilde_ortho_unibound}). In the analysis, we need to bound
 \begin{align*}
     &\frac{1}{\sqrt{m}}\sup_{\bx\in\calK}\sup_{\bu\in\calU_1^{\mathbb{S}}}\big\|\Im\big[\diag(\overline{\sign(\bPhi\bx)-\sign(\bPhi\bx_r)})\bPhi\bu\big]\big\|_2\\
     &=  \sup_{\bx\in\calK}\sup_{\bu\in\calU_1^{\mathbb{S}}}\left(\frac{1}{m}\sum_{i=1}^m \Big[\Im\big([\overline{\sign(\bPhi_i^*\bx)-\sign(\bPhi_i^*\bx_r)}]\cdot \bPhi_i^*\bu\big)\Big]^2\right)^{1/2}
 \end{align*}
 where $\bx_r:=\mathrm{arg}\min_{\bu\in \calN_r}\|\bu-\bx\|_2$ and $\calN_r$  is a minimal $r$-net of $\calK$. To guarantee that this term is sufficiently small, the idea is to separate the $m$ measurements into two parts---a small ``problematic part'' $\mathcal{E}_{\bx}$ where $|\sign(\bPhi_i^*\bx)-\sign(\bPhi_i^*\bx_r)|$ is hard to control, and the ``major part'' $\mathcal{E}_{\bx}^c=[m]\setminus \calE_{\bx}$.

 In \cite{chen2023uniform}, the authors simply let $\mathcal{E}_{\bx}=\calJ_{\bx,\eta}=\{i\in[m]:|\bPhi_i^*\bx|\le \eta\}$  for some small $\eta>0$, and then apply 
 \begin{gather} \label{877}
     |\sign(\bPhi_i^*\bx)-\sign(\bPhi_i^*\bx_r)|\le 2 \qquad \text{for } i\in \mathcal{E}_{\bx},\\ \label{888}
     |\sign(\bPhi_i^*\bx)-\sign(\bPhi_i^*\bx_r)|\le  \frac{2|\bPhi_i^*(\bx-\bx_r)|}{\eta} \qquad \text{for }i\in \mathcal{E}_{\bx}^c,
 \end{gather} 
 where (\ref{888}) follows from (\ref{eq:sign_conti}).
  Thus, by applying     (\ref{888}) to the ``major part'' $i\in\calE_{\bx}^c$, they have to show that 
 \begin{align}\nn
     &\sup_{\bx\in\calK}\sup_{\bu\in\calU_1^{\mathbb{S}}}\left(\frac{1}{m}\sum_{i\notin \calJ_{\bx,\eta}}\Big[\Im\big([\overline{\sign(\bPhi_i^*\bx)-\sign(\bPhi_i^*\bx_r)}]\cdot \bPhi_i^*\bu\big)\Big]^2\right)^{1/2} \\\nn
     &\le \sup_{\bx\in\calK}\sup_{\bu\in\calU_1^{\mathbb{S}}} \frac{1}{\eta}\left(\frac{1}{m}\sum_{i=1}^m |\bPhi_i^*(\bx-\bx_r)| ^2 |\bPhi_i^*\bu|^2\right)^{1/2}\\&\le \sup_{\bv\in \calK_{(r)}}  \sup_{\bu\in\calU_1^{\mathbb{S}}} \frac{1}{\eta}\left(\frac{1}{m}\sum_{i=1}^m |\bPhi_i^*\bv|^2|\bPhi_i^*\bu|^2\right)^{1/2}.\label{tobound}
 \end{align}
  is small enough. However, this is a heavy-tailed random process that is in general hard to control.

 The argument in \cite{chen2023uniform} is to use {\it extremely small}   $r$ with $o(1)$ scaling to ensure that this term is small enough. Note that   \cite[Eqs. (III.63)--(III.64)]{chen2023uniform} essentially bounds (\ref{tobound}) as 
 \begin{align*}
     O\left(\sup_{\bv\in\calK_{(r)}}\|\bPhi\bv\|_\infty \cdot \sup_{\bu\in\calU_1^{\mathbb{S}}}\frac{\|\bPhi\bu\|_2}{\eta\sqrt{m}}\right) \stackrel{(\ref{Oonebound})}{=} O\left(\sup_{\bv\in\calK_{(r)}}\frac{\|\bPhi\bv\|_\infty}{\eta}\right). 
 \end{align*}
 By Gaussian concentration, it is easy to show   
 \begin{align}\label{lowerbound}
     \sup_{\bv\in\calK_{(r)}}\frac{\|\bPhi\bv\|_\infty}{\eta} \ge   \sup_{\bv\in\calK_{(r)}}\frac{\|\bPhi^\Re\bv\|_1}{\eta m} =\Omega\left(\frac{\omega(\calK_{(r)})}{\eta}\right)
 \end{align}
 with high probability, 
 and this is also an upper bound on $\sup_{\bv\in\calK_{(r)}}\eta^{-1}\|\bPhi\bv\|_\infty$, up to log factors, due to Lemma \ref{lem:max_k_sum}. Therefore, for $\calK=\Sigma^{n,*}_s$, the authors of \cite{chen2023uniform} take $r=\tilde{O}(\frac{\eta}{\sqrt{s}})$ to guarantee sufficiently small $$\frac{\omega(\calK_{(r)})}{\eta} = \Theta\left(\frac{r\sqrt{s\log(\frac{en}{s})}}{\eta}\right).$$ 
 Note that $m=\Omega(\scrH(\calK,r))$ is needed to support a union bound over $\calN_r$.
 Since $\scrH(\Sigma^{n,*}_s,r)$ logarithmically depends on $r $ (see Equation (\ref{entroSigmans})), the choice $r=\tilde{O}(\frac{\eta}{\sqrt{s}})$ only adds to log factors in the sample complexity. However, such small $r$ will significantly worsens the sample complexity for  $\calK=\sqrt{s}\mathbb{B}_1^n\cap \mathbb{S}^{n-1}$ from $\tilde{O}(s\log(\frac{en}{s}))$ to $\Omega(s^2\log(\frac{en}{s}))$, since its metric entropy quadratically depends on the covering radius; see Equation (\ref{entroeffspa}).

To make an improvement, we first notice that their choice of $\mathcal{E}_\bx$  and (\ref{877})--(\ref{888}) are suboptimal: for $i\in \mathcal{E}_{\bx}^c$, the bound $\frac{2|\bPhi_i^*(\bx-\bx_r)|}{\eta}$ they used could be worse than $2$ when $|\bPhi_i^*(\bx-\bx_r)|\ge \eta$, and indeed by (\ref{lowerbound}), $|\bPhi_i^*(\bx-\bx_r)|$  could reach $\Omega(\omega(\calK_{(r)}))$ for some $i$.\footnote{We note that $\omega(\calK_{(r)})=\Theta(r\sqrt{s\log(\frac{en}{s})})$ for $\calK=\Sigma^n_s$, while it even scales as $\omega(\calK_{(r)})=\Theta(r\sqrt{n})$ for $\calK=\sqrt{s}\mathbb{B}_1^n$ under $r\le\sqrt{s/n}$.} Moreover,   the heavy-tailed random process in Equation (\ref{tobound}) is also a consequence of $|\bPhi_i^*(\bx-\bx_r)|$ arising from (\ref{888}).

Our remedy is to use $|\sign(\bPhi_i^*\bx)-\sign(\bPhi_i^*\bx_r)|\le 2$ for the measurements with overly large $|\bPhi_i^*(\bx-\bx_r)|$. To  formalise this idea, for some small enough $\eta'>0$ to be chosen, we introduce 
\begin{align*}
    \calI_{\bx-\bx_r,\eta'}:=\big\{i\in[m]:|\bPhi_i^*(\bx-\bx_r)|>\eta'\big\}
\end{align*}
and define the set of problematic measurements as $\mathcal{E}_\bx := \calJ_{\bx,\eta}\cup \calI_{\bx-\bx_r,\eta'}$. We use $|\sign(\bPhi_i^*\bx)-\sign(\bPhi_i^*\bx_r)|\le 2$ for $i\in \calE_{\bx}$. This is valid as $|\mathcal{E}_\bx| $ remains  small: in fact, $|\calE_{\bx}|\le|\calJ_{\bx,\eta}|+|\calI_{\bx-\bx_r,\eta'}|$, and one can invoke Lemma \ref{lem:max_k_sum} to uniformly control $|\calI_{\bx-\bx_r,\eta'}|$ over $\bx-\bx_r\in \calK_{(r)}$. On the other hand, for $i\in\calE_{\bx}^c= \calJ_{\bx,\eta}^c\cap \calI_{\bx-\bx_r,\eta'}^c$, we now have a bound 
\[
     |\sign(\bPhi_i^*\bx)-\sign(\bPhi_i^*\bx_r)|\le  \frac{2}{\eta}|\bPhi_i^*(\bx-\bx_r)|\le  \frac{2\eta'}{\eta},\] which is strictly tighter than $2$ as long as $\eta'<\eta$, and we indeed set $\eta'\ll\eta$ in the proof. Meanwhile, this avoids the heavy-tailed random process in (\ref{tobound}).

{\bf Other Refinements:}  We briefly note that we have refined or simplified some steps. As an example, in contrast to the covering approach taken in \cite{chen2023uniform}, we directly use known concentration bounds to establish the uniformity over $\calU_c^{\mathbb{S}}$. Consequently,   the sample complexity (\ref{samplelem4}) is only based on the Gaussian width of $\calU_c^{\mathbb{S}}$ and is free of its metric entropy.  As another example, while \cite[Lem. 9]{chen2023uniform} seeks to uniformly bound $|\calJ_{\bx,\eta}|$ over $\bx\in\calK=\Sigma^{n,*}_s$, we find that bounding $|\calJ_{\bx,\eta}|$ over the $r$-net of $\calK$ is sufficient.

\subsection{Fixed-$\bx$ Bound on the Parallel Term} 
Observe that $(\bu,\bw)\in \calU_c ^{\mathbb{S}}$ implies  $\bu \in \calU_1\cap \mathbb{B}_2^n $, and that $f^\|(\bx,\bu)$ does not depend on $\bw$.  We start with
\begin{align}\label{541}
      \sup_{\bx\in \calK}\sup_{(\bu,\bw)\in  \calU_c ^{\mathbb{S}}}|f^\|(\bx,\bu)|\le \sup_{\bx\in\calK}\sup_{\bu\in\calU_1\cap \mathbb{B}_2^n} |f^\|(\bx,\bu)| = \sup_{\bx\in\calK}\sup_{\bu\in \calU_1^{\mathbb{S}}}|f^\|(\bx,\bu)|,
\end{align}
where   the last equality follows from $|f^\|(\bx,t\bu)|= t^2 |f^\|(\bx,\bu)|$ for $t>0$. 
By   $a^2-b^2 = (a-b)(a+b)$, 
    \begin{align}\nn
  \sup_{\bx\in\calK}\sup_{\bu\in \calU_1^{\mathbb{S}}}|f^\|(\bx,\bu)| &\leq \left(\sup_{\bx\in\calK}\sup_{\bu\in \calU_1^{\mathbb{S}}} \Big|\frac{\Re(\bz^{*}\bPhi)\bu}{\kappa m}-\langle\bu,\bx\rangle\Big|\right)\cdot \left(\sup_{\bx\in\calK}\sup_{\bu\in \calU_1^{\mathbb{S}}} \Big|\frac{\Re(\bz^*\bPhi)\bu}{\kappa m}+\langle\bu,\bx\rangle\Big|\right)\\
    &\le C_1\sup_{\bx\in\calK}\sup_{\bu\in \calU_1^{\mathbb{S}}} \Big|\frac{\Re(\bz^*\bPhi)\bu}{\kappa m}-\langle\bu,\bx\rangle\Big|, \label{eq:bound_large_1}
\end{align}
where (\ref{eq:bound_large_1}) follows from
    \begin{align}
    \sup_{\bx\in\calK}\sup_{\bu\in \calU_1^{\mathbb{S}}}\Big|\frac{\Re(\bz^*\bPhi)\bu}{\kappa m}+\langle \bu,\bx\rangle\Big| \leq  \sup_{\bx\in\calK}\sup_{\bu\in \calU_1^{\mathbb{S}}}\Big|\frac{\|\bz\|_2\cdot \|\bPhi\bu\|_2}{\kappa m}\Big|+1\leq \sup_{\bu\in\calU_1^{\mathbb{S}}} \frac{\|\bPhi\bu\|_2}{\kappa\sqrt{m}}+1 =O(1), \label{eq:use_k_m}
\end{align}
which holds with probability at least $1-4\exp(-cm)$ provided that $m=\Omega(\omega^2(\calU_1^{\mathbb{S}}))$; see (\ref{Oonebound}). Therefore,  we only need to bound the term in (\ref{eq:bound_large_1}), which reads as 
 \begin{align}
\sup_{\bx\in\calK}\sup_{\bu\in\calU_1^{\mathbb{S}}} \Big|\frac{1}{\kappa m}\sum_{i=1}^m \Re\big(\overline{\sign(\bPhi_i^*\bx)}\bPhi_i^*\bu\big)-\langle \bx,\bu\rangle\Big|:=\sup_{\bx\in\calK}\sup_{\bu\in\calU_1^{\mathbb{S}}}|f_1^\|(\bx,\bu)|. \label{eq:f1_parallel}
\end{align} 
We note that $f_1^\|(\bx,\bu)$ is zero-mean. To see this, we use the decomposition $\bu = \bu_{\bx}^\| + \bu_{\bx}^{\bot}$ from  (\ref{eq:decompose}), and then by the independence between $(\bPhi_i^*\bx,\bPhi_i^*\bu_{\bx}^\bot)$, 
\begin{align}
\label{zmparallel}\mathbbm{E}\big[\kappa^{-1}\Re\big(\overline{\sign(\bPhi_i^*\bx)}\bPhi_i^*\bu\big)\big] = \mathbbm{E}\big[\kappa^{-1}\Re\big(\overline{\sign(\bPhi_i^*\bx)}\bPhi_i^*\bu_{\bx}^\|\big)\big] = \langle\bx,\bu\rangle\kappa^{-1}\mathbbm{E}|\bPhi_i^*\bx|  = \langle\bx,\bu\rangle. 
\end{align}
Because $f_1^\|(\bx,\bu)$  is linear in $\bu$, $\sup_{\bu\in\calU_1^{\mathbb{S}}}|f_1^\|(\bx,\bu)|$ with  a  fixed $\bx\in \mathbb{S}^{n-1}$ can be  treated as the supremum of a standard random process, and hence be directly bounded by the following lemma. 
      \begin{lem}[See Sec. 8.6 in \cite{vershynin2018high}]
    \label{lem:tala}
  Let $(R_{\bu})_{\bu\in\calT}$ be a random process (not necessarily zero-mean) on a subset $\calT\subset \mathbb{R}^n$. Assume that $R_{0}=0$, and for all $\bu,\bv\in\calT\cup\{0\}$ we have $\|R_{\bu}-R_{\bv}\|_{\psi_2}\leq K\|\bu-\bv\|_2$. Then, for every $t\geq 0$, the event $$\sup_{\bu\in \calT}\big|R_{\bu}\big|\leq CK \big(\omega(\calT)+t\cdot \rad(\calT)\big)$$ 
   holds with probability at least $1-2\exp(-t^2)$.  
\end{lem}

We then use it to bound $\sup_{\bu\in\calU_1^\mathbb{S}}|f_1^\|(\bx,\bu)|$ for a fixed $\bx$. 
\begin{lem}\label{lem:para_fix_x}
  Consider $\sup_{\bu\in\calU_1^{\mathbb{S}}}f_1^\|(\bx,\bu)$ as in  (\ref{eq:f1_parallel}) with a fixed $\bx \in \mathbb{S}^{n-1}$ and $\calU_1^{\mathbb{S}}\subset \mathbb{S}^{n-1}$. Then for some absolute constant $C$ and any $t\geq 0$, we have 
    \begin{equation}\label{eq:bound_supu}
       \mathbbm{P}\left( \sup_{\bu\in\calU_1^{\mathbb{S}}} |f_1^\|(\bx,\bu)| \leq \frac{C[\omega(\calU_1^{\mathbb{S}})+t]}{\sqrt{m}}\right)\ge 1-2\exp(-t^2).
    \end{equation}
 
\end{lem}
 \begin{proof}
    For any $\bu_1,\bu_2\in \calU_1^{\mathbb{S}}\cup\{0\}$,  
        \begin{align}
          &\big\|f_1^\|(\bx,\bu_1)-f_1^\|(\bx,\bu_2)\big\|_{\psi_2}\nn
          \\&=\frac{1}{\kappa}\Big\|\frac{1}{m}\sum_{i=1}^m\Re\big(\overline{\sign(\bPhi_i^*\bx)}\bPhi_i^*(\bu_1-\bu_2)\big)-\kappa\langle\bx,\bu_1-\bu_2\rangle\Big\|_{\psi_2}\nn\\
        &\le \frac{C_1}{\sqrt{m}}\big\|\Re\big(\overline{\sign(\bPhi_i^*\bx)}\bPhi_i^*(\bu_1-\bu_2)\big)-\kappa\langle\bx,\bu_1-\bu_2\rangle\big\|_{\psi_2}\label{eq:psi2_sum}\\
        &\leq\frac{C_1}{\sqrt{m}}\Big(\big\|\Re\big(\overline{z_i}\bPhi_i^*(\bu_1-\bu_2)\big)\big\|_{\psi_2}+\big\|\kappa\langle\bx,\bu_1-\bu_2\rangle\big\|_{\psi_2}\Big),\nn
    \end{align}
    where in (\ref{eq:psi2_sum}) we note $\mathbbm{E}[\Re\big(\overline{\sign(\bPhi_i^*\bx)}\bPhi_i^*(\bu_1-\bu_2)\big)]=\kappa\langle\bx,\bu_1-\bu_2\rangle$   (e.g., see (\ref{zmparallel})) and then apply (\ref{sgsum}). Moreover, we note that $\|\Re(\overline{z_i}\bPhi_i^*(\bu_1-\bu_2))\|_{\psi_2}\le C_2\|\bu_1-\bu_2\|_2$ because $$\|\Re(z_i\bPhi_i)\|_{\psi_2}\leq \|z_i^\Re \bPhi_i^\Re\|_{\psi_2}+\|z_i^\Im \bPhi_i^\Im\|_{\psi_2}\le \|\bPhi_i^\Re\|_{\psi_2}+\|\bPhi_i^\Im\|_{\psi_2}=O(1).$$ In addition, the simple upper bound $|\kappa\langle\bx,\bu_1-\bu_2\rangle| \leq \kappa\|\bu_1-\bu_2\|_{2}$ implies $\|\kappa\langle\bx,\bu_1-\bu_2\rangle\|_{\psi_2}\le C_3 \|\bu_1-\bu_2\|_2.$ Therefore, we have shown $$\|f_1^\|(\bx,\bu_1)-f_1^\|(\bx,\bu_2)\|_{\psi_2}\leq \frac{C_4}{\sqrt{m}}\|\bu_1-\bu_2\|_2$$ for some absolute constant $C_4$. Now we invoke Lemma \ref{lem:tala} to obtain   (\ref{eq:bound_supu}), as desired.  
\end{proof}

\subsection{Fixed-$\bx$ Bound on the Orthogonal Term} 
To deal with the orthogonal term $f^\bot(\bx,\bu,\bw)$ defined in (\ref{eq:tildefbot}), we   begin with 
\begin{align}\nn
      \sup_{\bx\in \calK}\sup_{(\bu,\bw)\in  \calU_c ^{\mathbb{S}}}|f^\bot(\bx,\bu,\bw)| &\le \left(  \sup_{\bx\in \calK}\sup_{(\bu,\bw)\in  \calU_c ^{\mathbb{S}}} \Big|\Big\|\frac{\Im(\diag(\bz)^*\bPhi)\bu}{\sqrt{m}}+\bw\Big\|_2-\sqrt{\|\bu_{\bx}^\bot\|_2^2+\|\bw\|_2^2}\Big|\right) \\&\cdot\left(  \sup_{\bx\in \calK}\sup_{(\bu,\bw)\in  \calU_c ^{\mathbb{S}}} \Big|\Big\|\frac{\Im(\diag(\bz)^*\bPhi)\bu}{\sqrt{m}}+\bw\Big\|_2+\sqrt{\|\bu_{\bx}^\bot\|_2^2+\|\bw\|_2^2}\Big|\right)\nn\\
    &\le C_1  \sup_{\bx\in \calK}\sup_{(\bu,\bw)\in  \calU_c ^{\mathbb{S}}} \Big|\underbrace{\Big\|\frac{\Im(\diag(\bz)^*\bPhi)\bu}{\sqrt{m}}+\bw\Big\|_2-\sqrt{\|\bu_{\bx}^\bot\|_2^2+\|\bw\|_2^2}}_{:=f_1^\bot(\bx,\bu,\bw)}\Big|, \label{eq:tildef1bot}
\end{align}
where the last inequality follows from $\|\bu_\bx^\bot\|_2^2+\|\bw\|_2^2\le 1$ and 
\begin{align*}     \sup_{\bx\in \calK}\sup_{(\bu,\bw)\in  \calU_c ^{\mathbb{S}}}\Big\|\frac{\Im(\diag(\bz)^*\bPhi)\bu}{\sqrt{m}}+\bw\Big\|_2 &\le \sup_{\bu \in \calU_1^{*}}\frac{\|\bPhi\bu\|_2}{\sqrt{m}}+1   =O(1),
\end{align*}
which holds with probability at least $1-4\exp(-cm)$ under $m=\Omega(\omega^2(\calU_1^{\mathbb{S}}))$; see (\ref{Oonebound}).

It remains to bound $  \sup_{\bx\in \calK}\sup_{(\bu,\bw)\in  \calU_c ^{\mathbb{S}}}|f_1^\bot(\bx,\bu,\bw)|$, where $f_1^\bot(\bx,\bu,\bw)$ is defined in (\ref{eq:tildef1bot}). We first bound $\sup_{(\bu,\bw)\in  \calU_c ^{\mathbb{S}}}|f_1^\bot(\bx,\bu,\bw)|$ for a fixed $\bx$ by using the following extended matrix deviation inequality, along with the rotational invariance of $\bPhi$. 
\begin{lem}[see \cite{chen2018stable}]\label{lem:extended}
  Let $\bA$ be an $m\times n$ matrix whose rows $\ba_i$ are independent centered isotropic sub-Gaussian vectors in $\mathbb{R}^n$. Given any bounded subset $\calT \subset \mathbb{R}^n\times \mathbb{R}^m$ and $t\ge 0$, the event 
    \begin{align*}
        \sup_{(\bu,\bw)\in \calT}\Big|\big\|\bA \bu+\sqrt{m}\bw\big\|_2-\sqrt{m}\cdot \sqrt{\|\bu\|_2^2+\|\bw\|_2^2}\Big| \le CK^2\big(\omega(\calT)+t\cdot\rad(\calT)\big)
    \end{align*}
    holds with probability at least $1-\exp(-t^2)$, where $K= \max_i\|\ba_i\|_{\psi_2}$.
\end{lem}

\begin{lem}
\label{lem:tildef_ortho_fixedx}
Let $\bx\in \mathbb{S}^{n-1}$ be fixed and $f_1^\bot(\bx,\bu,\bw)$ be given in (\ref{eq:tildef1bot}). Then, for any $t\ge0$, we have 
    \begin{align*}
        \mathbbm{P}\left(\sup_{(\bu,\bw)\in \calU_c^{\mathbb{S}}}|f_1^\bot(\bx,\bu,\bw)|\le \frac{C[\omega(\calU_c^{\mathbb{S}})+t]}{\sqrt{m}}\right)\ge 1-\exp(-t^2). 
    \end{align*}
 \end{lem}
\begin{proof}
   We find an orthogonal matrix $\bP$ such that $\bP\bx=\be_1$ and consider $\tilde{\bPhi}=\bPhi\bP^\top$, which has the same distribution as $\bPhi$. Due to $\Im(\diag(\bz)^*\bPhi\bx)=0$ and (\ref{eq:decompose}),  \[\Im(\diag(\bz)^*\bPhi)\bu=\Im(\diag(\bz)^*\bPhi)\bu_{\bx}^\bot.\] In view of (\ref{eq:tildef1bot}), $\tilde{\bPhi}=\bPhi\bP^\top$ and $\bP\bx=\be_1$ 
    \begin{align}\label{rewritebot}
        f_1^\bot(\bx,\bu,\bw)= \left\|\frac{\Im(\diag(\sign(\tilde{\bPhi}\be_1))^*)\tilde{\bPhi})\bP\bu_\bx^\bot}{\sqrt{m}}+\bw\right\|_2- \sqrt{\|\bu_\bx^\bot\|_2^2+\|\bw\|_2^2}~.
    \end{align}
   By $\bP\bu_\bx^\bot = \bP(\bI_n-\bx\bx^\top)\bu = \bP(\bI_n-\bx\bx^\top) \bP^\top\bP \bu = (\bI_n - \be_1\be_1^\top)  \bP\bu$, we  let $\tilde{\bu}:=\bP\bu_\bx^\bot$ and know that its first entry is always $0$. We further observe that $(\bu,\bw)\in\calU_c^{\mathbb{S}}$ implies 
   \begin{align}\label{tildeU0construct}
       \begin{bmatrix}
           \tilde{\bu}\\\bw
       \end{bmatrix} = \begin{bmatrix}
           \bP(\bI_n-\bx\bx^\top)& 0 \\
           0 & \bI_m
       \end{bmatrix} \begin{bmatrix}
          \bu \\\bw
       \end{bmatrix}\in \begin{bmatrix}
           \bP(\bI_n-\bx\bx^\top)& 0 \\
           0 & \bI_m
       \end{bmatrix}\calU_c^{\mathbb{S}}:= \tilde{\calU}_0.
   \end{align}
   Let $\tilde{\bu}_1 \in \mathbb{R}^{n-1}$ denote the subvector of $\tilde{\bu} \in \mathbb{R}^n$ consisting of its last $n-1$ entries, and let $\tilde{\calU} \subset \mathbb{R}^{m+n-1}$ denote the projection of $\tilde{\calU}_0 \subset \mathbb{R}^{m+n}$ onto its last $m+n-1$ coordinates, that is, 
\begin{align}
    \tilde{\calU}=\bR\tilde{\calU}_0,\quad\text{where }\,\bR=[0,\bI_{m+n-1}]\in \mathbb{R}^{(m+n-1)\times (m+n)}.\label{tildeUconstruct}
\end{align}
Note  that we have $(\tilde{\bu}_1^\top,\bw^\top)^\top\in \tilde{\calU}$. By (\ref{tildeU0construct}),
  (\ref{tildeUconstruct}), and \cite[Exercise 7.5.4]{vershynin2018high}, we have
  \begin{gather}\label{gwlink}
      \omega(\tilde{\calU})\le \omega(\tilde{\calU}_0)\le \omega(\calU_c^{\mathbb{S}}),\\ \label{radiuslink}
      \rad(\tilde{\calU})\le \rad(\tilde{\calU}_0)\le\rad(\calU_c^{\mathbb{S}})\le1.
  \end{gather}

   With these preparations, we let $\tilde{\bPhi}=[\tilde{\bPhi}^{[1]},\tilde{\bPhi}^{[2:n]}]$ where $\tilde{\bPhi}^{[1]}$ is the first column of $\tilde{\bPhi}$ and $\tilde{\bPhi}^{[2:n]}\in \mathbb{R}^{m\times (n-1)}$ contains the last $n-1$ columns of $\tilde{\bPhi}$, then continuing from (\ref{rewritebot}) and using the fact that the first entry of $\tilde{\bu}=\bP\bu_{\bx}^\bot$ is $0$, and the construction that $\tilde{\bu}_1$ is the subvector of $\tilde{\bu}$ consisting of the last $n-1$ entries,   
       \begin{align}\nn
        & \sup_{(\bu,\bw)\in\calU_c^{\mathbb{S}}}|f_1^\bot(\bx,\bu,\bw)|\\&\nn = \sup_{(\bu,\bw)\in\calU_c^{\mathbb{S}}}\Big|\Big\|\frac{\Im(\diag(\sign(\tilde{\bPhi}\be_1))^*\tilde{\bPhi})\bP\bu_\bx^\bot}{\sqrt{m}} +\bw\Big\|_2 - \sqrt{\|\bu_\bx^\bot\|_2^2+\|\bw\|_2^2}\Big| \\
        &\nn \le \sup_{(\tilde{\bu},\bw)\in \tilde{\calU}_0} \Big| \Big\|\frac{\Im(\diag(\sign(\tilde{\bPhi}\be_1))^*\tilde{\bPhi})\tilde{\bu}}{\sqrt{m}}+\bw\Big\|_2 - \sqrt{\|\tilde{\bu}\|_2^2+\|\bw\|_2^2}\Big|\\
        & = \sup_{(\tilde{\bu}_1,\bw)\in \tilde{\calU}} \Big|\Big\|\frac{\Im (\diag(\sign(\tilde{\bPhi}^{[1]}))^* \tilde{\bPhi}^{[2:n]})\tilde{\bu}_1}{\sqrt{m}}+\bw\Big\|_2-\sqrt{\|\tilde{\bu}_1\|_2^2+\|\bw\|_2^2}\Big|\nn\\
        \nn:&=  \sup_{(\tilde{\bu}_1,\bw)\in \tilde{\calU}}\Big|\Big\|\frac{\hat{\bPhi}\tilde{\bu}_1}{\sqrt{m}}+ \bw\Big\|_2- \sqrt{\|\tilde{\bu}_1\|_2^2+\|\bw\|_2^2}\Big|,
    \end{align} 
   where $\hat{\bPhi}:= \Im (\diag(\sign(\tilde{\bPhi}^{[1]}))^* \tilde{\bPhi}^{[2:n]})$ has the same distribution as a matrix with i.i.d. $\calN(0,1)$ entries when conditioning on $\tilde{\bPhi}^{[1]}$. By (\ref{gwlink}) and (\ref{radiuslink}), a straightforward application of Lemma \ref{lem:extended}   yields the   claim. 
\end{proof}
\subsection{Proof of Theorem \ref{thm:nonuni} (RIP of $\widetilde{\bA}_{\bz}$ for Fixed $\bx$)}\label{proofthm7}
\begin{proof}  We let $\calK=\{\bx\}$ for a fixed $\bx\in\mathbb{S}^{n-1}$. Substituting the fixed-$\bx$ bound  in Lemma   \ref{lem:para_fix_x} into (\ref{541})--(\ref{eq:bound_large_1}) yields 
\begin{align}\label{fix1b}
    \mathbbm{P}\left(\sup_{(\bu,\bw)\in\calU_c^{\mathbb{S}}}|f^\|(\bx,\bu)|\le \frac{C_1[\omega(\calU_1^{\mathbb{S}})+t]}{\sqrt{m}}\right)\ge 1-2\exp(-t^2)- 4\exp(-c_2m).
\end{align}
Similarly, we substitute the fixed-$\bx$ bound in Lemma \ref{lem:tildef_ortho_fixedx} into (\ref{eq:tildef1bot}) to obtain 
\begin{align}\label{fix2b}
    \mathbbm{P}\left(\sup_{(\bu,\bw)\in\calU_c^{\mathbb{S}}}|f^\bot(\bx,\bu,\bw)|\le \frac{C_3[\omega(\calU_c^{\mathbb{S}})+t]}{\sqrt{m}}\right)\ge 1-\exp(-t^2)-4\exp(-c_4m). 
\end{align}
By (\ref{eq:divide_extend}), $m=\Omega(\delta^{-2}\omega^2(\calU_c^{\mathbb{S}}))$, and letting $t=c_5\sqrt{m}\delta$ with sufficiently small $c_5$ in (\ref{fix1b})--(\ref{fix2b}), we obtain 
\begin{align*}
   \mathbbm{P}\left(\sup_{(\bu,\bw)\in\calU_c^{\mathbb{S}}}|f(\bx,\bu,\bw)| \le \delta \right) \ge 1-11\exp(-c_6\delta^2m).
\end{align*}
Note that this event is equivalent to $\widetilde{\bA}_\bz\sim \rip(\calU_c,\delta)$,  so the proof is complete. 
\end{proof}

\subsection{Uniform (All-$\bx$) Bound on the Parallel Term} \label{sec:uniformr}

 We further extend Lemma \ref{lem:para_fix_x}   to a uniform bound for all $\bx\in \calK$ by a covering argument. 


\begin{lem} \label{lem:para_bound}
    Under the setting of Lemma \ref{lem:para_fix_x}, let $\calK$ be an arbitrary subset of $\mathbb{S}^{n-1}$. 
    There exist   absolute constants $c_1,C_2,C_3,c_4,C_5$ such that
    for any $\eta\in(0,c_1)$ and $r=\eta^2(\log(\eta^{-1}))^{1/2}$, if
\begin{equation}\label{eq:para_sample_size}
        m\geq \frac{C_2}{\log(\eta^{-1})}\left(\frac{\omega^2(\calU_1^{\mathbb{S}})}{\eta^2}+\frac{\scrH(\calK,r)}{\eta^2 }+\frac{\omega^2(\calK_{(r)})}{\eta^4}\right),
    \end{equation}
  \rev{then with probability at least $1-C_3\exp(-c_4\eta^2\log(\eta^{-1})m)$,} we have 
   $$\sup_{\bx\in \calK}\sup_{\bu\in\calU_1^{\mathbb{S}}}\big|f_1^\|(\bx,\bu)\big|\leq C_5\eta \sqrt{\log(\eta^{-1})}.$$ 
\end{lem}

\begin{proof}
    We   first extend the bound in Lemma \ref{lem:para_fix_x} to an $r$-net $\calN_r$ of $\calK$, and then bound the approximation error induced by approximating $\bx\in\calK$ by $ \bx\in\calN_r$. For clarity, we break down the proof into several pieces. 
    
    {\bf Uniform Bound on an $r$-Net:} For some $r>0$ that will be chosen later, we let $\calN_r$ be a   $r$-net of $\calK$ that is minimal in that $\log|\calN_r|=\scrH(\calK,r)$. Then   by Lemma \ref{lem:para_fix_x} and a union bound, 
    \begin{equation*} 
        \mathbbm{P}\left(\sup_{\bx\in \calN_r}\sup_{\bu\in\calU_1^{\mathbb{S}}}\big|f_1^\|(\bx,\bu)\big|\leq \frac{C(\omega(\calU_1^{\mathbb{S}})+t)}{\sqrt{m}}\right) \geq 1-2\exp\big(\scrH(\calK,r)-t^2\big)
    \end{equation*} 
    holds for any $t\geq 0$. \rev{Under the sample complexity \begin{align}
        m=\Omega\Big(\frac{1}{\eta^2\log(\eta^{-1})}[\omega^2(\calU_1^\mathbb{S})+\scrH(\calK,r)]\Big) \label{sampleparallel1}
    \end{align}
    with large enough implied constant,  
    setting $t=\eta \sqrt{\log(\eta^{-1})}\sqrt{m}$ yields that the event \begin{equation}\label{eq:bound_on_net1}
        \sup_{\bx\in \calN_r}\sup_{\bu\in\calU^{\mathbb{S}}}\big|f_1^\|(\bx,\bu)\big|\le  C_1 \eta\sqrt{\log(\eta^{-1})} 
    \end{equation}
 holds with probability at least $1-2\exp(- \frac{1}{2}\eta^2\log(\eta^{-1})m)$.}

     {\bf Bounding the Number of Small Measurements:} 
      For  small enough $\eta>0$, recall that $\calJ_{\bx,\eta} := \big\{i\in[m]:|\bPhi_i^*\bx|\leq \eta\big\}.$ We now  bound  $|\calJ_{\bx,\eta}|$ over $\bx\in \calN_r$. For a fixed $\bx\in\mathbb{S}^{n-1}$, by $\Re(\bPhi_i^*\bx)\sim\calN(0,1)$,  
      \[  p_0:=\mathbbm{P}\big(|\bPhi_i\bx|\leq \eta\big) \leq \mathbbm{P}\big(|\Re(\bPhi_i^*\bx)|\leq \eta\big)\leq \sqrt{\frac{2}{\pi}}\eta.\]  
          Note that $|\calJ_{\bx,\eta}|\sim \text{Bin}(m,p_0)$, so the Chernoff bound (e.g., \cite[Sec. 4.1]{motwani1995randomized}) gives $ \mathbbm{P}\big(|\calJ_{\bx,\eta}|\geq \eta m\big) \leq \exp(-c_1\eta m)$  for some absolute constant $c_1$. Thus,  a union bound over  $\bx\in \calN_r$ gives 
\begin{equation}
    \label{eq:number_NVM}\mathbbm{P}\left(\sup_{\bx\in\calN_r}|\calJ_{\bx,\eta}|< \eta m\right)\geq 1 - \exp\Big(\scrH(\calK,r)-c_1\eta m\Big),
\end{equation}
which  holds with probability at least $1-\exp(-c_2\eta m)$ as long as $m\ge  \frac{C_3\scrH(\calK,r
)}{\eta}$ for large enough $C_3$; note that this is implied by (\ref{sampleparallel1}). The remainder of the proof proceeds on the events (\ref{eq:bound_on_net1}) and (\ref{eq:number_NVM}).

    {\bf Bounding the Approximation Error:} We seek to bound the gap between $\sup_{\bx\in\calN_r}\sup_{\bu\in\calU_1^{\mathbb{S}}}|f_1^\|(\bx,\bu)|$ and $\sup_{\bx\in\calK}\sup_{\bu\in\calU_1^{\mathbb{S}}}|f_1^\|(\bx,\bu)|.$ For any $\bx\in \calK$ we let $\bx_r =\mathrm{arg}\min_{\bu\in \calN_r}\|\bu-\bx\|_2$.  Here, $\bx_r$ depends on $\bx$, but we drop such dependence to avoid cumbersome notation. Note that $\|\bx-\bx_r\|_2\leq r$, and indeed we have $\bx-\bx_r\in \calK_{(r)}=(\calK-\calK)\cap\mathbb{B}_2^n(r).$ For clarity we consider a given $\bx\in\calK$, while we note beforehand that the final bound (\ref{may20}) hold uniformly for all $\bx\in\calK$ (since all the arguments are uniform for all $\bx\in\calK$). 
  In view of $f_1^\|(\bx,\bu)$ defined in (\ref{eq:f1_parallel}), 
        \begin{align}
        \label{eq:cal_the_gap}&\sup_{\bu\in\calU_1^{\mathbb{S}}}\big|f_1^\|(\bx,\bu)\big|- \sup_{\bu\in\calU_1^{\mathbb{S}}}\big|f_1^\|(\bx_r,\bu)\big|
       \leq  \sup_{\bu\in\calU_1^{\mathbb{S}}}\big|f_1^\|(\bx,\bu)-f_1^\|(\bx_r,\bu)\big| \\
        &\leq \sup_{\bu\in\calU_1^{\mathbb{S}}}\Big|\frac{1}{\kappa m}\sum_{i=1}^m \Re\big([\overline{\sign(\bPhi_i^*\bx)-\sign(\bPhi_i^*\bx_r)}]\bPhi_i^*\bu\big)\Big|+\sup_{\bu\in\calU_1^{\mathbb{S}}}\big|\langle\bx_0-\bx_r,\bu\rangle\big|\nn\\
        &\le \sup_{\bu\in\calU_1^{\mathbb{S}}}\Big|\frac{1}{\kappa m}\sum_{i=1}^m \Re\big([\overline{\sign(\bPhi_i^*\bx)-\sign(\bPhi_i^*\bx_r)}]\bPhi_i^*\bu\big)\Big|+ r. \label{eq:para_appro_terms}
    \end{align} 
   To bound the first term in (\ref{eq:para_appro_terms}), we first divide it into two terms according to  $ \calJ_{\bx_r,\eta}$:  
        \begin{align}\nn
        & \sup_{\bu\in\calU_1^{\mathbb{S}}}\Big|\frac{1}{\kappa m}\sum_{i=1}^m \Re\big([\overline{\sign(\bPhi_i^*\bx)-\sign(\bPhi_i^*\bx_r)}]\bPhi_i^*\bu\big)\Big|
        \\ \nn
        &\leq  \sup_{\bu\in\calU_1^{\mathbb{S}}}\Big|\frac{1}{\kappa m}\sum_{i\in\calJ_{\bx_r,\eta}} \Re\big([\overline{\sign(\bPhi_i^*\bx)-\sign(\bPhi_i^*\bx_r)}]\bPhi_i^*\bu\big)\Big|\\&+\sup_{\bu\in\calU_1^{\mathbb{S}}}\Big|\frac{1}{\kappa m}\sum_{i\notin\calJ_{\bx_r,\eta}} \Re\big([\overline{\sign(\bPhi_i^*\bx)-\sign(\bPhi_i^*\bx_r)}]\bPhi_i^*\bu\big)\Big|:= I_1+I_2.\label{eq:decom_I1I2_para} 
    \end{align}

    \textit{Bounding $I_1$:} On the event of (\ref{eq:number_NVM}), we have $|\calJ_{\bx_r,\eta}|< \eta m$. Thus, by the universal bound $|\sign(\bPhi_i^*\bx)-\sign(\bPhi_i^*\bx_r)|\leq 2$, 
     \begin{align}
       \nn   I_1&\leq   \frac{1}{\kappa m}\left(\sum_{i\in\calJ_{\bx_r,\eta}}|\sign(\bPhi_i^*\bx)-\sign(\bPhi_i^*\bx_r)|^2\right)^{1/2}\sup_{\bu\in\calU_1^{\mathbb{S}}}\left(\sum_{i\in\calJ_{\bx_r,\eta}}\big|\bPhi_i^*\bu\big|^2\right)^{1/2}\\
         &\leq  \frac{2\eta}{\kappa}\sup_{\bu\in\calU_1^{\mathbb{S}}}\max_{\substack{I\subset [m]\\|I|\leq\eta m}} \left(\frac{1}{\eta m} \sum_{i\in I}|\bPhi_i^*\bu|^2\right)^{1/2}
         \le C_4\eta \left(\frac{\omega(\calU_1^{\mathbb{S}})}{\sqrt{\eta m}}+\sqrt{\log(\eta^{-1})}\right), \label{eq:bound_on_I1}
     \end{align}
     where the last inequality follows from Lemma \ref{lem:max_k_sum} and  holds with probability at least  $1-2\exp(-c_5\eta m\log(\eta^{-1}))$.

    \textit{Bounding $I_2$:} When $i\notin \calJ_{\bx_r,\eta}$ we have $|\bPhi_i^*\bx_r|>\eta$, and hence (\ref{eq:sign_conti}) gives $
        |\sign(\bPhi_i^*\bx )-\sign(\bPhi_i^*\bx_r)| \leq \frac{2}{\eta}|\bPhi_i^*(\bx -\bx_r)|$, and therefore
            \begin{align}
     \label{eq:I2_initial}  & I_2 \leq \frac{2}{\kappa \eta} \frac{1}{m}\sup_{\bu\in\calU_1^{\mathbb{S}}}\sum_{i\notin \calJ_{\bx_r,\eta}}\big|\bPhi_i^*(\bx -\bx_r)\big|\big|\bPhi_i^*\bu\big|\\\label{eq:change1}
        &\leq \frac{2}{\kappa \eta} \left(\frac{1}{m}\sum_{i=1}^m|\bPhi_i^*(\bx -\bx_r)|^2\right)^{1/2}\sup_{\bu\in\calU_1^{\mathbb{S}}}\left(\frac{1}{m}\sum_{i=1}^m |\bPhi_i^*\bu|^2\right)^{1/2} \\
        &\le \frac{C_5}{\eta} \sup_{\bv\in \calK_{(r)}}\left(\frac{1}{m}\sum_{i=1}^m|\bPhi_i^*\bv|^2\right)^{1/2}\le \frac{C_6}{\eta}\left(\frac{\omega(\calK_{(r)})}{\sqrt{m}}+r\right),
        \label{eq:I2_end} 
    \end{align}
    where in the first inequality of (\ref{eq:I2_end}) we use  $\bx -\bx_r\in  \calK_{(r)}$ and (\ref{Oonebound}) which holds with probability at least $1-4\exp(-cm)$, and in the second inequality  of (\ref{eq:I2_end}) we use (\ref{upperdevi}) with $t=\sqrt{m}$ to obtain 
      \begin{align*}
           \sup_{\bv\in \calK_{(r)}}\left(\frac{1}{m}\sum_{i=1}^m|\bPhi_i^*\bv|^2\right)^{1/2} \le  \sup_{\bv\in \calK_{(r)}}\frac{\|\bPhi^\Re\bv\|_2}{\sqrt{m}}+ \sup_{\bv\in \calK_{(r)}}\frac{\|\bPhi^\Im\bv\|_2}{\sqrt{m}} = O\left(  \frac{\omega(\calK_{(r)})}{\sqrt{m}}+ r\right)
      \end{align*} with probability at least $1-4\exp(-m)$.

Substituting (\ref{eq:bound_on_I1}) and (\ref{eq:I2_end}) into (\ref{eq:decom_I1I2_para}) and taking supremum over $\bx\in\calK$, we  bound   the first term in (\ref{eq:para_appro_terms}) as 
\begin{align}
    \label{may20}\sup_{\bx\in\calK}\sup_{\bu\in\calU_1^{\mathbb{S}}}\left|\frac{1}{\kappa m}\sum_{i=1}^m \Re\big([\overline{\sign(\bPhi_i^*\bx )-\sign(\bPhi_i^*\bx_r)}]\bPhi_i^*\bu\big)\right|\le C_7\left( \frac{\sqrt{\eta}\cdot\omega(\calU_1^{\mathbb{S}})}{\sqrt{m}} +\eta\sqrt{\log(\eta^{-1})}+\frac{\omega(\calK_{(r)})}{\eta\sqrt{m}}+ \frac{r}{\eta}\right).
\end{align} 

    {\bf Completing the Proof:}  
     Substituting (\ref{may20}) into (\ref{eq:cal_the_gap})--(\ref{eq:para_appro_terms}) and taking the supremum over 
    $\bx\in\calK$, we obtain
    \begin{align*}
        \sup_{\bx\in\calK}\sup_{\bu\in\calU_1^{\mathbb{S}}}\big|f_1^\|(\bx,\bu)\big| &\le \sup_{\bx\in\calN_r}\sup_{\bu\in\calU_1^{\mathbb{S}}}\big|f_1^\|(\bx,\bu)\big| +r +C_7\left( \frac{\sqrt{\eta}\cdot\omega(\calU_1^{\mathbb{S}})}{\sqrt{m}} +\eta\sqrt{\log(\eta^{-1})}+\frac{\omega(\calK_{(r)})}{\eta\sqrt{m}}+ \frac{r}{\eta}\right).
    \end{align*}
Then, we further apply (\ref{eq:bound_on_net1}),  $\eta<1$, and $\frac{\omega(\calU_1^\mathbb{S})}{\sqrt{m}}\le \eta\sqrt{\log(\eta^{-1})}$ that holds under (\ref{sampleparallel1}), to arrive at 
    \begin{equation}
\label{eq:final_bound}\sup_{\bx\in\calK}\sup_{\bu\in\calU_1^{\mathbb{S}}}\big|f_1^\|(\bx,\bu)\big|\le C_8\left(\eta\sqrt{\log(\eta^{-1})}+\frac{\omega(\calK_{(r)})}{\eta\sqrt{m}}+ \frac{r}{\eta}\right).
\end{equation}
In summary, this bound holds with probability at least  $
        1-C_9\exp(-c_{10}\eta^2\log(\eta^{-1})m) 
    $ under the sample complexity (\ref{sampleparallel1}); 
    here, the sample complexity and probability term can be seen from the events (\ref{eq:bound_on_net1}), (\ref{eq:number_NVM}) and (\ref{eq:bound_on_I1}). 
    We now set $r=\eta^2(\log(\eta^{-1}))^{1/2}$, and using the assumed sample complexity (\ref{eq:para_sample_size}), 
    then the bound in  (\ref{eq:final_bound}) reads as $O(\eta \sqrt{\log(\eta^{-1})})$.
  This completes the proof. 
\end{proof}

\subsection{Uniform (All-$\bx$) Bound on the Orthogonal Term} 

Similarly to Section \ref{sec:uniformr}, we strengthen   Lemma \ref{lem:tildef_ortho_fixedx} to a universal bound over $\bx\in\calK$. Our major technical refinement over \cite[Lem. 14]{chen2023uniform} lies in the introduction of $\calI_{\bu,\eta'}$. 

\begin{lem} \label{lem:tilde_ortho_unibound}  
In the setting of Lemma \ref{lem:tildef_ortho_fixedx}, there exist some absolute constants $c_1,C_2,C_3,c_4,C_5$ such that
     given any  $\eta\in(0,c_1)$ and $\calK\subset \mathbb{S}^{n-1}$, if 
      \begin{equation}\label{eq:tilde_ortho_con}
        m\geq C_2\left(\frac{\omega^2(\calU_c^{\mathbb{S}})}{\eta^2\log(\eta^{-1})} +\frac{\omega^2(\calK_{(\eta^3)})}{\eta^8\log(\eta^{-1})}+\frac{\scrH(\calK,\eta^3)}{\eta^2}\right),
    \end{equation}
  \rev{then with probability at least $1-C_3\exp(-c_4\eta^2m)$, we have}
    \begin{align*}
        \sup_{\bx\in\calK}\sup_{(\bu,\bw) \in\calU_c^{\mathbb{S}}}\big|{f}_1^\bot(\bx,\bu,\bw)\big| \le C_5\eta\sqrt{\log(\eta^{-1})}.
    \end{align*}
\end{lem}
\begin{proof}
    The proof will be presented in several steps.  
     
      {\bf Uniform Bound on an $r$-Net:}
      For some $r>0$ that will be chosen later, we let $\calN_r$ be an $r$-net of $\calK$ that is minimal in that $\log|\calN_r|=\scrH(\calK,r)$. We apply Lemma \ref{lem:tildef_ortho_fixedx} to every $\bx\in\calN_r$, along with a union bound, to obtain that for any $t\geq 0$,  the event  $$\sup_{\bx\in\calN_r}\sup_{(\bu,\bw)\in\calU_c^{\mathbb{S}}}|f_1^\bot(\bx,\bu,\bw)|\le  \frac{C_1(\omega(\calU_c^{\mathbb{S}})+t)}{\sqrt{m}}$$ holds with probability at least $1-2\exp(\scrH(\calK,r)-t^2)$.
      Therefore, \rev{under the sample complexity \begin{align}
          m=\Omega\Big(\frac{1}{\eta\log(\eta^{-1})}[\omega^2(\calU_c^\mathbb{S})+\scrH(\calK,r)]\Big)\label{samplebot1}
      \end{align} with large enough implied constant, setting $t =\sqrt{\eta\log(\eta^{-1})m}$ yields that the event 
    \begin{equation}
        \label{eq:ortho_bound_on_net}
\sup_{\bx\in\calN_r}\sup_{\bu\in\calU^{\mathbb{S}}}\big|f_1^\bot(\bx,\bu,\bw)\big|\le C_2 \sqrt{\eta\log(\eta^{-1})}
    \end{equation}
    holds with probability at least $1-2\exp(-\frac{1}{2}\eta\log(\eta^{-1})m)$.}

    {\bf Bounding the Number of Small Measurements:}
    As shown in the proof of Lemma \ref{lem:para_bound}, under the sample complexity \begin{align}
        m\ge \frac{C_3\scrH(\calK,r)}{\eta} \label{samplebot2},
    \end{align}
   the event 
   \begin{align}\label{netboundJxeta}
       \sup_{\bx\in\calN_r}|\calJ_{\bx,\eta}|:=\sup_{\bx\in\calN_r}|\{i\in [m]:|\bPhi_i^*\bx|\le\eta\}|<\eta m
   \end{align}
holds   
     with probability at least $1-\exp(-c_4\eta m)$. We will   utilize this event.

         {\bf Bounding the Approximation Error:} 
       For any $\bx\in\calK$, we let $\bx_r=\mathrm{arg}\min_{\bu\in\calN_r}\|\bu-\bx \|_2$. Note that $\|\bx -\bx_r\|_2\leq r$ and we have $\bx-\bx_r \in \calK_{(r)}$. 
       For clarity we consider a given $\bx\in\calK$, but we note that the forthcoming arguments hold uniformly for all $\bx\in\calK.$
       By $f_1^\bot(\bx,\bu,\bw)$ defined in (\ref{eq:tildef1bot}) and triangle inequality, 
            \begin{align}
    \label{eq:tilde_show_gap} &\sup_{(\bu,\bw)\in\calU_c^{\mathbb{S}}}\big|f_1^\bot(\bx,\bu,\bw)\big|-\sup_{(\bu,\bw)\in\calU_c^{\mathbb{S}}}\big|f_1^\bot(\bx_r,\bu,\bw)\big|\\\nn
    &\leq \sup_{(\bu,\bw)\in\calU_c^{\mathbb{S}}}\big|f_1^\bot(\bx,\bu,\bw)-f_1^\bot(\bx_r,\bu,\bw)\big|\\&\leq\sup_{(\bu,\bw)\in\calU_c^{\mathbb{S}}}\Big|\Big\|\frac{\Im(\diag(\sign(\bPhi\bx))^*\bPhi)\bu}{\sqrt{m}}+\bw\Big\|_2-\Big\|\frac{\Im(\diag(\sign(\bPhi\bx_r))^*\bPhi)\bu}{\sqrt{m}}+\bw\Big\|_2\Big| \label{eq:tilde_ortho_hard}\\&\quad\quad\quad\quad+\sup_{(\bu,\bw)\in\calU_c^{\mathbb{S}}}\Big|\Big\|\begin{bsmallmatrix}
        \bu_{\bx}^\bot \\\bw
    \end{bsmallmatrix}\Big\|_2-\Big\|\begin{bsmallmatrix}
        \bu_{\bx_r}^\bot \\\bw
\end{bsmallmatrix}\Big\|_2\Big|.\label{eq:tilde_ortho_easy}
     \end{align}
     For any $\bu\in\mathbb{S}^{n-1}$, $\bx\in \calK$ and its associated $\bx_r\in \calN_r$, by the triangle inequality, \begin{align}\nn
         \text{(the term in (\ref{eq:tilde_ortho_easy}))}&\le \big|\|\bu^\bot_{\bx}\|_2-\|\bu_{\bx_r}^\bot\|_2\big| \le \big\|\bu_{\bx}^\bot - \bu_{\bx_r}^\bot\big\|_2   = \big\| [\bu - \langle \bu,\bx\rangle \bx]-[\bu-\langle \bu,\bx_r\rangle \bx_r] \big\|_2 
         \\&\le\big\| \langle \bu,\bx-\bx_r\rangle \bx\big\|_2 + \big\|\langle \bu,\bx_r\rangle(\bx-\bx_r)\big\|_2
         \le 2\|\bx-\bx_r\|_2\le 2r,
\label{tworbound}     \end{align}
     so     the term in (\ref{eq:tilde_ortho_easy}) is bounded  by $2r$ uniformly for all $\bx\in \calK$. It remains to bound the term  in (\ref{eq:tilde_ortho_hard}).

         By the triangle inequality and the observation that $(\bu,\bw)\in \calU_c^{\mathbb{S}}$ gives $\bu\in \calU_1\cap\mathbb{B}_2^n$, 
     \begin{align}\nn
     \text{(the term in (\ref{eq:tilde_ortho_hard}))}&\le \frac{1}{\sqrt{m}}\sup_{\bu \in \calU_1\cap \mathbb{B}_2^n}\big\|\Im\big[\diag(\overline{\sign(\bPhi\bx )-\sign(\bPhi\bx_r)})\bPhi\bu\big]\big\|_2 \\
     &= \frac{1}{\sqrt{m}}\sup_{\bu \in \calU_1^{\mathbb{S}}}\big\|\Im\big[\diag(\overline{\sign(\bPhi\bx )-\sign(\bPhi\bx_r)})\bPhi\bu\big]\big\|_2.\label{eq:coincide_pre_pro}
     \end{align}
   We   divide the $m$ measurements into two parts according to certain index sets. For some small enough $\eta'>0$ to be chosen and $\bu \in\mathbbm{R}^n$, we further introduce the index set 
         \begin{equation}\label{eq:important_index}
             \calI_{\bu,\eta'} = \big\{i\in[m]:|\bPhi_i^*\bu|> \eta'\big\}. 
         \end{equation}
         We pause to establish a uniform bound on $|\calI_{\bu,\eta'}|$ for $\bu\in\calK_{(r)}$.

         \textit{Bounding $|\calI_{\bu,\eta'}|$ uniformly over $\bu\in\calK_{(r)}$:}
         For $\beta\ge \frac{1}{m}$ to be chosen, by Lemma \ref{lem:max_k_sum}, the event 
         \begin{equation}\label{eq:bound_large_measurement}
             \sup_{\bu\in \calK_{(r)}}\big|\calI_{\bu,\eta'}\big| \le\beta m
         \end{equation}
         holds {with probability at least} $1-4\exp(-c_5\beta m\log(\beta^{-1}))$, as long as 
         \begin{equation}\label{eq:ortho_size1}
            \frac{\omega(\calK_{(r)})}{\sqrt{\beta m}}+r\sqrt{\log(\beta^{-1})}\le c_6 \eta '
         \end{equation}
holds for some sufficiently small $c_6$. To see why this is sufficient, note that with the promised probability (\ref{eq:ortho_size1}) implies \begin{align*}
    \sup_{\bv\in \calK_{(r)}} \max_{\substack{I\subset[m]\\|I|\le \beta m}}\left(\frac{1}{\beta m}\sum_{i\in I}|\bPhi_i^*\bu|^2\right)^{1/2}\le \frac{\eta'}{2},
\end{align*}
and this further implies (\ref{eq:bound_large_measurement}).
Our subsequent analysis is built upon the bound in (\ref{eq:bound_large_measurement}).

For a specific $(\bx,\bx_r)\in \calK\times \calN_r$, we  define the index set for the ``problematic measurements'' as
\begin{align}\label{eq:ortho_bad}
\calE_{\bx}:=\calJ_{\bx_r,\eta}\cup \calI_{\bx-\bx_r,\eta'}.
\end{align}  Then, we bound the term in (\ref{eq:coincide_pre_pro}) by $I_3+I_4$, where 
         \begin{align*}
             &I_3= \sup_{\bu\in\calU_1^{\mathbb{S}}}\left(\frac{1}{m}\sum_{i\in\calE_{\bx}}\Big[\Im\Big(\big[\overline{\sign(\bPhi_i^*\bx)-\sign(\bPhi_i^*\bx_r)}\big]\bPhi_i^*\bu\Big)\Big]^2\right)^{1/2},
             \\&I_4=\sup_{\bu\in\calU_1^{\mathbb{S}}}\left(\frac{1}{m}\sum_{i\notin\calE_{\bx}}\Big[\Im\Big(\big[\overline{\sign(\bPhi_i^*\bx)-\sign(\bPhi_i^*\bx_r)}\big]\bPhi_i^*\bu\Big)\Big]^2\right)^{1/2}. 
         \end{align*}

          \textit{Bounding $I_3$:} The issue for measurements in $\calE_{\bx}$ is   the lack of a good bound on $|\sign(\bPhi_i^*\bx)-\sign(\bPhi_i^*\bx_r)|$. Fortunately, these measurements are quite few: by  (\ref{netboundJxeta}) and (\ref{eq:bound_large_measurement}), 
          \begin{equation*}
              |\calE_{\bx}|\leq |\calJ_{\bx_r,\eta}|+|\calI_{\bx-\bx_r,\eta'}|\le \sup_{\bx\in\calN_r}|\calJ_{\bx,\eta}| + \sup_{\bu\in\calK_{(r)}}|\calI_{\bu,\eta'}|<(\eta+\beta)m\,,\quad\forall \bx\in\calX.
          \end{equation*}
          Combining with $|\Im([\overline{\sign(\bPhi_i^*\bx) - \sign(\bPhi_i^*\bx_r)}]\bPhi_i^*\bu)|\leq 2|\bPhi_i^*\bu|$,  
          \begin{align}
          \nn
                I_3&\leq 2\sup_{\bu\in\calU_1^{\mathbb{S}}}\Big(\frac{1}{m}\sum_{i\in\calE_{\bx}}|\bPhi_i^*\bu|^2\Big)^{1/2}\\&\leq  2\sqrt{\eta +\beta}\sup_{\bu\in\calU_1^{\mathbb{S}}}\max_{\substack{I\subset [m]\\|I|\leq (\eta+\beta)m}}\left(\frac{1}{(\eta+\beta)m}\sum_{i\in I}|\bPhi_i^*\bu|^2\right)^{1/2} \nn\\
              \label{eq:use_jems_1}&\le  C_7\left(\frac{\omega(\calU_1^{\mathbb{S}})}{\sqrt{m}}+\sqrt{(\eta+\beta)\log\big(\frac{e}{\eta+\beta}\big)} \right),
          \end{align}
          where (\ref{eq:use_jems_1}) follows from Lemma \ref{lem:max_k_sum} and holds with probability at least  $1-4\exp(-c_8(\eta+\beta)m\log(\frac{e}{\eta+\beta}))$.

           \textit{Bounding $I_4$:} For $i\notin \calE_{\bx}$ we have $|\bPhi_i^*\bx_r|\geq \eta$ and $|\bPhi_i^*(\bx-\bx_r)|<\eta'$, and hence (\ref{eq:sign_conti}) implies
           \begin{equation}
                 \label{eq:ortho_better_boud}\big|\sign(\bPhi_i^*\bx)-\sign(\bPhi_i^*\bx_r)\big| \leq \frac{2|\bPhi_i^*(\bx-\bx_r)|}{\eta}\leq \frac{2\eta'}{\eta}.
           \end{equation}
           Therefore, by $|\Im([\overline{\sign(\bPhi_i^*\bx)-\sign(\bPhi_i^*\bx_r)}]\bPhi_i^*\bu)|\leq \frac{2\eta'}{\eta}|\bPhi_i^*\bu|$,
           \begin{equation}\label{eq:bound_I4}
               I_4 \leq \frac{2\eta'}{\eta}\sup_{\bu\in\calU_1^{\mathbb{S}}}\frac{\|\bPhi\bu\|_2}{\sqrt{m}}\le \frac{C_9\eta'}{\eta},
           \end{equation}
           where the second inequality holds with probability at least $1-4\exp(-c_{10}m)$ if $m = \Omega(\omega^2(\calU_1^{\mathbb{S}}))$; see 
            (\ref{Oonebound}).  Combining (\ref{eq:use_jems_1}) and (\ref{eq:bound_I4}) and recalling (\ref{eq:coincide_pre_pro}), we arrive at  
           \begin{align}
              \text{(the term in (\ref{eq:tilde_ortho_hard}))}\le C_{11}\left(  \frac{\omega(\calU_1^{\mathbb{S}})}{\sqrt{m}}+\sqrt{(\eta+\beta)\log\Big(\frac{e}{\eta+\beta}\Big)}+\frac{\eta'}{\eta}\right) . \label{eq:ortho_hard_bound}
           \end{align}

    {\bf Completing the Proof:}
 By Equations (\ref{tworbound}) and (\ref{eq:ortho_hard_bound}), the terms in (\ref{eq:tilde_ortho_hard}) and (\ref{eq:tilde_ortho_easy}) are respectively bounded by the right-hand side of (\ref{eq:ortho_hard_bound}) and $2r$,  uniformly for all $\bx\in \calK$. Substituting them into (\ref{eq:tilde_show_gap})--(\ref{eq:tilde_ortho_easy}), along with a supremum over $\bx\in\calK$,  yields 
    \begin{align*}
        &\sup_{\bx\in\calK}\sup_{(\bu,\bw)\in\calU_c^{\mathbb{S}}}\big|f_1^\bot(\bx,\bu,\bw)\big|\le  \sup_{\bx\in\calK}\sup_{(\bu,\bw)\in\calU_c^{\mathbb{S}}}\big|f_1^\bot(\bx_r,\bu,\bw)\big|\\
        &\qquad +  C_{11}\left(  \frac{\omega(\calU_1^{\mathbb{S}})}{\sqrt{m}}+\sqrt{(\eta+\beta)\log\Big(\frac{e}{\eta+\beta}\Big)}+\frac{\eta'}{\eta}\right)+2r.
    \end{align*}
    Combining with the bound in (\ref{eq:ortho_bound_on_net}), taking $\beta = \Theta( \eta)$, and also summarizing the sample complexity and probability terms, we arrive at the following conclusion:  Suppose \begin{align}\label{ccc1}
        m\ge C_{12}\left[\frac{\omega^2(\calU_c^{\mathbb{S}})}{\eta\log(\eta^{-1})}+\frac{\scrH(\calK,r)}{\eta}\right]\quad\text{with large enough }C_{12},
    \end{align} and \begin{align}
        \frac{\omega(\calK_{(r)})}{\sqrt{\eta m}}+r\sqrt{\log(\eta^{-1})}\le c_{13} \eta'\quad\text{with small enough }c_{13},\label{requrrr}
    \end{align} the event \begin{equation*}
    \sup_{\bx\in\calK}\sup_{\bu\in\calU^{\mathbb{S}}}|f_1^\bot(\bx,\bu)| \le C_{14}\left(\sqrt{\eta\log(\eta^{-1})}+\frac{\eta'}{\eta}+r\right)
              \end{equation*}
                holds with probability at least $1-C_{15}\exp(-c_{16}\eta m)$.  We mention that the condition (\ref{ccc1}) is  needed for ensuring (\ref{eq:ortho_bound_on_net}),   (\ref{netboundJxeta}), and the second inequality of (\ref{eq:bound_I4}), and the condition (\ref{requrrr}) is needed in (\ref{eq:bound_large_measurement}).

                {\bf Further Simplification:} We now take the tightest choice for $\eta'$ that satisfies the required (\ref{requrrr}), namely \[ 
                    \eta'  =\Theta\Big(\frac{\omega(\calK_{(r)})}{\sqrt{\eta m}}+r\sqrt{\log(\eta^{-1})}\Big).\]             We further set $\eta=\hat{\eta}^2$ and  $r=\eta^{3/2}=\hat{\eta}^3$ for some $\hat{\eta}>0$,  
                and then the above statement simplifies to the following:  
                if  
                \begin{align*}
                    m\ge C_{17}\left(\frac{\omega^2(\calU^{\mathbb{S}}_c)}{\hat{\eta}^2\log(\hat{\eta}^{-1})}+\frac{\scrH(\calK,\hat{\eta}^{3})}{\hat{\eta}^2}\right), 
                \end{align*} 
                then the event 
                \begin{align}
                    \label{eq:right123}\sup_{\bx\in\calK}\sup_{\bu\in\calU^{\mathbb{S}}}|f_1^\bot(\bx,\bu)| \le C_{18}\left(\frac{\hat{\eta}^{-3}\omega(\calK_{(\hat{\eta}^{3})})}{\sqrt{m}}+ \hat{\eta}\sqrt{\log (\hat{\eta}^{-1})}\right)
                \end{align}
               holds with probability at least $1-C_{19}\exp(-c_{20}\hat{\eta}^2m)$.      Under the sample complexity 
               \[m=\Omega\left(\frac{\omega^2(\calU_c^{\mathbb{S}})}{\hat{\eta}^2\log(\hat{\eta}^{-1})} +\frac{\omega^2(\calK_{(\hat{\eta}^3)})}{\hat{\eta}^8\log(\hat{\eta}^{-1})}+\frac{\scrH(\calK,\hat{\eta}^3)}{\hat{\eta}^2}\right),\]
               which is identical to (\ref{eq:tilde_ortho_con})   except for the notation $\eta$ versus $\hat{\eta}$, we have $\frac{\hat{\eta}^{-3}\omega(\calK_{(\hat{\eta}^3)})}{\sqrt{m}}=O(\hat{\eta}\sqrt{\log(\hat{\eta}^{-1})})$ and hence the right-hand side of (\ref{eq:right123}) is bounded by $O(\hat{\eta}\sqrt{\log(\hat{\eta}^{-1})})$. 
               Further renaming $\hat{\eta}$ to $\eta$ completes the proof.     
\end{proof}

\subsection{Proof of Lemma \ref{lem:RIPextend1}}
\begin{proof}
    We are ready to substitute the bounds for the parallel term and the orthogonal term into (\ref{eq:divide_extend}) to establish Lemma \ref{lem:RIPextend1}. Recall from  (\ref{541}) and (\ref{eq:bound_large_1}) that 
    \begin{align*}
  \sup_{\bx\in \calK}\sup_{\bu\in \calU_1^{\mathbb{S}}}|f^\|(\bx,\bu)|\le C_1\sup_{\bx\in \calK}\sup_{\bu\in \calU_1^{\mathbb{S}}}|f_1^\|(\bx,\bu)|
\end{align*}   holds with probability $1-4\exp(-c_2m)$, and that 
\begin{align*}
    \sup_{\bx\in\calK}\sup_{(\bu,\bw)\in\calU_c^{\mathbb{S}}}|{f}^\bot(\bx,\bu,\bw)|\le C_3\sup_{\bx\in\calK}\sup_{(\bu,\bw)\in\calU_c^{\mathbb{S}}}|{f}_1^\bot(\bx,\bu,\bw)|
\end{align*}  
holds with probability at least $1-4\exp(-c_4m)$ due to (\ref{eq:tildef1bot}). 
We now observe that the stated sample complexity (\ref{samplelem4}) implies   (\ref{eq:para_sample_size})  and (\ref{eq:tilde_ortho_con}), and hence we can apply Lemma \ref{lem:para_bound} and Lemma \ref{lem:tilde_ortho_unibound} to obtain
\begin{gather*}
    \sup_{\bx\in\calK}\sup_{\bu\in\calU_1^{\mathbb{S}}}|{f}_1^\|(\bx,\bu)|=O\big(\eta\sqrt{\log(\eta^{-1})}\big),\\
    \sup_{\bx\in\calK}\sup_{(\bu,\bw)\in\calU_c^{\mathbb{S}}}|f_1^\bot(\bx,\bu,\bw)|=O\big(\eta\sqrt{\log(\eta^{-1})}\big),
\end{gather*} 
 that hold with the promised probability. Therefore, we arrive at 
\begin{align*}
    \sup_{\bx\in\calK}\sup_{(\bu,\bw)\in\calU_c^{\mathbb{S}}}|f^\|(\bx,\bu)|+\sup_{\bx\in\calK}\sup_{(\bu,\bw)\in\calU_c^{\mathbb{S}}}|f^\|(\bx,\bu,\bw)| = O\big(\eta\sqrt{\log(\eta^{-1})}\big). 
\end{align*} 
In view of (\ref{eq:extended_goal}) and (\ref{eq:divide_extend}), we derive $\sup_{\bx\in\calK}\sup_{(\bu,\bw)\in\calU_c^{\mathbb{S}}}f(\bx,\bu,\bw)=O(\eta\sqrt{\log(\eta^{-1})})$, which is just the desired RIP with     distortion $\delta_\eta$ in (\ref{eq:tildeAzrip}). The proof is complete.
\end{proof}

\section{Deferred Proofs (Proposition \ref{pro2} \& Lemma \ref{lem:improved_lem9})}\label{app:missing} 

\subsection{Proof of Proposition \ref{pro2}}\label{app:zeta}
\begin{proof} 
By $\bA_{\bz}\bx^\star=\be_1$ and $\bA_{\breve{\bz}}-\bA_{\bz}=\bA_{\breve{\bz}-\bz}=\bA_{\bzeta}$,  \[\varepsilon\ge(1+\delta_1)\|\bA_{\breve{\bz}}\bx^\star-\be_1\|_2=(1+\delta_1)\|\bA_{\bzeta}\bx^\star\|_2.\] 
As in the proof of Theorem \ref{thm:zetabound} (especially the part of ``Establishing  (\ref{triAzu})''),  $\bA_{\breve{\bz}}\sim\rip(\Sigma^n_{2s},\frac{1}{3})$ holds with probability at least $1-2\exp(-c_1m)$. On this event, we claim that all points in $\Sigma^n_s\cap \mathbb{B}_2^n(\bx^\star;\frac{4}{5}\delta_1\|\bA_{\bzeta}\bx^\star\|_2)$ satisfy the constraint $\|\bA_{\breve{\bz}}\bu-\be_1\|_2\le\varepsilon$. To see this, if $\bu\in \Sigma^n_s\cap \mathbb{B}_2^n(\bx^\star;\frac{4}{5}\delta_1\|\bA_{\bzeta}\bx^\star\|_2)$, then we have
\begin{align*}
    \|\bA_{\breve{\bz}}\bu - \be_1\|_2 \le \|\bA_{\breve{\bz}}(\bu-\bx^\star)\|_2+\|\bA_{\breve{\bz}}\bx^\star-\be_1\|_2 \le\sqrt{\frac{4}{3}}\frac{4\delta_1\|\bA_{\bzeta}\bx^\star\|_2}{5} + \|\bA_{\bzeta}\bx^\star\|_2 \le\varepsilon. 
\end{align*}

\textbf{Bounding $\|\bA_{\bzeta}\bx^\star\|_2$ from below.} Next, we lower bound $\|\bA_{\bzeta}\bx^\star\|_2$. We start with \[\|\bA_{\bzeta}\bx^\star\|_2 = \frac{\kappa m}{\|\bPhi\bx\|_2}\|\bA_{\bzeta}\bx\|_2\ge\frac{1}{2}\|\bA_{\bzeta}\bx\|_2,\] where in the inequality we use $\|\bx^\star\|_2=\frac{\kappa m}{\|\bPhi\bx\|_2}\ge\frac{1}{2}$ that holds with the  probability at least $1-2\exp(-c_2m)$ due to   Equation (\ref{l1l2fixedpoint}). We   denote the index set for the $\zeta_0m$ measurements with the largest $|\bPhi_i^*\bx|$ by $I_{\zeta_0}$. By recalling (\ref{Awmatrix}) and  that $\bzeta$ changes the  measurements in $I_{\zeta_0}$ from $z_i$ to $\bi z_i$, 
\begin{align}\nn
     \|\bA_{\bzeta}\bx\|_2 &\ge \frac{\|\Im(\diag(\bzeta^*)\bPhi\bx)\|_2}{\sqrt{m}} =\frac{1}{\sqrt{m}}\left(\sum_{i\in I_{\zeta_0}} \big[\Im\big(\overline{(\bi-1)\bPhi_i^*\bx}\cdot\bPhi_i^*\bx\big)\big]^2\right)^{1/2}\\
    & = \frac{1}{\sqrt{m}}\left(\sum_{i\in I_{\zeta_0}}|\bPhi_i^*\bx|^2\right)^{1/2} \ge \frac{1}{\sqrt{m}}\cdot\frac{1}{\sqrt{\zeta_0m}}\sum_{i\in I_{\zeta_0}}|\bPhi_i^*\bx|,\label{120bound}
\end{align}
where the last step follows from Cauchy-Schwarz inequality. 
 We further let $I_{\zeta_0}'$ be the index set for the $\zeta_0m$ measurements with the largest $|(\bPhi^\Re_i)^\top\bx|$. Since $I_{\zeta_0}$ corresponds to the $\zeta_0m$ measurements with the largest $|\bPhi_i^*\bx|$, we continue from (\ref{120bound}) to obtain 
\begin{align*}
    \|\bA_{\bzeta}\bx\|_2 \ge  \frac{1}{\sqrt{m}}\cdot \frac{1}{\sqrt{\zeta_0m}}\sum_{i\in I_{\zeta_0}'}\big|(\bPhi_i^\Re)^\top\bx\big|.
\end{align*}
Now let us construct a set $\calV:=\big\{\frac{1}{\sqrt{\zeta_0m}},-\frac{1}{\sqrt{\zeta_0m}},0\big\}^m\cap\Sigma^m_{\zeta_0m}$
whose elements are $m$-dimensional, $\zeta_0m$-sparse $\{\frac{\pm 1}{\sqrt{\zeta_0m}},0\}$-valued vectors. Combining with the definition of $I_{\zeta_0}'$ we can write 
\begin{align*}
    \frac{1}{\sqrt{\zeta_0m}}\sum_{i\in I_{\zeta_0}'}\big|(\bPhi_i^\Re)^\top\bx\big| = \max_{\bv\in\calV}\bv^\top\bPhi^\Re \bx \stackrel{d}{=} \max_{\bv\in\calV} \bg^\top \bv, 
\end{align*}
where ``$\stackrel{d}{=}$'' means $\max_{\bv\in\calV}\bv^\top\bPhi^\Re \bx $ and $\max_{\bv\in\calV} \bg^\top \bv$ have the same distribution. Therefore,   Gaussian concentration (e.g., \cite[Thm. 5.2.2]{vershynin2018high}) yields that  \[\max_{\bv\in\calV}\bv^\top\bPhi^\Re\bx \ge \frac{1}{2}\omega(\calV)\] holds with probability at least $1-2\exp(-c_3\omega^2(\calV))$. 
We now seek lower bound on $\omega(\calV)$. By the Sparse Varshamov-Gilbert construction (e.g., \cite[Lem. 4.14]{rigollet2015high}) there exist $(1+\frac{1}{2\zeta_0})^{\frac{\zeta_0m}{8}}$ distinct points contained in $\calV$ with mutual $\ell_2$ distances greater than $\frac{1}{\sqrt{2}}$. This implies a lower bound on the metric entropy, 
$$\scrH\Big(\calV,\frac{1}{2\sqrt{2}}\Big) = \log\scrN\Big(\calV,\frac{1}{2\sqrt{2}}\Big)\ge \frac{\zeta_0m}{8}\log\Big(1+\frac{1}{2\zeta_0}\Big),$$ and     Sudakov's inequality (\ref{sudakov}) further gives $\omega^2(\calV)\ge c_4\zeta_0m\log(\frac{e}{\zeta_0})$. Combining these pieces, with the promised probability, we have \[\|\bA_{\bzeta}\bx\|_2\ge c_5\sqrt{\zeta_0\log(e/\zeta_0)}.\]

Recall that all points in $\Sigma^n_s\cap\mathbb{B}_2^n(\bx^\star;\frac{4\delta_1}{5}\|\bA_{\bzeta}\bx^\star\|_2)$ satisfy the constraint $\|\bA_{\breve{\bz}}\bu-\be_1\|_2\le\varepsilon$. Since \[\|\bA_{\bzeta}\bx^\star\|_2\ge\frac{1}{2}\|\bA_{\bzeta}\bx\|_2 \ge \frac{c_5}{2}\sqrt{\zeta_0\log(e/\zeta_0)}\,,\] all points in $\Sigma^n_s\cap \mathbb{B}_2^n(\bx^\star;\frac{2c_3\delta_1}{5}\sqrt{\zeta_0\log(e/\zeta_0)})$ satisfy the constraint of (\ref{eq:nbp}).

\textbf{Completing the proof:} 
To conclude the proof, it remains to show $\|\hat{\bx}-\bx^\star\|_2\ge \frac{2c_3\delta_1}{5}\sqrt{\zeta_0\log(e/\zeta_0)}$. To do so, we proceed under the assumption $$\|\hat{\bx}-\bx^\star\|_2\le \frac{2c_3\delta_1}{5}\sqrt{\zeta_0\log({e}/{\zeta_0})},$$ and we seek to show that equality must hold (i.e., $\|\hat{\bx}-\bx^\star\|_2= \frac{2c_3\delta_1}{5}\sqrt{\zeta_0\log(e/\zeta_0)}$).

We first show that $\hat{\bx}\in \Sigma^n_s$. In fact, if $\hat{\bx}\notin\Sigma^n_s$, we construct $\hat{\bx}'$ from $\hat{\bx}$ by setting all entries not in $\supp(\bx^\star)$ to zero; this gives $\|\hat{\bx}'\|_1<\|\hat{\bx}\|_1$, since at least one nonzero entry becomes zero. Moreover, by the construction of $\hat{\bx}'$, $$\|\hat{\bx}'-\bx^\star\|_2\le \|\hat{\bx}-\bx^\star\|_2\le\frac{2c_3\delta_1}{5}\sqrt{\zeta_0\log(e/\zeta_0)}\,.$$ Thus,
\[\hat{\bx}'\in \Sigma^n_s\cap \mathbb{B}_2^n\Big(\bx^\star;\frac{2c_3\delta_1}{5}\sqrt{\zeta_0\log(e/\zeta_0)}\Big),\]
and hence $\hat{\bx}'$ satisfies the constraint of (\ref{eq:nbp}). This contradicts the optimality of $\hat{\bx}$ to (\ref{eq:nbp}). Thus, we obtain \[\hat{\bx}\in\Sigma^n_s\cap \mathbb{B}_2^n\Big(\bx^\star;\frac{2c_3\delta_1}{5}\sqrt{\zeta_0\log(e/\zeta_0)}\Big).\]
Because $\Sigma^n_s\cap \mathbb{B}_2^n(\bx^\star;\frac{2c_3\delta_1}{5}\sqrt{\zeta_0\log(e/\zeta_0)})$ is a subset of the feasible domain of (\ref{eq:nbp}), 
\begin{align}\label{hatxc}
    \hat{\bx} = \mathrm{arg}\min~\|\bu\|_1,\quad\text{subject to }\bu \in \Sigma^n_s\cap \mathbb{B}_2^n\Big(\bx^\star;\frac{2c_3\delta_1}{5}\sqrt{\zeta_0\log({e}/{\zeta_0})}\Big). 
\end{align}
Under small enough $\zeta_0$ and $\delta_1\in(0,1)$, we use $\|\bx^\star\|_2\ge\frac{1}{2}$ to obtain  \[\frac{2c_3\delta_1\sqrt{\zeta_0\log(e/\zeta_0)}}{5}<\frac{1}{2}\le\|\bx^\star\|_2.\] In view of (\ref{hatxc}), it is not hard to observe that $\hat{\bx}_c$ must live in the boundary of $\mathbb{B}_2^n(\bx^\star;\frac{2c_3\delta_1}{5}\sqrt{\zeta_0\log(e/\zeta_0)})$. Hence, with the promised probability, we have \[\|\hat{\bx}-\bx^\star\|_2
=\frac{2c_3\delta_1}{5}\sqrt{\zeta_0\log(e/\zeta_0)}\,.\] The result   follows.     
\end{proof}

\subsection{Proof of Lemma \ref{lem:improved_lem9}}\label{Jxbound}
We have the following  refined statement. 
\begin{lem}\label{lem:refine}
   Suppose the entries of $\bPhi$ are drawn i.i.d.~from  $\calN(0,1)+\calN(0,1)\bi$. Given some small enough $\eta \in [\frac{C_1}{m},1]$  and some  $\calK\subset \mathbb{S}^{n-1}$, we   let $r=\frac{c_1\eta}{(\log (
    \eta^{-1}))^{1/2}}$ with sufficiently small $c_1$. If $$ 
         m \geq C_2 \left(\frac{\scrH(\calK,r)}{\eta}+ \frac{\omega^2(\calK_{(r)})}{\eta^3}\right)$$holds for sufficiently large $C_2$, 
    then with probability at least $1-3\exp(-c_3\eta m)$,  we have \[ \sup_{\bx\in\calK}~|\calJ_{\bx,\eta}|\leq   \eta m.\] 
\end{lem}
Before proving this, we   note that it immediately leads to Lemma \ref{lem:improved_lem9}: by   $\omega^2(\calK_{(r)})\le \omega^2(\calK-\calK)\le 4\omega^2(\calK)$  \cite[Sec. 7.5.1]{vershynin2018high} and $\scrH(\calK,r)\le C_1\frac{\omega^2(\calK)}{r^2}=\Theta(\frac{\log(\eta^{-1})\omega^2(\calK)}{\eta^2})$ from (\ref{sudakov}), we find that $m=\Omega(\eta^{-3}\log(\eta^{-1})\omega^2(\calK))$  in Lemma \ref{lem:improved_lem9} suffices to imply the sample complexity in Lemma \ref{lem:refine}. 
\begin{proof}[Proof of Lemma \ref{lem:refine}] 
 We use a covering approach  to bound \[\sup_{\bx\in\calK}|\calJ_{\bx,\eta}|=\sup_{\bx\in\calK}\sum_{i=1}^m \mathbbm{1}\big(|\bPhi_i^*\bx|\leq\eta\big).\]  We let $\calN_r$ be a minimal $r$-net of $\calK$ with $\log |\calN_r| = \scrH(\calK,r)$, then for any $\bx\in\calK$ we let $\bx'=\mathrm{arg}\min_{\bu\in\calN_r}\|\bu-\bx\|_2$. Here,   $\bx'$ depends on $\bx$, but we drop such dependence to avoid cumbersome notation. Note that we have $\|\bx'-\bx\|_2\le r$ and $\bx - \bx'\in \calK_{(r)}$. 
   By the triangle inequality,  
    \begin{align*}
       &\sum_{i=1}^m \mathbbm{1}\big(|\bPhi_i^*\bx|\leq\eta\big)\leq \sum_{i=1}^m \mathbbm{1}\big(|\bPhi_i^*\bx'|-|\bPhi_i^*(\bx-\bx')|\leq\eta\big) \\&\leq \sum_{i=1}^m \mathbbm{1}\big(|\bPhi_i^*\bx'|\leq 1.1\eta\big) + \sum_{i=1}^m\mathbbm{1}\big(|\bPhi_i^*(\bx-\bx')|>0.1\eta\big),
    \end{align*}
    which implies 
    \begin{align}\label{eq:main_bound}
        \sup_{\bx\in\calK} \sum_{i=1}^m \mathbbm{1}\big(|\bPhi_i^*\bx|\leq\eta\big)\leq \sup_{\bx\in\calN_r}\sum_{i=1}^m \mathbbm{1}\big(|\bPhi_i^*\bx|\leq 1.1\eta \big)+\sup_{\bu\in\calK_{(r)}}\sum_{i=1}^m \mathbbm{1}\big(|\bPhi_i^* \bu|>0.1\eta \big).
    \end{align}

    We first bound \[\sup_{\bx\in\calN_r}\sum_{i=1}^m \mathbbm{1}\big(|\bPhi_i^*\bx|\leq 1.1\eta \big)=\sup_{\bx\in\calN_r}|\calJ_{\bx,1.1\eta}|.\] For fixed $\bx\in \mathbb{S}^{n-1}$,  $$\mathbbm{P}(|\bPhi_i^*\bx|\le 1.1\eta)\le \mathbbm{P}(|\calN(0,1)|\le 1.1\eta)\le 1.1\sqrt{\frac{2}{\pi}}\eta \le 0.9\eta,$$ and hence Chernoff bound gives $|\calJ_{\bx,1.1\eta}|\le 0.95\eta m$ with probability at least $1-\exp(-c_1\eta m)$, where $c_1$ is some absolute constant. Therefore, when $m\ge \frac{C_2\scrH(\calK,r)}{\eta}$ with large enough $C_2$, we can take a union bound and obtain 
$\sup_{\bx\in\calN_r}|\calJ_{\bx,1.1\eta}| \le 0.95\eta m$
  with probability at least $1-\exp(-\frac{c_1\eta m}{2})$.

  All that remains is to show 
\begin{align}
    \sup_{\bu\in\calK_{(r)}}\sum_{i=1}^m \mathbbm{1}\big(|\bPhi_i^*\bu|>0.1\eta \big)\le 0.05\eta m. \label{hihi}
\end{align}
   For notational convenience, suppose that $ 0.05\eta m $ is a positive integer (we can round otherwise). We observe that a sufficient condition for (\ref{hihi}) is
    \begin{align}\label{hihihi}
        \sup_{\bu\in\calK_{(r)}}\max_{\substack{I\subset[m]\\|I|=0.05\eta m}}\Big(\frac{1}{0.05\eta m}\sum_{i\in I}|\bPhi_i^*\bu|^2\Big)^{1/2}\le 0.05\eta .
    \end{align}
    Thus, it suffices to show (\ref{hihihi}). By Lemma \ref{lem:max_k_sum}, with probability at least $1-2\exp(-c_2\eta m \log (\eta^{-1}))$, it suffices to ensure     
\begin{align*}
    \frac{\omega(\calK_{(r)})}{\sqrt{\eta m}}+ r\sqrt{\log(\eta^{-1})} \le c_3\eta 
\end{align*}
for some small enough absolute constant $c_3$. Hence, it is sufficient to have 
$  m\geq C_4\frac{\omega^2(\calK_{(r)})}{\eta^3}$ and $ r=\frac{c_5\eta}{(\log(\eta^{-1}))^{1/2}}$, where $C_4$ is sufficient large and $c_5$ is small enough. 
These assumptions are made in our statement, and hence the proof is complete. 
 \end{proof} 
\end{appendix}

\end{document}